%% file: ALL_Restat.tex
\documentclass[12pt]{article}
\usepackage[T1]{fontenc}
\usepackage[english]{babel}
\usepackage{graphicx}
\usepackage{adjustbox}
\usepackage{amsmath}
\usepackage{amsthm}
\usepackage{minitoc}
\usepackage{caption}
\usepackage{subcaption}
\usepackage{thmtools}
\usepackage{amsfonts}
\usepackage{bbm}
\usepackage{setspace}
\AtBeginDocument{%
  \addtolength\abovedisplayskip{-0.3\baselineskip}%
  \addtolength\belowdisplayskip{-0.3\baselineskip}%
}
\usepackage{supertabular, longtable, caption, afterpage, tabularx, threeparttable}
\usepackage{dcolumn}
\newcolumntype{d}[1]{D..{#1}} 
\usepackage[resetlabels]{multibib} 
\usepackage[authoryear]{natbib}

\newcites{OA}{References (Online Appendix)}
\usepackage[colorlinks=true,urlcolor= blue,citecolor=blue]{hyperref}
\usepackage{booktabs}
\usepackage{enumitem}
\usepackage{comment}
\newtheorem{proposition}{Proposition}
\newtheorem{lemma}{Lemma}
\newtheorem{assumption}{Assumption}
\declaretheoremstyle[
  bodyfont=\normalfont\itshape,
  headformat=\NAME\NUMBER  
]{nospacetheorem}
\declaretheorem[style=nospacetheorem,name=CD]{condition}
\renewcommand{\conditionautorefname}{CD\kern-4pt}

\newtheorem{theorem}{Theorem}
\newtheorem{example}{Example}
\newtheorem{corollary}{Corollary}
\newtheorem{definition}{Definition}

\newtheorem{algorithm}{Algorithm}
\usepackage[top=1in, bottom=1in, left=1in, right=1in]{geometry}
\everymath{\displaystyle}

\newtheorem{innercustomthm}{Lemma}
\newenvironment{customthm}[1]
  {\renewcommand\theinnercustomthm{#1}\innercustomthm}
  {\endinnercustomthm}

\newcommand{\bmu}{\ensuremath{\boldsymbol{\mu}}}
\newcommand{\btheta}{\ensuremath{\boldsymbol{\theta}}}
\newcommand{\bbeta}{\ensuremath{\boldsymbol{\beta}}}

\newcommand{\blambda}{\ensuremath{\boldsymbol{\Lambda}}}
\newcommand{\bgamma}{\ensuremath{\boldsymbol{\gamma}}}

\newcommand{\by}{\ensuremath{\boldsymbol{y}}}
\newcommand{\bX}{\ensuremath{\boldsymbol{X}}}

\DeclareMathOperator{\Exp}{\mathbb{E}}

\DeclareMathOperator{\Var}{\mathbb{V}}
\DeclareMathOperator{\plim}{plim}

\newcommand{\bG}{\ensuremath{\boldsymbol{G}}}

\newcommand{\bA}{\ensuremath{\boldsymbol{A}}}
\newcommand\scaleddot{\scalebox{.95}{.}}

\renewcommand \thepart{}
\renewcommand \partname{}

\makeatletter
\renewcommand{\dddot}[1]{%
  {\mathop{\kern\z@#1}\limits^{\makebox[0pt][c]{\vbox to-1.5\ex@{\kern-\tw@\ex@
   \hbox{\normalfont\scaleddot\kern-0.5pt\scaleddot\kern-0.5pt\scaleddot}\vss}}}}}
\renewcommand{\ddddot}[1]{%
  {\mathop{\kern\z@#1}\limits^{\makebox[0pt][c]{\vbox to-1.5\ex@{\kern-\tw@\ex@
   \hbox{\normalfont\scaleddot\kern-0.5pt\scaleddot\kern-0.5pt\scaleddot\kern-0.5pt\scaleddot}\vss}}}}}
\makeatother
\usepackage{color, xcolor}
\definecolor{colari}{rgb}{0.7, 0, 0.7}
\newcommand{\aristide}[1]{\textcolor{colari}{#1}}
\begin{document}
\doparttoc
\faketableofcontents 

\thispagestyle{empty}
\begin{center}
\Large Estimating Peer Effects Using Partial Network Data

\vspace{1cm}%

\large Vincent Boucher$^{\dagger}$ and Aristide Houndetoungan$^{\ddagger}$

\normalsize September 2025

\vspace{1cm}

Abstract
\end{center}

\noindent We study the estimation of peer effects through social networks when researchers do not observe the entire network structure. Special cases include sampled networks, censored networks, and misclassified links. We assume that researchers can obtain a consistent estimator of the distribution of the network. We show that this assumption is sufficient for estimating peer effects using a linear-in-means model. We provide an empirical application to the study of peer effects on students’ academic achievement using the widely used Add Health database, and show that network data errors have a large downward bias on estimated peer effects.

\vspace{1cm}

\noindent \textit{JEL Codes: C31, C36, C51}

\noindent Keywords: Social networks, Peer effects, Missing variables, Measurement errors

\noindent $^{\dagger}${\footnotesize Boucher: Department of Economics, Universit\'{e} Laval, CRREP, CREATE, CIRANO; email: \url{vincent.boucher@ecn.ulaval.ca}. \\
$^{\ddagger}$ Department of Economics, Universit\'{e} Laval; email: \url{ahoundetoungan@ecn.ulaval.ca}. \\
An R package, including all replication codes, is available at: \url{https://github.com/ahoundetoungan/PartialNetwork}.}\newpage 

\noindent\textbf{Acknowledgements:} We would like to thank Yann Bramoullé, and Bernard Fortin for their helpful comments and insights, as always. We would also like to thank Isaiah Andrews, Eric Auerbach, Arnaud Dufays, Stephen Gordon, Chih-Sheng Hsieh, Arthur Lewbel, Tyler McCormick, Angelo Mele, Francesca Molinari, Onur \"Ozg\"ur, Eleonora Patacchini, Xun Tang, and Yves Zenou for helpful comments and discussions. Thank you also to the participants of the many seminars at which we presented this research. This research uses data from Add Health, a program directed by Kathleen Mullan Harris and designed by J. Richard Udry, Peter S. Bearman, and Kathleen Mullan Harris at the University of North Carolina at Chapel Hill, and funded by Grant P01-HD31921 from the Eunice Kennedy Shriver National Institute of Child Health and Human Development, with cooperative funding from 23 other federal agencies and foundations. Special acknowledgment is given to Ronald R. Rindfuss and Barbara Entwisle for assistance in the original design. Information on how to obtain Add Health data files is available on the Add Health website (\url{http://www.cpc.unc.edu/addhealth}). No direct support was received from Grant P01-HD31921 for this research.
\newpage
\clearpage

\setcounter{page}{1}

\input{intro}
\input{model}
\input{iv}

\input{bayesian}

\input{addhealth}

\input{discussion}

\bibliographystyle{ecta}
\bibliography{biblio}
\clearpage
\appendix

\pagenumbering{arabic}
\renewcommand{\thepage}{A\arabic{page}}

\setcounter{table}{0}
\numberwithin{table}{section}

\setcounter{figure}{0}
\numberwithin{figure}{section}

\input{appendix_proof}

\cleardoublepage

\input{appendix_online.tex}

\cleardoublepage
\bibliographystyleOA{ecta}
\bibliographyOA{biblio}
\end{document}

%% file: intro.tex
\section{Introduction}
There is a large and growing literature on the impact of peer effects in social networks.\footnote{For recent reviews, see \cite{bramoulle2019peer}, and \cite{de2015econometrics}.} However, since eliciting network data is expensive \citep{breza2017using}, relatively few data sets contain comprehensive network information, and existing ones are prone to data errors. Despite some recent contributions, existing methodologies for the estimation of peer effects with incomplete or erroneous network data either focus on a specific kind of sampling or error, or they are highly computationally demanding.

In this paper, we propose a unifying framework that allows for the estimation of peer effects under the widely used linear-in-means model (e.g. \cite{manski1993identification,bramoulle2009identification}) when the researcher does not observe the entire network structure (but still observes \emph{some} information about the network). Our methodology is computationally attractive and sufficiently flexible to cover cases where, for example, network data are sampled \citep{chandrasekhar2011econometrics,liu2013estimation,lewbel2019social}, censored \citep{griffith2019namefriends}, or misclassified \citep{hardy2019estimating}. Our central assumption is that the researcher is able to estimate a network formation model using some partial information about the network structure. Leveraging recent contributions on the estimation of network formation models and focusing primarily on models with conditional link independence, we show that this assumption is sufficient to identify and estimate peer effects.

Assuming that the network is exogenous, we propose two estimators. First, we present a computationally attractive estimator based on a simulated generalized method of moments (SGMM). The moments are built using draws from the (estimated) network formation model. We study the finite sample properties of our SGMM estimator via Monte Carlo simulations. We show that the estimator performs very well, even when a large fraction of the links are missing or misclassified. Second, we present a flexible likelihood-based (Bayesian) estimator allowing us to exploit the entire structure of the data-generating process. The Bayesian approach is flexible as it allows to cover cases for which the asymptotic framework of our SGMM fails. Although the computational cost is higher than that of the SGMM, we exploit recent computational advances in the literature, e.g. \cite{mele2017,hsieh2019structural}, and show that the estimator can be successfully implemented on common-sized data sets. In particular, we apply our estimator to study peer effects on academic achievement using the widely used Add Health database. We find that data errors have a large downward bias on the estimated endogenous effect.

Our SGMM estimator is built as a bias-corrected version of the instrumental strategy proposed by \cite{bramoulle2009identification}. Using a network formation model, we obtain a consistent estimator of the distribution of the true network.\footnote{This generally requires some form of conditionally random sampling of the true network.} We then use this estimated distribution to obtain different draws from the distribution of the network. We show that our moment conditions are asymptotically valid and that the estimator is consistent and asymptotically normal, even with a finite number of draws from the estimated distribution of the network. This property significantly reduces the computational cost of the method compared to methods that rely on integrating the moment conditions \citep[e.g.,][]{chandrasekhar2011econometrics}.

Importantly, our SGMM strategy requires only the (partial) observation of a \emph{single} cross-section, unlike, for example, the approach of \cite{zhang2020spillovers}. The presence of this feature is because of two main properties of the model. First, we can consistently estimate the distribution of the mismeasured variable (i.e., the network) using a single (partial) observation of the variable. Second, in the absence of measurement error, valid instruments for the endogenous peer variable are available \citep{bramoulle2009identification}.


Our Bayesian estimator is based on the likelihood function and therefore uses more information about the structure of the model, leading to more precise estimates. In the context of this estimator, the estimated distribution for the network acts as a prior distribution, and the inferred network structure is updated through a Markov chain Monte Carlo (MCMC) algorithm. Our approach relies on data augmentation \citep{tanner1987calculation}, which treats the network as an additional set of parameters to be estimated. In particular, our MCMC builds on recent developments from the empirical literature on network formation \citep[e.g.,][]{mele2017,hsieh2019structural}. We show that the computational cost of our estimator is reasonable and that it can easily be applied to standard data sets.

We study the impact of errors in adolescents' friendship network data for the estimation of peer effects in education \citep{calvo2009peer}. We show that the widely used Add Health database features many missing links---around 45\% of the within-school friendship nominations are coded with error---and that these data errors strongly bias the estimated peer effects. Specifically, we estimate a model of peer effects on students' academic achievement. We probabilistically reconstruct the missing links, accounting for the potential censoring, and we obtain a consistent estimator of peer effects using both our estimators. The bias due to data errors is qualitatively important, even assuming that the network is exogenous. Our estimated endogenous peer effect coefficient is 1.5 times larger than that obtained by assuming the data contains no errors.

This paper contributes to the recent literature on the estimation of peer effects when the network is either not entirely observed or observed with noise. In particular, our framework is valid when network data are either sampled, censored, or misclassified.\footnote{For related literature that studies the estimation of peer effects when researchers have no network data, see \cite{manresa2016estimating,de2018recovering,lewbel2019social}.} We unify these strands in the literature and provide a flexible and computationally tractable framework for estimating peer effects with incomplete or erroneous network data.

\emph{Sampled networks and censoring:} \cite{chandrasekhar2011econometrics} show that models estimated using sampled networks are generally biased. They propose an analytical correction as well as a two-step General Method of Moments (GMM) estimator. \cite{liu2013estimation} shows that when the interaction matrix is not row-normalized, instrumental variable estimators based on an out-degree distribution are valid, even with sampled networks. Finally, \cite{zhang2020spillovers} studies program evaluation in a context in which networks are sampled locally and where some links might be unobserved. Assuming that the researcher has access to two measurements of the network for each sampled unit, she presents a nonparametric estimator of the treatment and spillover effects.\footnote{Relatedly, part of the literature focuses on models that depend on network statistics (as opposed to peer effect models). See in particular \cite{hsieh2018non}, \cite{chen2013impact},  \cite{thirkettleidentification} and \cite{reeves2024model}.}

Relatedly, \cite{griffith2019namefriends} explores the impact of imposing an upper bound on the number of links when eliciting network data. He presents a bias-correction method and explores the impact of censoring using two empirical applications. He finds that censoring underestimates peer effects. \cite{griffith2023} present a characterization of the analytic bias of censoring for the reduced-form parameters in the linear-in-means and linear-in-sums models under an Expectational Equivalence assumption.

We contribute to this literature by proposing two simple and flexible estimators for the estimation of peer effects based on a linear-in-means model. Our estimators do not require many observations of the sampled network. Similar to \cite{griffith2019namefriends} and \cite{griffith2023}, we find---using the Add Health database---that sampling leads to an underestimation of peer effects, although we find that censoring has a negligible impact, in the context of peer effects, on academic achievement.

Our SGMM estimator does not suffer from the computational cost resulting from integrating the moment conditions \citep[as in][]{chandrasekhar2011econometrics} and can produce precise estimates with as few as three network simulations. While our Bayesian estimator is more computationally demanding, we exploit recent developments from the empirical literature on network formation \cite[e.g.,][]{mele2017,hsieh2019structural} and show that it is computationally tractable. Moreover, the Bayesian estimator is valid in finite samples, which allows, in particular, to cover cases not covered by the asymptotic framework on which our SGMM relies.

\emph{Misclassification:} \cite{hardy2019estimating} look at the estimation of (discrete) treatment effects when the network is observed noisily. Specifically, they assume that observed links are affected by iid errors and present an expectation maximization (EM) algorithm that allows for a consistent estimator of the treatment effect. \cite{lewbel2023social} show that when the expected number of misclassified links grows at a rate strictly lower than the number of sampled individuals $n$, the 2SLS estimator in \cite{bramoulle2009identification} is consistent.\footnote{When the growth rate is strictly smaller than $\sqrt{n}$, the inference is also valid.} \cite{lewbel2021social} develop a two-stage least squares estimator for the linear-in-sum model when some links are potentially misclassified. They propose valid instruments under some restrictions on the observed and true interactions matrices, or when researchers observe at least two samples of the same true network.

Our model allows for the misclassification of all links with positive probability, and we do not impose restrictions on the rate of misclassification. As in \cite{hardy2019estimating}, we use a network formation model to estimate the probability of false positives and false negatives. However, our two-stage strategy---estimating the network formation model and then the peer effect model---allows for greater flexibility. 
Our paper is closest to \cite{herstad2023}, who also studies a two-step estimation approach, but focuses on the observation of a single large (mismeasured) network, adapting the network formation model in \cite{graham2017econometric}. We do not impose a specific network formation process and focus on the case in which the data are partitioned into bounded groups, such as schools or small villages.

The remainder of the paper is organized as follows. In Section \ref{sec:linearinmean}, we present the econometric model as well as the main assumptions. In Section \ref{sec:iv}, we present our SGMM estimator and study its performance via Monte Carlo simulations. In Section \ref{sec:genestim}, we present our likelihood-based estimation strategy. In Section \ref{sec:addhealth}, we present our application to peer effects on academic achievement. Section \ref{sec:conclusion} concludes the paper.

%% file: model.tex
\section{The Model}\label{sec:linearinmean}
We assume that the data are partitioned into $M>1$ groups, where group $m$ contains $N_m$ individuals. A sample consists of the following:
$$
\lbrace\mathbf{y}_m,\mathbf{X}_m,\boldsymbol{\varepsilon}_m;\mathcal{A}_m,\mathbf{A}_m \rbrace_{m=1}^M.
$$
For individuals in group $m$, $\mathbf{y}_m$ is a vector of an observed outcome of interest (e.g., academic achievement), $\mathbf{X}_m$ is an observed matrix of individual characteristics (e.g., age and gender), and $\boldsymbol{\varepsilon}_m$ is a vector of \emph{unobserved} individual heterogeneity.

The matrix $\mathbf{A}_m$ is the $N_m\times N_m$ \emph{adjacency matrix} of the network between individuals in group $m$. We assume a directed network:\footnote{All of our results hold for undirected networks.} $a_{ij,m}\in \{0,1\}$, where $a_{ij,m}=1$ if $i$ is linked to $j$. We normalize $a_{ii,m}=0$ for all $i$ and let $n_{i,m}=\sum_j a_{ij,m}$ denote the number of links of $i$ within group $m$.

We assume that $\mathbf{A}_m$ is \emph{not observed} but that researchers observe $\mathcal{A}_m$ instead. Informally, the idea is that $\mathcal{A}_m$ contains some information about the adjacency matrix $\mathbf{A}_m$. Our specific assumptions are presented in Section \ref{sec:partial}. The next assumptions formalize the above discussion.
\begin{assumption}\label{as:manymarkets} The population is partitioned into $M>1$ groups, where the size $N_m$ of each group $m=1,...,M$ is bounded. The sequence $\lbrace\mathbf{y}_m,\mathbf{X}_m,\boldsymbol{\varepsilon}_m;\mathcal{A}_m,\mathbf{A}_m \rbrace$ is independent across $m$. Moreover, $\mathbf{X}_m$ is uniformly bounded in $m$.\footnote{i.e., $\sup_{m \geq 1}\lVert \mathbf X_m \rVert_2 < \infty$, where $\lVert.\rVert_2$ is the Euclidean norm.} 
\end{assumption}
\begin{assumption}\label{as:nosampling} For each group $m$, the variables $\mathbf{y}_m$, $\mathbf{X}_m$ and $\mathcal{A}_m$ are observed. The variables $\boldsymbol{\varepsilon}_m$ and $\mathbf{A}_m$ are not.
\end{assumption}

Assumption \ref{as:manymarkets} implies a ``many markets'' asymptotic framework, meaning that the number of groups $M$ goes to infinity as the number of individuals $N$ goes to infinity. It is a standard assumption in the literature on the econometrics of games and the literature on peer effects.\footnote{See for example \cite{bramoulle2019peer}, \cite{breza2016}, and \cite{de2015econometrics}.} For example, the data could consist of a collection of small villages \citep{banerjee2013diffusion} or schools \citep{calvo2009peer}. Assumption \ref{as:nosampling} implies in particular that the data are composed of group-level censuses for $\mathbf{y}_m$ and $\mathbf{X}_m$.\footnote{Contrary to \cite{liu2017peer} or \cite{wang2013estimation}, for example.} A similar assumption is made by \cite{breza2017using}. 


\subsection{The Linear-in-Means Model}
In this section, we present the linear-in-means model \citep{manski1993identification,bramoulle2009identification}, arguably the most widely used model for studying peer effects in networks \citep[see][for a recent review]{bramoulle2019peer}.

Let $\mathbf{G}_m=f(\mathbf{A}_m)$, the $N_m\times N_m$ \emph{interaction matrix} for some function $f$. Unless otherwise stated, we assume that $\mathbf{G}_m$ is a row-normalization of the adjacency matrix $\mathbf{A}_m$.\footnote{In such a case, $g_{ij,m}=a_{ij,m}/n_{i,m}$ whenever $n_{i,m}>0$, whereas $g_{ij,m}=0$ otherwise.} Most of our results hold for any function $f$.

We focus on the following model:
\begin{equation}\label{eq:linearinmean}
\mathbf{y}_m=c\mathbf{1}_m+\mathbf{X}_m\boldsymbol{\beta}+\alpha\mathbf{G}_m\mathbf{y}_m+{\mathbf{G}_m\mathbf{X}_m}\boldsymbol{\gamma}+\boldsymbol{\varepsilon}_m,
\end{equation}
where $\mathbf{1}_m$ is a $N_m-$dimensional vector of $1$'s.
The parameter $\alpha$ therefore captures the impact of the average outcome of one's peers on their behavior (the endogenous peer effect). The parameter $\boldsymbol{\beta}$ captures the impact of one's characteristics on their behavior (the individual effects). The parameter $\boldsymbol{\gamma}$ captures the impact of the average characteristics of one's peers on their behavior (the contextual peer effects). For simplicity, we assume that the constant $c$ does not vary across $m$. However, our results can be adapted to include group-level fixed effects.

We impose the following assumptions.
\begin{assumption}\label{as:alpha} $|\alpha|<1/\Vert \mathbf{G}_m\Vert$ for some submultiplicative norm $\Vert \cdot\Vert$, and all $m=1,...,M$.
\end{assumption}
\begin{assumption}\label{as:exonet} Exogeneity: $\mathbb{E}[\boldsymbol{\varepsilon}_m|\mathbf{X}_m,\mathbf{A}_m,\mathcal{A}_m]=\mathbf{0}$ for all $m=1,...,M$.
\end{assumption}

Assumption \ref{as:alpha} ensures that the model is coherent and that there exists a unique vector $\mathbf{y}_m$ compatible with (\ref{eq:linearinmean}). When $\mathbf{G}_m$ is row-normalized, $|\alpha|<1$ is sufficient.
Finally, Assumption \ref{as:exonet} implies that individual characteristics and the network structure are exogenous. While the exogeneity of the network is a strong assumption, we consider it as a benchmark and focus on the case in which the network is not perfectly observed. We now describe the network sampling process in more detail.

\subsection{The Network Formation Model}\label{sec:partial}

In this paper, we relax the costly assumption that the adjacency matrix $\mathbf{A}_m$ is observed. We assume instead that sufficient information about the network (i.e., $\mathcal{A}_m$) is observed so that a network formation model can be estimated. The discussion below formalizes our assumptions about the relationship between $\mathcal{A}_m$ and $\mathbf{A}_m$. We start by describing the data-generating process for $\mathbf{A}_m$.


We assume that for any group $m$, $P(\mathbf{A}_m|\mathbf{X}_m)=\Pi_{ij}P(a_{ij,m}|\mathbf{X}_m)$, where
\begin{equation}\label{eq:gennetfor}
P(a_{ij,m}|\mathbf{X}_m)=\frac{\exp\{a_{ij,m}Q(\boldsymbol{\rho}_0,\mathbf{w}_{ij,m})\}}{1+\exp\{Q(\boldsymbol{\rho}_0,\mathbf{w}_{ij,m})\}}, 
\end{equation}
and where $Q$ is some known function that is twice continuously differentiable in $\boldsymbol{\rho}$, and $\mathbf{w}_{ij,m}=\mathbf{w}_{ij,m}(\mathbf{X}_m)$ is a vector of observed characteristics for the pair $ij$ in group $m$.\footnote{Throughout, $P$ refers to the probability notation. Note that by construction, links are only defined between individuals of the same group so the probability that individuals from different groups are linked is zero.}

We focus on network formation models that are conditionally independent across links: $P(\mathbf{A}_m|\mathbf{X}_m)=\Pi_{ij}P(a_{ij,m}|\mathbf{X}_m)$. This notably includes models such as the ones in \cite{graham2017econometric} and \cite{breza2017using}, but excludes many models of strategic network formation such as the ones in \cite{mele2017} and \cite{de2018identifying}. 

This is mainly done for simplicity (and clarity) of the analysis in the main text. We however note that our methodology can be adapted to more general network formation models. In Online Appendix B, we further discuss how this can be done for a few specific network formation processes, such as the one in \cite{boucher2017my}, as well as exponential random graph models (ERGM) featuring only reciprocity.

\subsection{Using Partial Information}

In this section, we discuss how partial network information can be used in order to estimate the structural parameters $\boldsymbol{\rho}$ from the network formation process (\ref{eq:gennetfor}) and simulate networks from the implied distribution. We now present our main assumption.
\begin{assumption}[Partial Network Information]\label{as:partial} Given $\lbrace \mathcal{A}_m,\mathbf{X}_m\rbrace_{m=1}^M$ and the parametric model (\ref{eq:gennetfor}), there exists an estimator $\hat{\boldsymbol{\rho}}_M$, such that $\sqrt{M}(\hat{\boldsymbol{\rho}}_M-\boldsymbol{\rho}_0)\rightarrow_d N(\boldsymbol{0},\boldsymbol{V}_{\boldsymbol{\rho}})$ as $M\rightarrow\infty$, where $\boldsymbol{V}_{\boldsymbol{\rho}}$ is a positive semidefinite matrix.
\end{assumption}


Assumption \ref{as:partial} is a high level assumption, encompassing many things, but its main substantive content is that $\mathcal{A}_m$ is sufficient to \emph{identify} $\boldsymbol{\rho}$:\footnote{Indeed, from our Assumption \ref{as:manymarkets}, groups are bounded and independent. As such, consistency and asymptotic normality generally follow from standard LLN and CLT for independent, non-identically distributed, random variables and under standard regularity conditions.} the dependence between $\mathcal{A}_m$ and $\mathbf{A}_m$ needs to be strong enough so that, using (\ref{eq:gennetfor}), the researcher can estimate the data generating process for $\mathbf{A}_m$. We present leading examples below.

Note however that our asymptotic framework (Assumption \ref{as:manymarkets}) limits the amount of unobserved individual heterogeneity included in (\ref{eq:gennetfor}). Indeed, since groups are bounded, individuals' number of links is also bounded. This implies that the model in \cite{breza2017using} cannot accommodate Assumption \ref{as:partial}.\footnote{\cite{breza2019consistently} show that consistency is only achieved as the size of the groups goes to infinity.} Researchers interested in using this model should therefore rely on our Bayesian estimator, presented in Section \ref{sec:genestim}.

In some contexts, however, information in $\mathcal{A}_m$ can be sufficient to identify individual unobserved heterogeneity, even under Assumption \ref{as:manymarkets}. In particular, in Online Appendix B, we show how the model in \cite{graham2017econometric} can be adapted to our setting.

When Assumption \ref{as:partial} holds, we can define an estimator of the distribution of the true network.

\begin{definition}\label{def:estimator}
     A consistent estimator of the distribution of the true network for some function $\kappa$ is a probability distribution $\hat{P}(\mathbf{A}_m|\hat{\boldsymbol{\rho}},\mathbf{X}_m,\kappa(\mathcal{A}_m))$ such that\\ $\displaystyle\sup_{m, \mathbf{A}_m}\Vert \hat{P}(\mathbf{A}_m|\hat{\boldsymbol{\rho}},\mathbf{X}_m,\kappa(\mathcal{A}_m))-{P}(\mathbf{A}_m|\mathbf{X}_m,\kappa(\mathcal{A}_m))\Vert\rightarrow_p 0$ as $M\rightarrow\infty$.
\end{definition}


The function $\kappa$ controls how much information in $\mathcal{A}_m$ is used in order to complement the information obtained by estimating the network formation process in Equation \eqref{eq:gennetfor}. Two important polar cases are the identity function $\kappa(\mathcal{A}_m)=\mathcal{A}_m$ implying that all the information in $\mathcal{A}$ is used, and the constant function $\kappa(\mathcal{A}_m)=\kappa_0$ for all $\mathcal{A}_m$ in which no information on $\mathcal{A}$ is used. Although our methodology is valid for any $\kappa$, the choice of $\kappa$ may strongly affect the identification and precision of our estimators.

When $\kappa$ is the identity function, the estimator is obtained from Bayes' rule (See Examples \ref{ex:sampled}--\ref{ex:misclass} below):
\begin{equation}\label{eq:Bayes}
\hat{P}(\mathbf{A}_m|\hat{\boldsymbol{\rho}},\mathbf{X}_m,\mathcal{A}_m)=\frac{P(\mathcal{A}_m|\mathbf{X}_m,\mathbf{A}_m)P(\mathbf{A}_m|\hat{\boldsymbol{\rho}},\mathbf{X}_m)}{P(\mathcal{A}_m|\mathbf{X}_m)}.
\end{equation}
However, in some contexts, such a quantity may be hard to compute, depending on the nature of the information in $\mathcal{A}_m$. A solution, therefore, could be to disregard the information in $\mathcal{A}$. However, in that case, the precision of the estimator strongly depends on the network formation process in (\ref{eq:gennetfor}). Thus, the loss in precision is context-dependent. In particular, it depends on the heterogeneity in the probability of link formation implied by (\ref{eq:gennetfor}), and on the specificity about $\mathbf{A}_m$ that is contained in $\mathcal{A}_m$.


We specifically discuss three leading examples in which Assumption \ref{as:partial} holds and focus on how $\hat{P}(\mathbf{A}_m|\hat{\boldsymbol{\rho}},\mathbf{X}_m,\kappa(\mathcal{A}_m))$ is constructed: \emph{sampled networks} (Example \ref{ex:sampled}), \emph{censored networks }(Example \ref{ex:censored}), and \emph{misclassified network links} (Example \ref{ex:misclass}). 


\begin{example}[Sampled Networks]\label{ex:sampled}
Suppose that we observe the realizations of $a_{ij}$ for a random sample of pairs of individuals \cite[e.g.,][]{chandrasekhar2011econometrics}. Here $\mathcal{A}_m$ is simply a list of sampled pairs: $\mathcal{A}_{m}=\{a_{ij,m}\}_{ij\text{ is sampled}}$ \citep[see e.g.,][for concrete example]{conley2010learning}. Consider the following simple network formation model: $$P(a_{ij,m}=1|\mathbf{X}_m)=\frac{\exp\{\mathbf{w}_{ij,m}\boldsymbol{\rho}\}}{1+\exp\{\mathbf{w}_{ij,m}\boldsymbol{\rho}\}}.$$ In this case, a simple logistic regression on the subset of sampled pairs provides a consistent estimator of $\boldsymbol{\rho}$ since pairs of individuals for which $a_{ij,m}$ is observed is random.

In this simple framework, the linking status of sampled pairs of individuals is known. As such it is natural to define $\kappa$ as the identity map, which leads to the estimator\break $\hat{P}(a_{ij,m}|\hat{\boldsymbol{\rho}},\mathbf{X}_m,\mathcal{A}_m)=a_{ij,m}$ for all sampled pairs $ij$, and
$\hat{P}(a_{ij,m}|\hat{\boldsymbol{\rho}},\mathbf{X}_m, \mathcal{A}_m)={\exp\{\mathbf{w}_{ij,m}\hat{\boldsymbol{\rho}}\}}/\break({1+\exp\{\mathbf{w}_{ij,m}\hat{\boldsymbol{\rho}}\}})$ otherwise. In essence, sampled pairs are used to estimate the network formation model, which is then used in order to predict the probability of a link for pairs that are not sampled.
\end{example}

\begin{example}[Censored Network Data]\label{ex:censored}
As discussed in \cite{griffith2019namefriends}, network data is often censored. This typically arises when surveyed individuals are asked to name only $T>1$ links (among the $N_m$ possible links they may have). Here, $\mathcal{A}_m$ can be represented by an $N_m\times N_m$ binary matrix $\mathbf{A}_m^{obs}$ which takes value $a_{ij,m}=1$ if $i$ nominated $j$, and $0$ otherwise. 
Consider the same simple model as in Example \ref{ex:sampled}:
$$P(a_{ij,m}=1|\mathbf{X}_m)=\frac{\exp\{\mathbf{w}_{ij,m}\boldsymbol{\rho}\}}{1+\exp\{\mathbf{w}_{ij,m}\boldsymbol{\rho}\}}.$$ 
In Section \ref{sec:addhealth} and the Online Appendix G.2, we present how to estimate $\boldsymbol{\rho}$ in detail. Here, we discuss how to obtain the estimator $\hat{P}(\mathbf{A}_{m}|\hat{\boldsymbol{\rho}},\mathbf{X}_m,\kappa(\mathcal{A}_m))$ given $\hat{\boldsymbol{\rho}}$ and $\kappa(\mathcal{A}_m)=\mathcal{A}_m$. Note that $\hat{P}(a_{ij,m}=1|\hat{\boldsymbol{\rho}},\mathbf{X}_m,a_{ij,m}^{obs}=1)=1$ because observed links necessarily exist. Second, note also that for any individual $i$, such that $n_{i,m}<T$, we have $\hat{P}(a_{ij,m}|\hat{\boldsymbol{\rho}},\mathbf{X}_m,a_{ij,m}^{obs})=a_{ij}^{obs}$ for all $j$, as their network data are not censored.

Thus, the structural model is only used to obtain the probability of links that are not observed for individuals whose links are potentially censored, i.e., $\hat{P}(a_{ij,m}=1|\hat{\boldsymbol{\rho}},\mathbf{X}_m,a_{ij,m}^{obs}=0)={\exp\{\mathbf{w}_{ij,m}\hat{\boldsymbol{\rho}}\}}/({1+\exp\{\mathbf{w}_{ij,m}\hat{\boldsymbol{\rho}}\}})$ for all $ij$, such that $n_i= T$.

\end{example}

\begin{example}[Misclassification]\label{ex:misclass}
\cite{hardy2019estimating} study cases in which networks are observed but may include misclassified links (i.e., false positives and false negatives). Here, $\mathcal{A}_m$ can be represented by an $N_m\times N_m$ binary matrix $\mathbf{A}_m^{mis}$. Consider the same simple model as in Example \ref{ex:sampled} and \ref{ex:censored}:
$$P_{ij,m}^1\equiv P(a_{ij,m}=1|\mathbf{X}_m)=\frac{\exp\{\mathbf{w}_{ij,m}\boldsymbol{\rho}\}}{1+\exp\{\mathbf{w}_{ij,m}\boldsymbol{\rho}\}}.$$ 
The (consistent) estimation $\boldsymbol{\rho}$ in such a context follows directly from the existing literature on misclassification in binary outcome models, e.g., \cite{hausman1998misclassification}. In this context, the simplicity of the sampling scheme allows to consider the identity map $\kappa(\mathcal{A}_m)=\mathcal{A}_m$. The estimator for the distribution of the true network can be obtained using Bayes' rule:\begin{eqnarray*}
P(a_{ij,m} = 1|a_{ij,m}^{mis} = 0, \mathbf{X}_m) &=& \frac{\rho_2P_{ij,m}^1}{\rho_2P_{ij,m}^1 + (1 - \rho_1)(1 - P_{ij,m}^1)}\\
P(a_{ij,m} = 1|a_{ij,m}^{mis} = 1, \mathbf{X}_m) &=& \frac{(1 - \rho_2)P_{ij,m}^1}{(1 - \rho_2)P_{ij,m}^1 + \rho_1(1 - P_{ij,m}^1)},
\end{eqnarray*}
where $\rho_1$ and $\rho_2$ are the missclassification probabilities. We consider this case in our Monte Carlo simulations in Section \ref{sec:montecarlo}. 
\end{example}

%% file: iv.tex
\section{Simulated Generalized Method of Moment Estimators}\label{sec:iv}

In this section, we present an estimator based on a Simulated Generalized Method of Moments (SGMM). Our SGMM is constructed as a de-biased simulated version of the widely used linear GMM in \cite{bramoulle2009identification}. 

Before presenting the estimator, we start with an informal discussion of how the moment function is built. A formal treatment is presented in Appendix \ref{sec:appendix_proofs}. Recall first the linear-in-means model presented in the previous section:
\begin{eqnarray*}
\mathbf{y}_m
&=&\mathbf{V}_m\tilde{\boldsymbol{\theta}} + \alpha\mathbf{G}_m\mathbf{y}_m + \boldsymbol{\varepsilon}_m,
\end{eqnarray*}
where we defined $\mathbf{V}_m=[\mathbf{1}_m,\mathbf{X}_m,\mathbf{G}_m\mathbf{X}_m]$, and $\tilde{\boldsymbol{\theta}} =[c,\boldsymbol{\beta}',\boldsymbol{\gamma}']'$. A valid set of instruments is: $\mathbf{Z}_m=[\mathbf{1}_m,\mathbf{X}_m,\mathbf{G}_m\mathbf{X}_m,\mathbf{G}_m^2\mathbf{X}_m,\mathbf{G}_m^3\mathbf{X}_m,...]$ \citep{bramoulle2009identification}. This defines the following moment function: $\mathbf{m}_m(\boldsymbol{\theta})=\mathbf{Z}_m'\boldsymbol{\varepsilon}_m$, where $\boldsymbol{\theta}=[\alpha,c,\boldsymbol{\beta}',\boldsymbol{\gamma}']'$, and one can easily show that $\mathbb{E}[\mathbf{m}_m(\boldsymbol{\theta})|\mathbf{A}_m,\mathbf{X}_m]=\mathbf{0}$ for $\boldsymbol{\theta}=\boldsymbol{\theta}_0$ and that $\boldsymbol{\theta}_0$ is identified under the usual rank condition.\footnote{As standard, we use the subscript $0$ to denote the true value of the parameter. See e.g., \cite{bramoulle2009identification} and \cite{lee2010specification} for identification results when $\mathbf{G}_m$ is observed.}

Unfortunately, this approach is not feasible when $\mathbf{G}_m=f(\mathbf{A}_m)$ is not observed. As discussed, our strategy is to develop a simulated version of this simple linear GMM estimator. Indeed, equipped with a consistent estimator of the distribution of $\mathbf{A}_m$ (see Definition \ref{def:estimator}), we can draw network structures from that same distribution.

To simplify the notation, we denote $\dot{\mathbf{G}}_m=f(\dot{\mathbf{A}}_m)$, $\ddot{\mathbf{G}}_m=f(\ddot{\mathbf{A}}_m)$, and $\dddot{\mathbf{G}}_m=f(\dddot{\mathbf{A}}_m)$, where $\dot{\mathbf{A}}_m$, $\ddot{\mathbf{A}}_m$, and $\dddot{\mathbf{A}}_m$ are independent draws from the distribution $\hat{P}(\mathbf{A}_m|\hat{\boldsymbol{\rho}},\mathbf{X}_m,\kappa(\mathcal{A}_m))$. In particular, note that: $\dot{\mathbf{G}}_m=\dot{\mathbf{G}}_m(\hat{\boldsymbol{\rho}})=f(\{\dot{a}_{m,ij}\}_{ij})=f(\{\mathbbm{1}[\hat{P}(\dot{a}_{m,ij} = 1|\hat{\boldsymbol{\rho}};\mathbf{X}_m, \kappa(\mathcal{A}_m))\geq \dot{u}_{m,ij}]\}_{ij})$, where $\dot{u}_{m,ij}\sim_{iid}U[0,1]$ and independent of $\boldsymbol \varepsilon_m$ (and similarly for $\ddot{\mathbf{G}}_m$ and $\dddot{\mathbf{G}}_m$), and $\mathbbm{1}$ is the indicator function. We will also note $\dot{\mathbf{Z}}_m$ and $\dot{\mathbf{V}}_m$, the versions of ${\mathbf{Z}}_m$ and ${\mathbf{V}}_m$ in which $\mathbf{G}_m$ is replaced with $\dot{\mathbf{G}}_m$ (and similarly for $\ddot{\mathbf{G}}_m$ and $\dddot{\mathbf{G}}_m$).

Now, suppose that we replace the unobserved $\mathbf{G}_m$ with $\dot{\mathbf{G}}_m$ everywhere in the expression $\mathbf{Z}_m'\boldsymbol{\varepsilon}_m$. This would lead to a moment function with a conditional expectation given by:
\begin{eqnarray*}
\mathbb{E}(\dot{\mathbf{m}}_m(\boldsymbol{\theta})|\kappa(\mathcal{A}_m),\mathbf{X}_m)&=&\mathbb{E}(\dot{\mathbf{Z}}'_m[(\mathbf I_m - \alpha \dot{\mathbf G}_m )\mathbf{y}_m - \dot{\mathbf{V}}\tilde{\boldsymbol{\theta}}]| \kappa(\mathcal{A}_m),\mathbf{X}_m)\\
&=&\mathbb{E}(\dot{\mathbf{Z}}'_m\dot{\boldsymbol{\varepsilon}}_m|\kappa(\mathcal{A}_m),\mathbf{X}_m),
\end{eqnarray*}
where $\mathbf I_m$ is the identity matrix of dimension $N_m$ and $\dot{\boldsymbol{\varepsilon}}_m = (\mathbf I_m - \alpha \dot{\mathbf G}_m )\mathbf{y}_m - \dot{\mathbf{V}}\tilde{\boldsymbol{\theta}}$. The expectation of the moment function does not generally equal $\mathbf{0}$ when $\boldsymbol{\theta}=\boldsymbol{\theta}_0$, even asymptotically.\footnote{Recall from Definition \ref{def:estimator} that $\dot{\mathbf{G}}_m$ is drawn from the same distribution as $\mathbf{G}_m$ only as $M\rightarrow\infty$.}

There are two issues with the previous moment function. First, the instruments and the explanatory variables are generated using the same network draw $\dot{\mathbf{G}}_m$, which introduces a correlation between the $\dot{\mathbf{Z}}_m$ and $\dot{\boldsymbol{\varepsilon}}_m$ (which includes the approximation error), conditionally on $\kappa(\mathcal{A}_m)$, and $\mathbf X_m$, even when $\boldsymbol{\theta}=\boldsymbol{\theta}_0$. This can easily be resolved by simply using different draws to construct the instruments and the explanatory variables.\footnote{Using different draws for the instruments and the explanatory variables is not necessary for the validity of our estimator, but one can show that it decreases the variance of the moment condition and is thus more efficient.} This leads to: $\mathbb{E}\left(\ddot{\mathbf{m}}_m(\boldsymbol{\theta})|\kappa(\mathcal{A}_m),\mathbf{X}_m\right)=\mathbb{E}(\dot{\mathbf{Z}}'_m\ddot{\boldsymbol{\varepsilon}}_m|\kappa(\mathcal{A}_m),\mathbf{X}_m)$, where $\ddot{\boldsymbol{\varepsilon}}_m=(\mathbf I_m - \alpha \ddot{\mathbf G}_m )\mathbf{y}_m - \ddot{\mathbf{V}}_m\tilde{\boldsymbol{\theta}}$. 

However, in general, $\mathbb{E}(\ddot{\boldsymbol{\varepsilon}}_m|\kappa(\mathcal{A}_m),\mathbf{X}_m)\neq\mathbf{0}$ at $\boldsymbol{\theta}=\boldsymbol{\theta}_0$ since $\mathbf{y}_m$ is a function of the true network structure $\mathbf{G}_m$. Indeed, replacing $\mathbf{y}_m$, we can rewrite:
$$
\ddot{\boldsymbol{\varepsilon}}_m=(\mathbf I_m - \alpha \ddot{\mathbf G}_m )(\mathbf I_m - \alpha_0 \mathbf G_m )^{-1}[\mathbf{V}_m\tilde{\boldsymbol{\theta}}_0+\boldsymbol{\varepsilon}_m]-\ddot{\mathbf{V}}_m\tilde{\boldsymbol{\theta}}.
$$
While we can show that $\mathbb{E}[(\mathbf I_m - \alpha \ddot{\mathbf G}_m )(\mathbf I_m - \alpha_0 \mathbf G_m )^{-1}{\boldsymbol{\varepsilon}}_m|\kappa(\mathcal{A}_m),\mathbf{X}_m]=\mathbf{0}$ from the law of iterated expectations and Assumption \ref{as:exonet}, we have:
\begin{eqnarray*}
\mathbb{E}[(\mathbf I_m - \alpha \ddot{\mathbf G}_m )(\mathbf I_m - \alpha_0 \mathbf G_m )^{-1}\mathbf{V}_m\tilde{\boldsymbol{\theta}}_0 |\kappa(\mathcal{A}_m),\mathbf{X}_m]-\mathbb{E}[\ddot{\mathbf{V}}_m\tilde{\boldsymbol{\theta}}|\kappa(\mathcal{A}_m),\mathbf{X}_m]\neq 0,
\end{eqnarray*}
when $\boldsymbol{\theta}=\boldsymbol{\theta}_0$, even asymptotically. This is due to the approximation error in using $\ddot{\mathbf{G}}_m$ instead of $\mathbf{G}_m$. This approximation error does not vanish asymptotically. In particular, because groups are bounded, the product $(\mathbf I_m - \alpha_0 \ddot{\mathbf G}_m )(\mathbf I_m - \alpha_0 \mathbf G_m )^{-1}$ does not converge to the identity matrix. If it did, consistency would follow.\footnote{In Proposition \ref{prop:bias_nocontext} in Online Appendix D, we derive the asymptotic bias on $\boldsymbol{\theta}$ assuming that ${\mathbf{G}_m\mathbf{X}_m}$ is observed.}

Our SGMM presented below offers a bias-corrected version of this estimator. Specifically, consider the following (feasible) approximation of the bias of $\mathbb{E}(\ddot{\boldsymbol{\varepsilon}}_m|\kappa(\mathcal{A}_m),\mathbf{X}_m)$:
$$
\boldsymbol{\delta}_m=(\mathbf I_m - \alpha \ddot{\mathbf G}_m )(\mathbf I_m - \alpha \dddot{\mathbf G}_m )^{-1}\dddot{\mathbf{V}}_m\tilde{\boldsymbol{\theta}}-\ddot{\mathbf{V}}_m\tilde{\boldsymbol{\theta}}.
$$
We obtain $\ddot{\boldsymbol{\varepsilon}}_m-\boldsymbol{\delta}_m=(\mathbf I_m - \alpha \ddot{\mathbf G}_m )\mathbf{y}_m-(\mathbf I_m - \alpha \ddot{\mathbf G}_m )(\mathbf I_m - \alpha \dddot{\mathbf G}_m )^{-1}\dddot{\mathbf{V}}_m\tilde{\boldsymbol{\theta}}$, and we can show that $\mathbb{E}(\ddot{\boldsymbol{\varepsilon}}_m-\boldsymbol{\delta}_m|\kappa(\mathcal{A}_m),\mathbf{X}_m)=\mathbf{0}$ for $\boldsymbol{\theta}=\boldsymbol{\theta}_0$ as $M\rightarrow\infty$.\footnote{In the bias approximation expression, we use a third independent draw, $\dddot{\mathbf{G}}_m$, as an approximation of $\mathbf{G}_m$ to ensure that it remains independent of $\dot{\mathbf{G}}_m$ and $\ddot{\mathbf{G}}_m$, just as $\mathbf{G}_m$ is.}

The above discussion thus leads to the definition of our SGMM. Let $\lbrace \dot{\mathbf{G}}_m^{(r)}\rbrace_{r=1}^R$, $\lbrace \ddot{\mathbf{G}}^{(s)}_m\rbrace_{s=1}^S$, and $\lbrace \dddot{\mathbf{G}}^{(t)}_m\rbrace_{t=1}^T$ be sequences of independent draws from estimator of the network formation process (see Definition \ref{def:estimator}). Let also $\dot{\mathbf{Z}}^{(r)}_m=[\mathbf{1}_m,\mathbf{X}_m,\dot{\mathbf{G}}^{(r)}_m\mathbf{X}_m,(\dot{\mathbf{G}}^{(r)}_m)^2\mathbf{X}_m,(\dot{\mathbf{G}}_m^{(r)})^3\mathbf{X}_m,...]$ be matrices of simulated instruments and $\dddot{\mathbf{V}}_m^{(t)}=[\mathbf{1}_m,\mathbf{X}_m,\dddot{\mathbf{G}}_m^{(t)}\mathbf{X}_m]$ be matrices of simulated explanatory variables. Finally, define the (simulated) empirical moment function as follows:
\begin{equation}\label{eq:moment}
\bar{\mathbf{m}}_{M}(\boldsymbol{\theta})=\frac{1}{M}\sum_m\frac{1}{RST}\sum_{rst}\dot{\mathbf{Z}}^{(r)\prime}_m\left[(\mathbf{I}_m-\alpha\ddot{\mathbf{G}}^{(s)}_m)\left(\mathbf{y}_m-(\mathbf{I}_m-\alpha\dddot{\mathbf{G}}_m^{(t)})^{-1}\dddot{\mathbf{V}}_m^{(t)}\tilde{\boldsymbol{\theta}}\right)\right],
\end{equation}
which is the empirical version of the moment $\mathbb{E}(\dot{\mathbf{Z}}_m^\prime(\ddot{\boldsymbol{\varepsilon}}_m-\boldsymbol{\delta}_m)|\kappa(\mathcal{A}_m),\mathbf{X}_m)$ across multiple draws and groups. The next result follows.

\begin{theorem}[SGMM]\label{prop:non_observed} Suppose that Assumptions \ref{as:manymarkets}--\ref{as:partial} and regularity conditions \ref{as:regularity}-- \ref{as:reg_variance} hold (see Appendix \ref{sec:appendix_proofs}). Suppose also that for any $\boldsymbol{\theta}\ne \boldsymbol{\theta}_0$, $\lim_{M\rightarrow\infty} \mathbb{E}(\bar{\mathbf{m}}_M(\boldsymbol{\theta}, \boldsymbol{\rho}_0)) \ne \mathbf 0$. Then, for any positive integers $R$, $S$, and $T$, the (simulated) GMM estimator based on (\ref{eq:moment}) is consistent and asymptotically normally distributed.
\end{theorem}
The identification condition is standard and ensures that the moment condition is not solved at $\boldsymbol{\theta}\neq\boldsymbol{\theta}_0$.\footnote{In Lemma \ref{lemma:momentvalid} in Appendix \ref{sec:appendix_proofs}, with show that $\boldsymbol{\theta}_0$ solves the moment condition.} We discuss it in more detail below.

Theorem \ref{prop:non_observed} presents conditions for the consistency and asymptotic normality of our two-step estimator. In particular, similar to a standard simulated GMM \citep{gourieroux1996simulation}, consistency holds for a finite number of simulations.

Here, a few remarks regarding the consistency and asymptotic normality are in order. Note that the simulated moment function is based on network draws that depend on an \emph{estimated} distribution. This has two implications.

First, it implies that our SGMM estimator is a two-stage estimator and therefore that the asymptotic variance-covariance matrix for $\boldsymbol{\theta}$ has to account for the first-stage sampling uncertainty. We show how to estimate the resulting asymptotic variance-covariance in the Online Appendix C.2.

Second, the fact that the simulated variables are binary implies that the objective function of the two-stage estimator is not everywhere continuous in ${\boldsymbol{\rho}}$. While this has a limited impact on consistency, it does complicate the proof of the asymptotic normality. Our proof builds on the argument in \cite{andrews1994empirical}, and we show that the stochastic equicontinuity condition holds using a \emph{bracketing} argument.

Consistency and asymptotic normality also obviously depend on an identification condition. Here, the fact that the approximation of the bias $\boldsymbol{\delta}_m$ is non-linear in $\alpha$ implies that our SGMM is non-linear and that the identification condition cannot be simplified to a simple rank condition.

In the Online Appendix C.1, we show that the objective function of our SGMM can be concentrated around $\alpha$ and that \emph{conditional on $\alpha$}, the identification condition for $\tilde{\boldsymbol{\theta}}$ reduces to an asymptotic rank condition. Specifically, a sufficient condition for the identification of $\tilde{\boldsymbol{\theta}}$, conditional on $\alpha$, is that the expected value of:
$$
\frac{1}{M}\sum_m\frac{1}{RST}\sum_{rst}\dot{\mathbf{Z}}_m^{(r)\prime}(\mathbf{I}_m-\alpha\ddot{\mathbf{G}}_m^{(s)})(\mathbf{I}_m-\alpha\dddot{\mathbf{G}}_m^{(t)})^{-1}\dddot{\mathbf{V}}_m^{(t)}
$$
converges in probability to a full rank matrix for all $\alpha$.\footnote{In particular, identification of $\tilde{\boldsymbol{\theta}}$ conditional on $\alpha$ requires that the number of instruments (the number of column in $\dot{\mathbf{Z}}^{r}_m$) is at least as large as the number of explanatory variables (the number of column in $\dddot{\mathbf{V}}^{t}_m$). Thus, identification of $\boldsymbol{\theta}$ requires at least one instrument more than the number of explanatory variables. This condition is, however, not sufficient. The identification of $\alpha$ requires that the concentrated objective function is uniquely minimized, which cannot be reduced to a simple rank condition. See the Online Appendix C.1.} The last expression makes it clear that the non-linearity in our SGMM is sourced in the approximation of the asymptotic bias $\boldsymbol{\delta}_m=(\mathbf I_m - \alpha \ddot{\mathbf G}_m )(\mathbf I_m - \alpha \dddot{\mathbf G}_m )^{-1}\dddot{\mathbf{V}}_m\tilde{\boldsymbol{\theta}}-\ddot{\mathbf{V}}_m\tilde{\boldsymbol{\theta}}$.

When the matrix $\mathbf{G}_m$ is observed, we have $\dot{\mathbf{G}}_m^{(r)}=\ddot{\mathbf{G}}_m^{(s)}=\dddot{\mathbf{G}}_m^{(t)}={\mathbf{G}}_m$ and the expression reduces to $\mathbf{Z}'_m\mathbf{V}_m$, \citep{bramoulle2009identification} which does not depend on $\alpha$. This shows that the quality of $\hat{P}(\mathbf{A}_m|\hat{\boldsymbol{\rho}},\mathbf{X}_m,\kappa(\mathcal{A}_m))$ has strong implications for identification since the dependence on $\alpha$ is weaker when the correlation between network draws is strong. We study the finite sample properties of our estimator using Monte Carlo simulation in Section \ref{sec:montecarlo}.\footnote{Theorem \ref{prop:non_observed} assumes that the partial observability of $\mathbf{A}_m$ implies that $\mathbf{G}_m\mathbf{X}_m$ and $\mathbf{G}_m\mathbf{y}_m$ are both unobserved. However, in some cases, researchers can separately observe these quantities from survey questions. For example, one could simply obtain ${\mathbf{G}_m\mathbf{y}_m}$ from a question of the type: ``What is the average value of your friends' $y$?'' In these cases, it is possible to improve on our SGMM estimator by using this additional information. The resulting estimators are presented in Corollary \ref{prop:GXobs} and Corollary \ref{prop:Gyobs} of the Online Appendix D.}

\subsection{Monte Carlo Simulations}\label{sec:montecarlo}

In this subsection, we study the performance of our SGMM estimator using Monte Carlo simulations. We consider cases where links are missing at random (see Example \ref{ex:sampled}) and misclassified at random (see Example \ref{ex:misclass}). The simulated individual characteristics (i.e., the matrix $\mathbf X$) include two characteristics similar to "age" and "female" in our empirical application.\footnote{See Section \ref{sec:addhealth}. We simulate those variables from their empirical distributions in our sample. Parameter values are set to the estimates from our application: $(\alpha, \boldsymbol{\beta}, \boldsymbol{\gamma}) = (0.538, 3.806, -0.072, 0.132, 0.086, -0.003)$. {We assume that $\varepsilon$ is iid normally distributed with standard deviation of $\sigma = 0.707$.}}  The network formation process follows a logistic regression model: $$P(a_{m,ij} = 1|\mathbf X) = \frac{\exp\{\rho_1 + \rho_2 |x_{m,i1} - x_{m,j1}| + \rho_3 \mathbbm{1}\{x_{m,i2} = x_{m,j2}\}\}}{1 + \exp\{\rho_1 + \rho_2 |x_{m,i1} - x_{m,j1}| + \rho_3 \mathbbm{1}\{x_{m,i2} = x_{m,j2}\}\}},$$ where $x_{m,i1}$ represents "age" and $x_{m,i2}$ represents "female".\footnote{The parameter vector $\boldsymbol{\rho}$ is also set to its empirically estimated values: $\boldsymbol{\rho} = (-2.349, -0.700, 0.404)$.}


We analyze different proportions of randomly missing and misclassified entries in the network matrix. Figure \ref{fig:mc_noFE} presents estimates for the endogenous peer effect coefficient $\alpha$ using our SGMM estimator. Figure \ref{fig:misslinks} shows the peer effect estimates for the case of missing links, while Figure \ref{fig:misclass} displays the estimates for the case of misclassified links.\footnote{Tables \ref{tan:sim:noFE}--\ref{tab:simmclas:FE} in Online Appendix E provide the full set of estimated coefficients, including results that control for unobserved group heterogeneity through fixed effects.} Additionally, we report estimates obtained using the standard IV estimator of \cite{bramoulle2009identification}, treating the observed network with missing values or misclassified links as the true network.

\begin{figure}[h!]
    \centering
    \begin{subfigure}{0.49\textwidth}
        \centering
        \includegraphics[width=\textwidth]{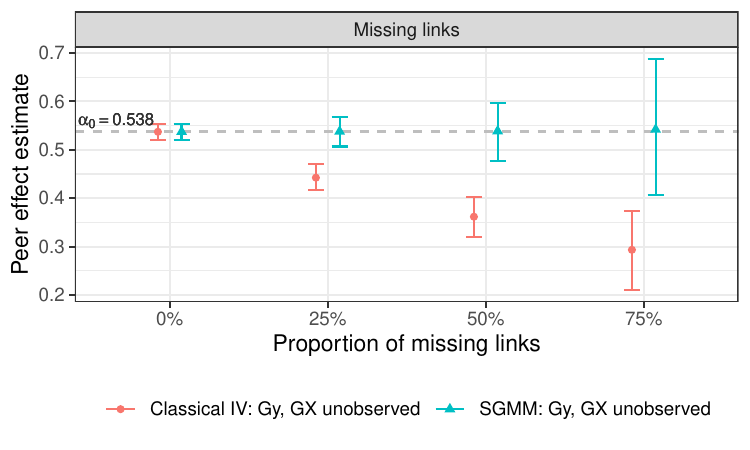}
        \subcaption{Missing Links}
        \label{fig:misslinks}
    \end{subfigure}
    \begin{subfigure}{0.49\textwidth}
        \centering
        \includegraphics[width=\textwidth]{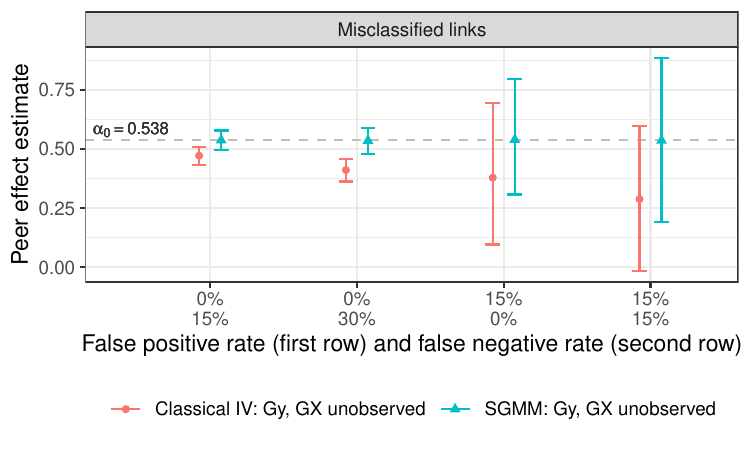}
        \subcaption{Misclassified Links}
        \label{fig:misclass}
    \end{subfigure}
    \caption{Estimated peer effects with mismeasured links}
    \label{fig:mc_noFE}

   \begin{minipage}{16cm}%
  \vspace{0.3cm}
\footnotesize{Note: Dots represent the average estimated values of $\alpha$, and bars indicate 95\% confidence intervals. Tables \ref{tan:sim:noFE}--\ref{tab:simmclas:FE} in Online Appendix E provide the full set of estimated coefficients. The "Classical IV" refers to the standard estimator of \cite{bramoulle2009identification}. We simulate data for 100 groups of 30 individuals each. We assume that $\varepsilon_i$ follows a normal distribution. We estimate $\boldsymbol{\rho}$ using a logit model based on the observed network entries (Figure \ref{fig:misslinks}) and a logit model with misclassification (Figure \ref{fig:misclass}). The resulting estimates allow us to construct the network distribution (see Definition \ref{def:estimator}) and subsequently compute our SGMM estimator. We set $R = 100$ and $S = T = 1$.}
  \end{minipage}%
\end{figure}

For the case of missing links, the estimates are centered around the true value. Although precision decreases as the fraction of missing links increases, our SGMM estimator maintains a reasonable level of accuracy, even when \emph{half} of the links are missing. In contrast, the standard IV estimator significantly underestimates the peer effect coefficient $\alpha$.  

For the case of misclassified links, the estimator performs well when there are false negatives only. Precision is affected when there are false positives, although the estimates remain centered around the true value. With false positives, the estimator for $\boldsymbol{\rho}$ loses precision since the network is simulated to match the one in our application: the density of the network is low.\footnote{This is typical of most network data: two randomly selected individuals are unlikely to be linked, even conditional on observables.} With few links, the finite sample cost of false positives is thus more important.

%% file: bayesian.tex
\section{Bayesian  Estimator}\label{sec:genestim}

In this section, we present a likelihood-based estimator. Accordingly, greater structure must be imposed on the errors $\boldsymbol{\varepsilon}_m$. 
Specifically, given parametric assumptions for $\boldsymbol{\varepsilon}_m$, one can write the log-likelihood of the outcome as:
\begin{equation}\label{eq:likelihood}
\ln \mathcal{P}(\mathbf{y}|\mathbf{A},\mathbf{X},\boldsymbol{\theta})=\sum_m \ln\mathcal{P}(\mathbf{y}_m|\mathbf{A}_m,\mathbf{X}_m;\boldsymbol{\theta}) ,    
\end{equation}
where notation without the index $m$ denotes vectors and matrices at the sample level. We abuse notation by letting $\boldsymbol{\theta}=[\alpha,\boldsymbol{\beta}',\boldsymbol{\gamma}',\boldsymbol{\sigma}']'$, which now includes $\boldsymbol{\sigma}$, additional unknown parameters of the distribution of $\boldsymbol{\varepsilon}_m$. Recall that from Equation (\ref{eq:linearinmean}), we have: $\mathbf{y}_m=(\mathbf{I}_m-\alpha\mathbf{G}_m)^{-1}(c\mathbf{1}_m+\mathbf{X}_m\boldsymbol{\beta}+{\mathbf{G}_m\mathbf{X}_m}\boldsymbol{\gamma}+\boldsymbol{\varepsilon}_m)$ since $(\mathbf{I}_m-\alpha\mathbf{G}_m)^{-1}$ exists under our Assumption \ref{as:alpha}.

If the adjacency matrix $\mathbf{A}_m$ is observed, then $\boldsymbol{\theta}$ could be estimated using a simple maximum likelihood estimator \citep[as in][]{lee2010specification} or using Bayesian inference \cite[as in][]{goldsmith2013social}. See in particular the identification conditions presented in \cite{lee2004asymptotic} and \cite{lee2010specification}. Since $\mathbf{A}_m$ is not observed, but $\mathcal{A}_m$ is observed, we focus on the following alternative likelihood:
$$
\ln \mathcal{P}(\mathbf{y}|\mathcal{A},\mathbf{X};\boldsymbol{\theta},\boldsymbol{\rho})=\sum_m \ln \sum_{\mathbf{A}_m}\mathcal{P}(\mathbf{y}_m|\mathbf{A}_m,\boldsymbol{\theta}){P}(\mathbf{A}_m|{\boldsymbol{\rho}},\mathbf{X}_m,\mathcal{A}_m).
$$
That is, we integrate the likelihood using the posterior distribution obtained from the network formation model in Equation (\ref{eq:gennetfor}) after observing $\mathcal{A}_m$.\footnote{See Equation (\ref{eq:Bayes}). Note also that, conceptually, we could condition on $\kappa(\mathcal{A})$ instead of $\mathcal{A}$ as in Section \ref{sec:iv}. However, this is much less attractive from a Bayesian perspective and thus limits ourselves to this (more efficient) case.}

One particular issue with estimating $\ln \mathcal{P}(\mathbf{y}|\mathcal{A},\mathbf{X};\boldsymbol{\theta},\boldsymbol{\rho})$ is that the summations over the set of all possible network structures $\mathbf{A}_m$, for each group $m$ is not tractable. Indeed, for a group of size $N_m$, the sum is over the set of possible adjacency matrices, which contain $2^{N_m(N_m-1)}$ elements. Then, simply simulating networks from ${P}(\mathbf{A}_m|{\boldsymbol{\rho}},\mathbf{X}_m,\mathcal{A}_m)$ and taking the average likely lead to poor approximations. A classical way to address this issue is to use an EM algorithm \citep{hardy2019estimating}. Although valid, we found that the Bayesian estimator proposed in this section is less restrictive and numerically outperforms its classical counterpart. The Bayesian treatment also has the advantage of being valid in finite samples, allowing for a richer set of network formation models and partially observed network information $\mathcal{A}$.\footnote{For example, models estimated using Aggregated Relational Data, see the Online Appendix H.}

For concreteness, we will assume that $\boldsymbol{\varepsilon}_m\sim\mathcal{N}(0,\sigma^2\mathbf{I}_m)$ for all $m$; however, it should be noted that our approach is valid for several alternative assumptions as long as it yields a computationally tractable likelihood. For each group $m$, and recalling that $\mathbf{G}_m=f(\mathbf{A}_m)$, we have:
\begin{eqnarray*}
\ln \mathcal{P}(\mathbf{y}_m|\mathbf{A}_m,\boldsymbol{\theta})&=&-N_m\ln(\sigma)+\ln|\mathbf{I}_m-\alpha\mathbf{G}_m|-\frac{N_m}{2}\ln(\pi)\\&&-\frac{1}{2\sigma^2}[(\mathbf{I}_m-\alpha\mathbf{G}_m)\mathbf{y}_m-c\mathbf{1}_m-\mathbf{X}_m\boldsymbol{\beta}-\mathbf{G}_m\mathbf{X}_m\boldsymbol{\gamma}]'\cdot\\ &&[(\mathbf{I}_m-\alpha\mathbf{G}_m)\mathbf{y}_m-c\mathbf{1}_m-\mathbf{X}_m\boldsymbol{\beta}-\mathbf{G}_m\mathbf{X}_m\boldsymbol{\gamma}].
\end{eqnarray*}
Because $\mathbf{A}_m$ is not observed, we follow \cite{tanner1987calculation}, and we use data augmentation to evaluate the posterior distribution of $\boldsymbol{\theta}$. That is, instead of focusing on the posterior distribution of $\boldsymbol{\theta}$ (i.e., $P(\boldsymbol{\theta}|\mathbf{y},\mathbf{A},\mathbf{X})$) in the case in which the network was observed, we focus instead on the posterior distribution $P(\boldsymbol{\theta},\mathbf{A}|\mathbf{y},\mathcal{A},\mathbf{X})$, treating $\mathbf{A}$ as another set of unknown parameters.

Since the number of parameters to be estimated is larger than the number of observations,\footnote{Each group contains $N_m$ observations while the dimension of $\mathbf{A}_m$ is $N_m(N_m-1)$.} the identification of the model rests on the a priori information on $\mathbf{A}$. A sensible prior for $\mathbf{A}$ is the consistent estimator of its distribution, i.e., $\Pi_m\hat{P}(\mathbf{A}_m|\hat{\boldsymbol{\rho}},\mathbf{X}_m,\mathcal{A}_m)$. Let $\pi(\boldsymbol{\rho}|\mathbf{X},\mathcal{A})$ be the prior density on $\boldsymbol{\rho}$. How to obtain $\pi(\boldsymbol{\rho}|\mathbf{X},\mathcal{A})$, depending on whether $\hat{\boldsymbol{\rho}}$ is obtained using a Bayesian or classical setting, is discussed in Examples 4 and 5 of the Online Appendix F.3. Given $\pi(\boldsymbol{\rho}|\mathbf{X},\mathcal{A})$, it is possible to obtain draws from the posterior distribution $P(\boldsymbol{\theta},\mathbf{A},\boldsymbol{\rho}|\mathbf{y},\mathcal{A})$ using the following Metropolis-Hastings MCMC:\footnote{As customary, for the rest of this section, we omit the dependence on $\mathbf{X}$ to lighten the notation. The notation with the index $t-1$ in this section refers to the $(t-1)$-th iteration of the MCMC, not the $(t-1)$-th group. Specifically, $\mathbf{A}_{t-1}$ denotes the adjacency matrix at the sample level in iteration $t-1$. Since the MCMC is a Metropolis-Hastings, the detailed balance and ergodicity conditions hold, so the MCMC converges to $P(\boldsymbol{\theta},\mathbf{A},\boldsymbol{\rho}|\mathbf{y},\mathcal{A})$. See \cite{cameron2005microeconometrics}, Section 13.5.4 for more details.}


\begin{algorithm}\label{algo:mcmc} The MCMC goes as follows for $t=1,...,T$, starting from any $\mathbf{A}_0,\boldsymbol{\theta}_0$, and $\boldsymbol{\rho}_0$.
\begin{enumerate}
    \item Draw $\boldsymbol{\rho}^*$ from the proposal distribution $q_\rho(\boldsymbol{\rho}^*|\boldsymbol{\rho}_{t-1})$ and accept $\boldsymbol{\rho}^*$ with probability
    $$
    \min\left\lbrace 1,\frac{P(\mathbf{A}_{t-1}|\boldsymbol{\rho}^*,\mathcal{A})q_\rho(\boldsymbol{\rho}_{t-1}|\boldsymbol{\rho}^*)\pi(\boldsymbol{\rho}^*|\mathcal{A})}
    {P(\mathbf{A}_{t-1}|\boldsymbol{\rho}_{t-1},\mathcal{A})q_\rho(\boldsymbol{\rho}^*|\boldsymbol{\rho}_{t-1})\pi(\boldsymbol{\rho}_{t-1}|\mathcal{A})} \right\rbrace.
    $$
    \item Propose $\mathbf{A}^*$ from the proposal distribution $q_A(\mathbf{A}^*|\mathbf{A}_{t-1})$ and accept $\mathbf{A}^*$ with probability
    $$
    \min\left\lbrace 1,\frac{\mathcal{P}(\mathbf{y}|\boldsymbol{\theta}_{t-1},\mathbf{A}^*)q_A(\mathbf{A}_{t-1}|\mathbf{A}^*)P(\mathbf{A}^*|{\boldsymbol{\rho}_{t-1}},\mathcal{A})}{\mathcal{P}(\mathbf{y}|\boldsymbol{\theta}_{t-1},\mathbf{A}_{t-1})q_A(\mathbf{A}^*|\mathbf{A}_{t-1})P(\mathbf{A}_{t-1}|{\boldsymbol{\rho}_{t-1}},\mathcal{A})} \right\rbrace.
    $$
    \item Draw $\alpha^*$ from the proposal $q_\alpha(\cdot|\alpha_{t-1})$ and accept $\alpha^*$ with probability
    $$
    \min\left\lbrace 1,\frac{\mathcal{P}(\mathbf{y}|\mathbf{A}_{t};\boldsymbol{\beta}_{t-1},\boldsymbol{\gamma}_{t-1},\alpha^*)q_\alpha(\alpha_{t - 1}|\alpha^*)\pi(\alpha^*)}{\mathcal{P}(\mathbf{y}|\mathbf{A}_{t};\boldsymbol{\theta}_{t-1})q_\alpha(\alpha^*|\alpha_{t-1})\pi(\alpha_{t-1})}\right\rbrace.
    $$
    \item Draw $[\boldsymbol{\beta,\gamma},\sigma]$ from their posterior conditional distributions (see Online Appendix F).
\end{enumerate}
\end{algorithm}

Step 1 allows to refine the estimation of $\boldsymbol{\rho}$. Indeed, in the first stage, $\boldsymbol{\rho}$ is inferred using the information provided by $\mathcal{A}$. In Step 1, however, $\boldsymbol{\rho}$ is updated conditional on $\mathcal{A}$ \emph{and} $\mathbf{A}_{t-1}$. This provides additional information not available in the first stage since $\mathbf{A}_{t-1}$ uses information provided by the likelihood function (\ref{eq:likelihood}).\footnote{In other words, $\boldsymbol{\rho}$ enters the likelihood of $\boldsymbol{y}$, conditional on $\mathcal{A}$.}

Steps 3 and 4 are standard, and detailed distributions can be found in the Online Appendix F. Step 2, however, requires some discussion. Indeed, the idea is the following: given the prior information $P(\textbf{A}|\boldsymbol{\rho}_{t-1},\mathcal{A})$, one must be able to draw samples from the posterior distribution of $\mathbf{A}$, given $\mathbf{y}$. This is not a trivial task.

In particular, there is no general rule for selecting the network proposal distribution $q_A(\cdot\vert\cdot)$. A natural candidate is a Gibbs sampling algorithm for each link, i.e., change only one link $ij$ at every step $t$ and propose $a_{ij}$ according to its marginal distribution, i.e., $a_{ij}\sim P(\cdot|\mathbf{A}_{-ij},\mathbf{y},\mathcal{A})$, where $\mathbf{A}_{-ij} = \{a_{kl}; k \ne i, l \ne j\}$. In this case, the proposal is always accepted.

However, it has been argued that Gibbs sampling could lead to slow convergence \citep[e.g.,][]{snijders2002markov,chatterjee2013estimating}, especially when the network is \emph{sparse} or exhibits a high level of \emph{clustering}. For example, \cite{mele2017} and \cite{bhamidi2008mixing} propose different blocking techniques meant to improve convergence.

Here, however, achieving Step 2 involves an additional computational issue because evaluating the likelihood ratio in Step 1 requires comparing the determinants $|\mathbf{I}-\alpha f(\mathbf{A}^*)|$ for each proposed $\mathbf{A}^*$, which is computationally intensive. 


Then, the appropriate blocking technique depends strongly on $P(\textbf{A}|\boldsymbol{\rho}_{t-1},\mathcal{A})$ and the assumed distribution for $\boldsymbol{\varepsilon}$. For the simulations and estimations presented in this paper, we use the Gibbs sampling algorithm for each link, adapting the strategy proposed by \cite{hsieh2019structural} to our setting (see Proposition 3 in the Online Appendix F.2). This can be viewed as a \emph{worst-case} scenario. Nonetheless, the Gibbs sampler performs reasonably well in practice; however, we encourage researchers to try other updating schemes if Gibbs sampling performs poorly in their specific contexts. In particular, we present a blocking technique in the Online Appendix F that is also implemented in our R package \texttt{PartialNetwork}.\footnote{The complexity of Step 2 is not limited to our Bayesian approach. Classical estimators, such as GMM estimators, face a similar challenge in requiring the integration over the entire set of networks. The strategy used here is to rely on a Metropolis-Hastings algorithm, a strategy that has also been successfully implemented in the related literature on ERGMs \cite[e.g.,][]{snijders2002markov,mele2017,mele2020does,badev2018nash,hsieh2019structural}.}

Finally, note that for simple network formation models, it is possible to jointly estimate $\boldsymbol{\rho}$ and $\boldsymbol{\theta}$ within the same MCMC instead of using the two-step procedure described above. In this case, Step 1 can simply be replaced by:
\begin{enumerate}
    \item[1'.]  Draw $\boldsymbol{\rho}^*$ from the proposal distribution $q_\rho(\boldsymbol{\rho}^*|\boldsymbol{\rho}_{t-1})$ and accept $\boldsymbol{\rho}^*$ with probability
    $$
    \min\left\lbrace 1,\frac{P(\mathbf{A}_{t-1}|\boldsymbol{\rho}^*,\mathcal{A})P(\mathcal{A}|\boldsymbol{\rho}^*)q_\rho(\boldsymbol{\rho}_{t-1}|\boldsymbol{\rho}^*)\pi(\boldsymbol{\rho}^*)}
    {P(\mathbf{A}_{t-1}|\boldsymbol{\rho}_{t-1},\mathcal{A})P(\mathcal{A}|\boldsymbol{\rho}_{t-1})q_\rho(\boldsymbol{\rho}^*|\boldsymbol{\rho}_{t-1})\pi(\boldsymbol{\rho}_{t-1})} \right\rbrace.
    $$
\end{enumerate}
Here, $P(\mathcal{A}|\boldsymbol{\rho}^*)$ is the likelihood of the network information $\mathcal{A}$ assuming the network formation model in (\ref{eq:gennetfor}). Note that $\pi(\boldsymbol{\rho})$, the prior density on $\boldsymbol{\rho}$, no longer depends on $\mathcal{A}$ and can be chosen arbitrarily (e.g., uniform).


%% file: addhealth.tex
\section{Application}\label{sec:addhealth}

In this section, we assume that the econometrician has access to network data but that the data may contain errors due to both \emph{sampling} (links coded with errors) and \emph{censoring}. To show how our method can be used to address these issues, we consider a simple example where we are interested in estimating peer effects on adolescents' academic achievements.

We use the widely used AddHealth database and show that network data errors have a large impact on the estimated peer effects. Specifically, we focus on a subset of schools from the Wave I ``In School'' sample that have less than 200 students (33 schools). Table G.1 in the Online Appendix G.3 presents the summary statistics.

Most papers estimating peer effects that use this particular database have taken the network structure as given. One notable exception is \cite{griffith2019namefriends}, looking at censoring: students can only report up to five male and five female friends. We also allow for censoring, but show that censoring is not the most important issue with the Add Health data. To understand why, we discuss the organization of the data.

Each adolescent is assigned a unique identifier. The data includes ten variables for the ten potential friendships (a maximum of five male and five female friends). These variables can contain missing values (no friendship was reported), an error code (the named friend could not be found in the database), or an identifier for the reported friends. These data are then used to generate the network's adjacency matrix $\mathbf{A}$.

Of course, error codes cannot be matched to any particular adolescent. Moreover, even in the case where the friendship variable refers to a valid identifier, the referred adolescent may still be absent from the database. A prime example is when the referred adolescent has been removed from the database by the researcher, perhaps because of other missing variables for these particular individuals. These missing links are quantitatively important as they account for roughly 45\% of the total number of links (7,830 missing for 10,163 observed links). Figure \ref{fig:misfr} displays the distribution of the number of ``unmatched named friends.''\footnote{We focus on within-school friendships; thus, nominations outside of school are not treated as ``unmatched friends.'' Note also that these data errors could be viewed as a special case of censoring \citep{griffith2019namefriends} in which researchers know exactly how many links are censored. The attenuation bias is thus expected.}
\begin{figure}[htbp]
    \centering
    \includegraphics[scale=0.6]{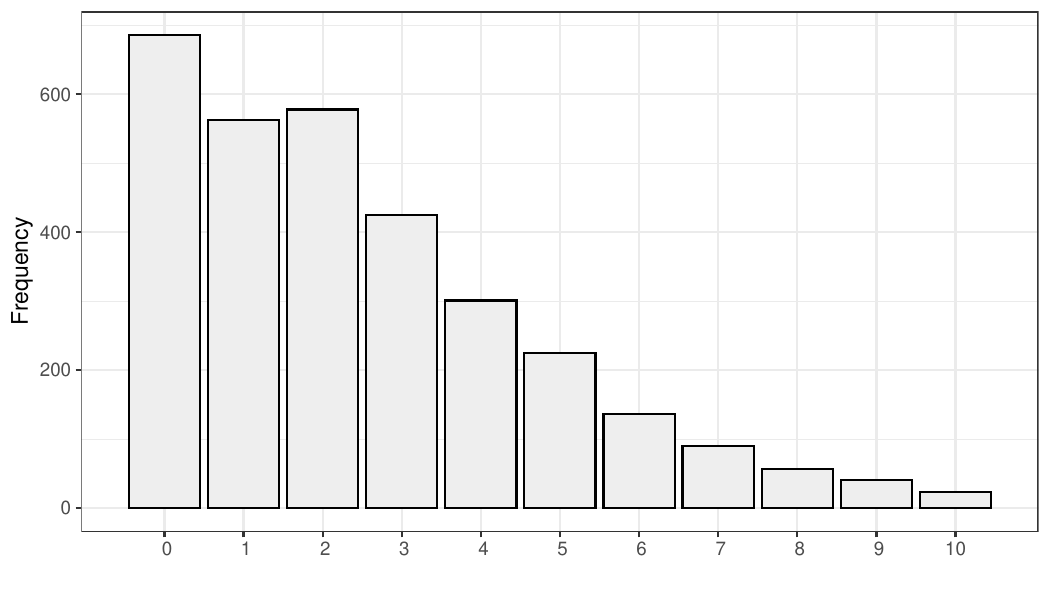}
    \caption{Frequencies of the number of missing links per adolescent}
    \label{fig:misfr}
\end{figure}

To use the methodology developed in sections \ref{sec:iv} and \ref{sec:genestim}, we first need to estimate a network formation model using the observed network data. In this section, we assume that links are generated using a simple logistic framework, i.e.,
$$
P(a_{ij,m}=1|\mathbf{X}_m)=\frac{\exp\{\mathbf{w}_{ij,m}\boldsymbol{\rho}\}}{1+\exp\{\mathbf{w}_{ij,m}\boldsymbol{\rho}\}},
$$
where $\mathbf{w}_{ij,m}$ is built to capture homophily on the observed characteristics of $i$ and $j$ (see Tables G.2 and G.3 in the Online Appendix G.3).

We estimate the network formation model on the set of individuals for which we observe no ``unmatched friends.'' For these students, we know for sure that their friendship data are complete. However, even under a missing at-random assumption, the estimation of $\boldsymbol{\rho}$ on this subsample is affected by a selection bias: individuals with more friends have a higher probability of being censored, or of having a friendship nomination coded with error.\footnote{Note that this is different from the random sampling discussed in our Example \ref{ex:sampled} and closer to the misclassification in Example \ref{ex:misclass}, with only false-negative type of errors.}

We control for this selection bias by weighting the log-likelihood of the network following \cite{manski1977estimation}. The details are presented in the Online Appendix G.1 and Online Appendix G.2. Intuitively, individuals in our restricted sample have fewer links. Therefore, the likelihood of \( a_{i,j} \) when \( i \) is selected in our restricted sample is weighted by the inverse selection probability. When accounting for missing data due to error codes only, we estimate the selection probability for an individual \( i \) who declares \( n_i \) friends as the proportion of individuals without missing network data who declare \( n_i \) friends. 

We use the same approach when controlling for missing data due to both error codes and censoring. However, in this case, the individual's censored number of friends has to be replaced with the (unobserved) true number of friends. We estimate individuals' true number of friends using a censored Poisson regression, where the observed number of friends in the network is used as the censored dependent variable: the variable is censored when individual \( i \) nominates five male friends or five female friends.

We present the estimation results for the SGMM and Bayesian estimator. Figure \ref{fig:ahestimations} summarizes the results for the endogenous peer effect coefficient $\alpha$, whereas the full set of results is presented in the Online Appendix G.3. The first two estimations (\emph{Obsv.Bayes} and \emph{Obsv.SGMM}) assume that the observed network is the true network for both estimators. The third and fourth estimations (\emph{Miss.Bayes} and \emph{Miss.SGMM}) account for missing data due to error codes but not for censoring. The last two estimations (\emph{TopMiss.Bayes} and \emph{TopMiss.SGMM}) account for missing data due to error codes and censoring.

\begin{figure}[h]
    \centering
    \includegraphics[scale=0.8]{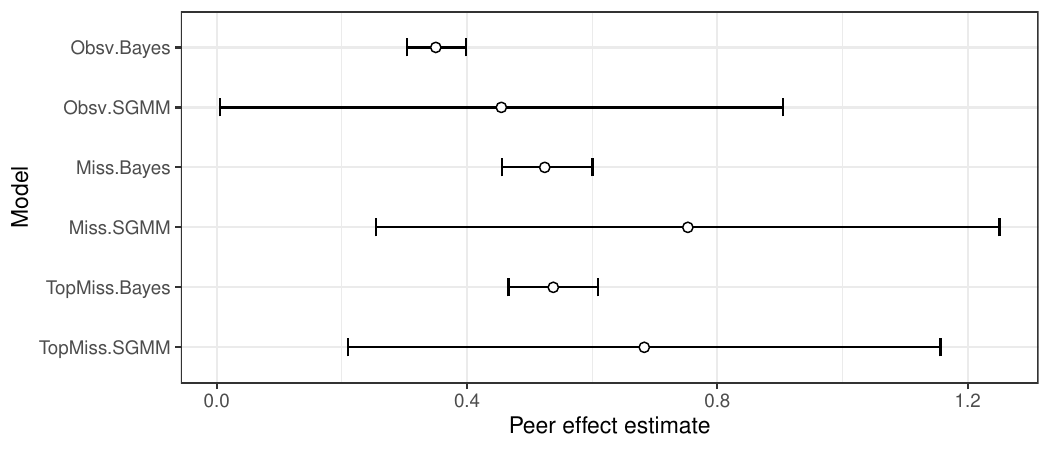}
    \caption{Peer effect estimate}
    \label{fig:ahestimations}

      \begin{minipage}{14cm}%
  \vspace{0.3cm}
    \footnotesize{Note: Dots represent estimated values (and posterior mean) of $\alpha$, and bars represent 95\% confidence intervals (and 95\% credibility intervals). Tables G.2 and G.3 in Online Appendix G.3 present the full set of estimated coefficients.}
  \end{minipage}%
\end{figure}

We first see that the SGMM estimator is less efficient than the Bayesian estimator. This should not be surprising since the Bayesian estimator uses more structure (in particular, homoscedastic, normally distributed errors). When we compare the estimations \emph{Obsv.SGMM} and \emph{Miss.SGMM}, the observed differences imply that the efficiency loss is because of the relative inefficiency of the GMM approach, and not of the missing links or specifically of our SGMM estimator.\footnote{Recall that when the network is observed, our SGMM uses the same moment conditions as, for example, those suggested by \cite{bramoulle2009identification}.}

Importantly, we see that the bias due to the assumption that the network is fully observed is quantitatively and qualitatively important. Using either estimator, the estimated endogenous peer effect using the reconstructed network is 1.5 times larger than that estimated assuming the observed network is the true network.\footnote{The difference is ``statistically significant'' for the Bayesian estimator.} Almost all of the bias is produced by the presence of error codes and not because of potential censoring.

This exercise shows that data errors are a main concern when using the Add Health database. Not only does the bias in the endogenous peer effect coefficient $\alpha$ have an impact on the social multiplier \citep{glaeser2003social}, but it can also affect the anticipated effect of targeted interventions, i.e., the identity of the key player \citep{ballester2006s}. We include a more detailed discussion in Appendix \ref{sec:keyplayer}.

However, we would like to stress that we do not argue that our estimated coefficients are causal, because the friendship network is likely endogenous \citep[e.g.,][]{goldsmith2013social,hsieh2018smoking,hsieh2019specification}. While previous literature has focused on the impact of network endogeneity, it has done so by assuming that the network is fully observed, despite the fact that roughly 45\% of the links are missing. Above, we showed that errors in the observed network have a large impact on the estimated peer effect, even when one assumes that the network is exogenous.

%% file: discussion.tex
\section{Conclusion}\label{sec:conclusion}
In this paper, we propose two estimators for which peer effects can be estimated without observing the entire network structure. We find, perhaps surprisingly, that even very partial information on network structure is sufficient. By specifying a network formation model, researchers can probabilistically reconstruct the true network and base the estimation of peer effects on this reconstructed network. Importantly, we provide computationally tractable and flexible estimators to do so, all of which are available in our R package \texttt{PartialNetwork}. We apply our methodology to the widely used Add Health data and find that missing links due to noise in the data have large effects on the estimated peer effect coefficient.

%% file: appendix_proof.tex
\section{Appendix: Proof of Theorem \ref{prop:non_observed}}\label{sec:appendix_proofs}
For the sake of clarity, we often write objects that depend on simulated networks as functions of $\boldsymbol{\rho}$; e.g., we write $\dot{\mathbf{Z}}_{m}(\boldsymbol{\rho})$ and $\dot{\mathbf{G}}_m(\boldsymbol{\rho})$ instead of $\dot{\mathbf{Z}}_{m}$ and $\dot{\mathbf{G}}_m$, unless this precision is unnecessary for the exposition. We define: $$\textstyle\mathbf{m}_{m,rst}(\boldsymbol{\theta},\boldsymbol{\rho})=\dot{\mathbf{Z}}_{m}^{(r)\prime}(\boldsymbol{\rho})(\mathbf{I}-\alpha\ddot{\mathbf{G}}_m^{(s)}(\boldsymbol{\rho}))\left(\mathbf{y}_m-(\mathbf{I}_m-\alpha\dddot{\mathbf{G}}_m^{(t)}(\boldsymbol{\rho}))^{-1}\dddot{\mathbf{V}}_m^{(t)}(\boldsymbol{\rho})\tilde{\boldsymbol{\theta}}\right).$$
Let also $\mathbf{m}_m(\boldsymbol{\theta},\boldsymbol{\rho})=\frac{1}{RST}\sum_{rst}{\mathbf{m}}_{m,rst}(\boldsymbol{\theta},\boldsymbol{\rho})$ and $\bar{\mathbf{m}}_M(\boldsymbol{\theta}, \boldsymbol{\rho})=\frac{1}{M}\sum_m\mathbf{m}_m(\boldsymbol{\theta}, \boldsymbol{\rho})$. The objective function of the SGMM is given by:
$$\mathcal{Q}_M(\boldsymbol{\theta}) =  [\bar{\mathbf{m}}_M(\boldsymbol{\theta},\hat{\boldsymbol{\rho}})]^{\prime}\mathbf{W}_M[\bar{\mathbf{m}}_M(\boldsymbol{\theta},\hat{\boldsymbol{\rho}})],$$
where $\mathbf{W}_M$ is a weighing matrix. The SGMM estimator is $\hat{\boldsymbol{\theta}} = \textstyle\arg\max_{\boldsymbol{\theta}} \mathcal{Q}_M(\boldsymbol{\theta})$.

We impose the following regularity assumptions.
\begin{assumption}\label{as:regularity}
$\boldsymbol{\rho}_0$ and $\boldsymbol{\theta}_0$ are interior points of $\Theta$ and $\mathcal{R}$, respectively, where both $\Theta$ and $\mathcal{R}$ are compact subsets of the Euclidean space.
\end{assumption}

\begin{assumption}\label{as:nonsingular}
(i) For all $m = 1, ..., M$, $r=1,...,R$, $s=1,...,S$, and $t=1,...,T$, $(\mathbf{I}_m-\alpha \mathbf{G}_m)$ and $(\mathbf{I}_m-\alpha \dddot{\mathbf{G}}^{(t)}_m)$ are non-singular. (ii) The $(i,j)$-th entries of $\mathbf{G}_m$ (so $\dot{\mathbf{G}}^{(r)}_m$, $\ddot{\mathbf{G}}^{(s)}_m$, and $\dddot{\mathbf{G}}^{(t)}_m$), $(\mathbf{I}_m-\alpha\mathbf{G}_m)^{-1}$,  and $(\mathbf{I}_m-\alpha\dddot{\mathbf{G}}_m^{(t)})^{-1}$ are bounded uniformly in $i$, $j$, and $m$.
\end{assumption}
In particular, when $\mathbf{G}_m$ is row-normalized (so $\dot{\mathbf{G}}^{(r)}_m$, $\ddot{\mathbf{G}}^{(s)}_m$, and $\dddot{\mathbf{G}}^{(t)}$ are also row-normalized), Assumption \ref{as:alpha} implies Assumption \ref{as:nonsingular}.

\begin{assumption}\label{as:finitevariance}
$\sup_{m \geq 1}\mathbb E\{\lVert\boldsymbol{\varepsilon}_m\rVert_2^{\mu}|\mathbf{X}_m, \mathcal{A}_m\}$ exists and is bounded, for some $\mu > 2$, where $\Vert . \rVert_2$ is the Euclidean norm.
\end{assumption}

\begin{assumption}\label{as:derivatives_P}
    The derivative of $\hat{P}(a_{ij,m}|{\boldsymbol{\rho}},\mathbf{X}_m,\kappa(\mathcal{A}_m))$ with respect to $\boldsymbol{\rho}$ is bounded uniformly in $i$, $j$, and $m$.
\end{assumption}

\begin{assumption}\label{as:Wm}
    $\mathbf{W}_M$ is positive definite and $\plim \mathbf{W}_M = \mathbf{W}_0$, where $\plim$ denotes the probability limit as $M$ goes to infinity and $\mathbf{W}_0$ is a non-stochastic and positive definite matrix.
\end{assumption}




\subsection{Proof of the consistency of the SGMM}\label{sec:proofconsistencySGMM}

We proceed to show that Theorem 2.1 in \cite{newey1994large} applies to our SGMM estimator. The proof relies on the following Lemmatta.

\begin{lemma}[Validity of the moment function]\label{lemma:momentvalid}
    The moment condition is verified for $(\boldsymbol{\theta}_0, \boldsymbol{\rho}_0)$; that is, $\mathbb E(\mathbf{m}_m(\boldsymbol{\theta}_0, \boldsymbol{\rho}_0)) = \mathbf 0$ for all $m$.
\end{lemma}
\begin{proof}
    Let us substitute $\mathbf{y}_m=(\mathbf{I}_m-\alpha_0\mathbf{G}_m)^{-1}(\mathbf{V}_m\tilde{\boldsymbol{\theta}}_0+\boldsymbol{\varepsilon}_m)$ in the moment function. We have
\begin{eqnarray}\label{eq:m:lineps}
\begin{split}
\textstyle\mathbf{m}_{m,rst}(\boldsymbol{\theta}_0,\boldsymbol{\rho}_0)&=& \dot{\mathbf{Z}}_{m}^{(r)\prime}(\boldsymbol{\rho}_0)(\mathbf{I}_m-\alpha_0\ddot{\mathbf{G}}_m^{(s)}(\boldsymbol{\rho}_0))\left[(\mathbf{I}_m-\alpha_0\mathbf{G}_m)^{-1}\mathbf{V}_m\right.\\
&&\left.-(\mathbf{I}_m-\alpha_0\dddot{\mathbf{G}}_m^{(t)}(\boldsymbol{\rho}_0))^{-1}\dddot{\mathbf{V}}_m^{(t)}(\boldsymbol{\rho}_0)\right]\tilde{\boldsymbol{\theta}}_0\\
&+&\textstyle\dot{\mathbf{Z}}_{m}^{(r)\prime}(\boldsymbol{\rho}_0)(\mathbf{I}_m-\alpha_0\ddot{\mathbf{G}}_m^{(s)}(\boldsymbol{\rho}_0))(\mathbf{I}_m-\alpha_0\mathbf{G}_m)^{-1}\boldsymbol{\varepsilon}.
\end{split}
\end{eqnarray}
Consider the last part first. We have, for any $r$ and $s$:
$$\textstyle
\mathbb{E}\left(\dot{\mathbf{Z}}_{m}^{(r)\prime}(\boldsymbol{\rho}_0)(\mathbf{I}_m-\alpha_0\ddot{\mathbf{G}}_m^{(s)}(\boldsymbol{\rho}_0))(\mathbf{I}_m-\alpha_0\mathbf{G}_m)^{-1}\boldsymbol{\varepsilon}_m|\mathbf{X}_m, \kappa(\mathcal{A}_m)\right)=\mathbf{0},
$$
from Assumption \ref{as:exonet}.

Consider now the first part. Since network draws are independent, we have:
\begin{align*}
    \hat{\mathbb{E}}_m[\dot{\mathbf{Z}}_{m}^{(r)\prime}(\boldsymbol{\rho}_0)]\hat{\mathbb{E}}_m[(\mathbf{I}_m-\alpha_0\ddot{\mathbf{G}}_m^{(s)}(\boldsymbol{\rho}_0))]\Big(\mathbb{E}^{(0)}_m[(\mathbf{I}_m-\alpha_0\mathbf{G}_m)^{-1}\mathbf{V}_m]- &\\
    &\hspace{-4cm} \hat{\mathbb{E}}_m[(\mathbf{I}_m-\alpha_0\dddot{\mathbf{G}}_m^{(t)}(\boldsymbol{\rho}_0))^{-1}\dddot{\mathbf{V}}_m^{(t)}(\boldsymbol{\rho}_0)]\Big)\tilde{\boldsymbol{\theta}}_0,
\end{align*}
where $\hat{\mathbb{E}}_m$ denotes the expectation with respect to the distribution of the simulated networks, conditional on $\mathbf{X}_m,\kappa(\mathcal{A}_m)$, and where $\mathbb{E}^{(0)}_m$ is the expectation with respect to the distribution of the true network $\mathbf{G}_m$, conditional on $\mathbf{X}_m,\kappa(\mathcal{A}_m)$. Since, at $\boldsymbol{\rho}_0$, these are the same distributions, the terms in the big parenthesis cancel out, and thus, $\mathbb{E}[\mathbf{m}_{m,rst}(\boldsymbol{\theta}_0,\boldsymbol{\rho}_0)|\mathbf{X}_m,\kappa(\mathcal{A}_m)]=\mathbf{0}$. As a result, $\mathbb{E}[\mathbf{m}_m(\boldsymbol{\theta}_0,\boldsymbol{\rho}_0)]=\frac{1}{RST}\sum_{rst}\mathbb{E}[{\mathbf{m}}_{m,rst}(\boldsymbol{\theta}_0,\boldsymbol{\rho}_0)]=\mathbf{0}$ by the law of iterated expectations.
\end{proof}

\begin{lemma}[Differentiability]\label{lemma:continuity}
    $\mathbb{E}[\mathbf{m}_{m}(\boldsymbol{\theta},\boldsymbol{\rho})]$ is continuously differentiable in $(\boldsymbol{\theta}, \boldsymbol{\rho})$.
\end{lemma}

\begin{proof}
See the Online Appendix C.
\end{proof}

\begin{lemma}[Uniform convergence]\label{lemma:uniformconvergence} We establish the following results.
\begin{enumerate}[label=(\alph*), nosep]
    \item $\mathbb{E}[\bar{\mathbf{m}}_M(\boldsymbol{\theta},\hat{\boldsymbol{\rho}})] - \mathbb{E}[\bar{\mathbf{m}}_M(\boldsymbol{\theta},{\boldsymbol{\rho}}_0)]$ converges uniformly to $\mathbf{0}$ in $\boldsymbol{\theta}$ as $M \to \infty$.\label{lemma:uniformconvergence:A}
    \item $\bar{\mathbf{m}}_M(\boldsymbol{\theta},\hat{\boldsymbol{\rho}}) - \mathbb{E}[\bar{\mathbf{m}}_M(\boldsymbol{\theta},\hat{\boldsymbol{\rho}})]$ converges uniformly in probability to $\mathbf{0}$ in $\boldsymbol{\theta}$ as $M \to \infty$.\label{lemma:uniformconvergence:B}
    \item $\bar{\mathbf{m}}_M(\boldsymbol{\theta},\hat{\boldsymbol{\rho}}) - \mathbb{E}[\bar{\mathbf{m}}_M(\boldsymbol{\theta},{\boldsymbol{\rho}}_0)]$ converges uniformly in probability to $\mathbf{0}$ in $\boldsymbol{\theta}$ as $M \to \infty$. \label{lemma:uniformconvergence:C}
\end{enumerate}
\end{lemma}
\begin{proof}
See the Online Appendix C.
\end{proof}

The needed result from Lemma \ref{lemma:uniformconvergence} is Statement \ref{lemma:uniformconvergence:C}, which allows us to replace \( \hat{\boldsymbol{\rho}} \) with its limit \( \boldsymbol{\rho}_0 \) to show the consistency of $\hat{\boldsymbol{\theta}}$. However, this result is not trivial because the moment function is not continuous for all \( \boldsymbol{\rho} \). We thus first show Statements \ref{lemma:uniformconvergence:A} and \ref{lemma:uniformconvergence:B}, which together imply Statement \ref{lemma:uniformconvergence:C}.

\subsubsection*{Proof of Theorem \ref{prop:non_observed}}
We define:
$$\mathcal{Q}_0(\boldsymbol{\theta}) =  \Big[\lim\mathbb{E}[\bar{\mathbf{m}}_M(\boldsymbol{\theta},\boldsymbol{\rho}_0)]\Big]^{\prime}\mathbf{W}_0\Big[\lim\mathbb{E}[\bar{\mathbf{m}}_M(\boldsymbol{\theta},\boldsymbol{\rho}_0)]\Big].$$

As $\bar{\mathbf{m}}_M(\boldsymbol{\theta},\hat{\boldsymbol{\rho}})$ converges uniformly in probability to $\lim\mathbb{E}[\bar{\mathbf{m}}_M(\boldsymbol{\theta},{\boldsymbol{\rho}}_0)]$ in $\boldsymbol{\theta}$ (Lemma \ref{lemma:uniformconvergence}, Statement \ref{lemma:uniformconvergence:C}) and $\plim \mathbf W_M = \mathbf W_0$ (Assumption \ref{as:Wm}), by Cauchy-Schwartz \citep[see e.g., Theorem 2.6 in][]{newey1994large}, $\mathcal{Q}_M(\boldsymbol{\theta})$ converges uniformly in probability to $\mathcal{Q}_0(\boldsymbol{\theta})$.

From Theorem 2.1 in \cite{newey1994large}, consistency of $\hat{\boldsymbol{\theta}}$ requires: (i) $\mathcal{Q}_0(\boldsymbol{\theta})$ is uniquely minimized at $\boldsymbol{\theta}_0$ (which holds from Lemma \ref{lemma:momentvalid} and the Assumption that, for any $\boldsymbol{\theta}\ne \boldsymbol{\theta}_0$, $\lim_{M\rightarrow\infty} \mathbb{E}(\bar{\mathbf{m}}_M(\boldsymbol{\theta}, \boldsymbol{\rho}_0)) \ne \mathbf 0$), (ii) the parameter space for $\boldsymbol{\theta}$ is compact (which holds by Assumption \ref{as:regularity}), (iii) $\mathcal{Q}_0(\boldsymbol{\theta})$ is continuous (which holds from Lemma \ref{lemma:continuity}), (iv) $\mathcal{Q}_M(\boldsymbol{\theta})$ converges uniformly in probability to $\mathcal{Q}_0(\boldsymbol{\theta})$ (which holds by Lemma \ref{lemma:uniformconvergence} and Assumption \ref{as:Wm} as pointed out above). Therefore, $\hat{\boldsymbol{\theta}}$ is consistent.

\subsection{Proof of the Asymptotic Normality of the SGMM}\label{sec:asymptnorm}

We show that the SGMM estimator is asymptotically normally distributed. Recall that $$\bar{\mathbf{m}}_M(\boldsymbol{\theta}, {\boldsymbol{\rho}})=\frac{1}{M}\sum_m\bar{\mathbf{m}}_m(\boldsymbol{\theta}, {\boldsymbol{\rho}}),$$ and let $$\bar{\mathbf{m}}_M^*(\boldsymbol{\theta}, {\boldsymbol{\rho}})=\frac{1}{M}\sum_m\mathbb{E}\left(\bar{\mathbf{m}}_m(\boldsymbol{\theta}, {\boldsymbol{\rho}})\right).$$

Our proof relies on the following stochastic equicontinuity condition, which is formally shown in Lemma 4 in Online Appendix C:
\begin{condition}\label{cond:sequi}
    $\sqrt{M}[\bar{\mathbf{m}}_M({\boldsymbol{\theta}}_0,  \hat{\boldsymbol{\rho}}) -
\bar{\mathbf{m}}^{\ast}_M(\boldsymbol{\theta}_0, \hat{\boldsymbol{\rho}})]- \sqrt{M}[\bar{\mathbf{m}}_M(\boldsymbol{\theta}_0,  \boldsymbol{\rho}_0)-\bar{\mathbf{m}}^{\ast}_M({\boldsymbol{\theta}}_0, {\boldsymbol{\rho}}_0) ] =o_p(1)$
\end{condition}

The first order condition of the empirical objective function $\mathcal{Q}_M$ with respect to $\boldsymbol{\theta}$ is $\frac{\partial\bar{\mathbf{m}}_M'(\hat{\boldsymbol{\theta}},  \hat{\boldsymbol{\rho}})}{\partial {\boldsymbol{\theta}}}\mathbf W_M \bar{\mathbf{m}}_M(\hat{\boldsymbol{\theta}},  \hat{\boldsymbol{\rho}}) = \mathbf 0$. As $\bar{\mathbf{m}}_M({\boldsymbol{\theta}}_0,  {\boldsymbol{\rho}}_0)  - \bar{\mathbf{m}}_M({\boldsymbol{\theta}}_0,  {\boldsymbol{\rho}}_0) = \mathbf 0$, this implies: $$\frac{\partial\bar{\mathbf{m}}_M'(\hat{\boldsymbol{\theta}},  \hat{\boldsymbol{\rho}})}{\partial {\boldsymbol{\theta}}}\mathbf W_M \left[\bar{\mathbf{m}}_M(\hat{\boldsymbol{\theta}},  \hat{\boldsymbol{\rho}}) - \bar{\mathbf{m}}_M({\boldsymbol{\theta}}_0,  {\boldsymbol{\rho}}_0)  + \bar{\mathbf{m}}_M({\boldsymbol{\theta}}_0,  {\boldsymbol{\rho}}_0) \right]=\mathbf{0},$$

Given that $\bar{\mathbf{m}}_M(\boldsymbol{\theta},  \hat{\boldsymbol{\rho}})$ is differentiable in $\boldsymbol{\theta}$, we replace $\bar{\mathbf{m}}_M(\hat{\boldsymbol{\theta}},  \hat{\boldsymbol{\rho}})$ in the previous equation with its mean value expansion. After rearranging the terms, we obtain:
\begin{align}\label{eq:mvt1}
\begin{split}
    &\frac{\partial\bar{\mathbf{m}}_M'(\hat{\boldsymbol{\theta}},  \hat{\boldsymbol{\rho}})}{\partial {\boldsymbol{\theta}}}\mathbf W_M \frac{\partial \bar{\mathbf{m}}_M(\boldsymbol{\theta}^+, \hat{\boldsymbol{\rho}})}{\partial \boldsymbol{\theta}^{\prime}}(\hat{\boldsymbol{\theta}} - \boldsymbol{\theta}_0) = \\& 
    \quad
    -\frac{\partial\bar{\mathbf{m}}_M'(\hat{\boldsymbol{\theta}},  \hat{\boldsymbol{\rho}})}{\partial {\boldsymbol{\theta}}}\mathbf W_M  \left[\bar{\mathbf{m}}_M(\boldsymbol{\theta}_0,  \hat{\boldsymbol{\rho}}) - \bar{\mathbf{m}}_M({\boldsymbol{\theta}}_0,  {\boldsymbol{\rho}}_0)  + \bar{\mathbf{m}}_M({\boldsymbol{\theta}}_0,  {\boldsymbol{\rho}}_0) \right]=\mathbf{0},
\end{split}
\end{align}
for some $\boldsymbol{\theta}^+$ lying between $\boldsymbol{\theta}_0$ and $\hat{\boldsymbol{\theta}}$.

From Lemma \ref{lemma:continuity}, $\bar{\mathbf{m}}_M^*(\boldsymbol{\theta}_0, {\boldsymbol{\rho}})$ is continuously differentiable in ${\boldsymbol{\rho}}$. Thus, by the mean value theorem, we have:
$$\bar{\mathbf{m}}^{\ast}_M(\boldsymbol{\theta}_0, \hat{\boldsymbol{\rho}}) =\bar{\mathbf{m}}^{\ast}_M({\boldsymbol{\theta}}_0, {\boldsymbol{\rho}}_0)+  \frac{\partial \bar{\mathbf{m}}^{\ast}_M(\boldsymbol{\theta}_0, \boldsymbol{\rho}^+)}{\partial {\boldsymbol{\rho}^{\prime}}}(\hat{\boldsymbol{\rho}} - \boldsymbol{\rho}_0),$$ for some $\boldsymbol{\rho}^+$ lying between $\boldsymbol{\rho}_0$ and $\hat{\boldsymbol{\rho}}$.
By premultiplying the last equation by  $\frac{\partial\bar{\mathbf{m}}'_M(\hat{\boldsymbol{\theta}},  \hat{\boldsymbol{\rho}})}{\partial{\boldsymbol{\theta}}}\mathbf W_M$, we obtain:
$$
\frac{\partial\bar{\mathbf{m}}'_M(\hat{\boldsymbol{\theta}},  \hat{\boldsymbol{\rho}})}{\partial{\boldsymbol{\theta}}}\mathbf W_M\frac{\partial \bar{\mathbf{m}}^{\ast}_M(\boldsymbol{\theta}_0, \boldsymbol{\rho}^+)}{\partial {\boldsymbol{\rho}^{\prime}}}(\hat{\boldsymbol{\rho}} - \boldsymbol{\rho}_0) = \frac{\partial\bar{\mathbf{m}}'_M(\hat{\boldsymbol{\theta}},  \hat{\boldsymbol{\rho}})}{\partial{\boldsymbol{\theta}}}\mathbf W_M\left[\bar{\mathbf{m}}^{\ast}_M(\boldsymbol{\theta}_0, \hat{\boldsymbol{\rho}}) -\bar{\mathbf{m}}^{\ast}_M({\boldsymbol{\theta}}_0, {\boldsymbol{\rho}}_0) \right].
$$
By adding the previous equation to (\ref{eq:mvt1}) and rearranging the terms, we have:
\begin{eqnarray}\label{eq:mvt_p0}
&&\frac{\partial\bar{\mathbf{m}}'_M(\hat{\boldsymbol{\theta}},  \hat{\boldsymbol{\rho}})}{\partial{\boldsymbol{\theta}}}\mathbf W_M\frac{\partial \bar{\mathbf{m}}_M(\boldsymbol{\theta}^+, \hat{\boldsymbol{\rho}})}{\partial \boldsymbol{\theta}^{\prime}}(\hat{\boldsymbol{\theta}} - \boldsymbol{\theta}_0) + \frac{\partial\bar{\mathbf{m}}'_M(\hat{\boldsymbol{\theta}},  \hat{\boldsymbol{\rho}})}{\partial{\boldsymbol{\theta}}}\mathbf W_M\frac{\partial \bar{\mathbf{m}}^{\ast}_M(\boldsymbol{\theta}_0, \boldsymbol{\rho}^+)}{\partial {\boldsymbol{\rho}^{\prime}}}(\hat{\boldsymbol{\rho}} - \boldsymbol{\rho}_0)\nonumber\\
&=&-\frac{\partial\bar{\mathbf{m}}'_M(\hat{\boldsymbol{\theta}},  \hat{\boldsymbol{\rho}})}{\partial{\boldsymbol{\theta}}}\mathbf W_M\left[ 
\{ \bar{\mathbf{m}}_M({\boldsymbol{\theta}}_0,  \hat{\boldsymbol{\rho}}) -
\bar{\mathbf{m}}^{\ast}_M(\boldsymbol{\theta}_0, \hat{\boldsymbol{\rho}})\} - \{\bar{\mathbf{m}}_M(\boldsymbol{\theta}_0,  \boldsymbol{\rho}_0)-\bar{\mathbf{m}}^{\ast}_M({\boldsymbol{\theta}}_0, {\boldsymbol{\rho}}_0) \}\right]\nonumber\\ &&-
\frac{\partial\bar{\mathbf{m}}'_M(\hat{\boldsymbol{\theta}},  \hat{\boldsymbol{\rho}})}{\partial{\boldsymbol{\theta}}}\mathbf W_M\bar{\mathbf{m}}_M({\boldsymbol{\theta}}_0,  {\boldsymbol{\rho}}_0).
\end{eqnarray}

As for the empirical moment in Lemma \ref{lemma:uniformconvergence}, $\frac{\partial\bar{\mathbf{m}}_M(\boldsymbol{\theta},  \boldsymbol{\rho})}{\partial{\boldsymbol{\theta}^{\prime}}}$ converges uniformly in $\boldsymbol{\theta}$ and $\boldsymbol{\rho}$ because it can be written as an average of independent elements that are differentiable with bounded derivatives. Thus, since $\plim \hat{\boldsymbol{\theta}} = \plim \boldsymbol{\theta}^+ = \boldsymbol{\theta}_0$ and $\plim \hat{\boldsymbol{\rho}} = \boldsymbol{\rho}_0$, we have:
\begin{equation}\label{eq:plimdiffm}
    \plim \frac{\partial\bar{\mathbf{m}}_M(\hat{\boldsymbol{\theta}},  \hat{\boldsymbol{\rho}})}{\partial{\boldsymbol{\theta}^{\prime}}}=\plim \frac{\partial \bar{\mathbf{m}}_M(\boldsymbol{\theta}^+, \hat{\boldsymbol{\rho}})}{\partial \boldsymbol{\theta}^{\prime}} = \plim \frac{\partial \bar{\mathbf{m}}^{\ast}_M(\boldsymbol{\theta}_0, \boldsymbol{\rho}_0)}{\partial \boldsymbol{\theta}^{\prime}} \equiv \mathbf{H}_0.
\end{equation}
As usual, we also impose the following assumption so that, under Assumption \ref{as:Wm}, the matrix $\mathbf{H}_0'\mathbf{W}_0\mathbf{H}_0$ is not singular.
\begin{assumption}\label{as:reg_variance}
    The matrix $\mathbf{H}_0$ has full rank.
\end{assumption}
Equation (\ref{eq:mvt_p0}) implies that:
\begin{align*}
    &\sqrt{M}(\hat{\boldsymbol{\theta}} - \boldsymbol{\theta}_0)=-(\mathbf{H}_0'\mathbf{W}_0\mathbf{H}_0)^{-1}\mathbf{H}_0'\mathbf{W}_0\Big[ \sqrt{M}\bar{\mathbf{m}}_M({\boldsymbol{\theta}}_0,  {\boldsymbol{\rho}}_0)
+ \\
&\hspace{3cm}\frac{\partial \bar{\mathbf{m}}^{\ast}_M(\boldsymbol{\theta}_0, \boldsymbol{\rho}^+)}{\partial {\boldsymbol{\rho}^{\prime}}}\sqrt{M}(\hat{\boldsymbol{\rho}} - \boldsymbol{\rho}_0)\Big] + o_p(1)
\end{align*}
provided that the stochastic equicontinuity condition \autoref{cond:sequi} holds (see Lemma 4 in Online Appendix C).

Under Assumption \ref{as:partial}, $\sqrt{M}(\hat{\boldsymbol{\rho}} - \boldsymbol{\rho}_0)$ converges in distribution to a $N(\boldsymbol{0},\mathbf{V}_{\boldsymbol{\rho}})$, and
$$\plim \frac{\partial \bar{\mathbf{m}}^{\ast}_M(\boldsymbol{\theta}_0, \boldsymbol{\rho}^+)}{\partial {\boldsymbol{\rho}^{\prime}}}\equiv \boldsymbol{\Gamma}_0,
$$
exists by the uniform law or large numbers. Thus, $\frac{\partial \bar{\mathbf{m}}^{\ast}_M(\boldsymbol{\theta}_0, \boldsymbol{\rho}^+)}{\partial {\boldsymbol{\rho}^{\prime}}}\sqrt{M}(\hat{\boldsymbol{\rho}} - \boldsymbol{\rho}_0)$ converges in distribution to a $N(\mathbf{0},\boldsymbol{\Gamma}_0\mathbf{V}_{\boldsymbol{\rho}}\boldsymbol{\Gamma}_0^{\prime})$.

We now apply the Lyapunov CLT to $\sqrt{M}\bar{\mathbf{m}}_M({\boldsymbol{\theta}}_0,  {\boldsymbol{\rho}}_0)$ since it is a normalized sum of independent elements. However, as we need the joint asymptotic distribution of $\sqrt{M}\bar{\mathbf{m}}_M({\boldsymbol{\theta}}_0,  {\boldsymbol{\rho}}_0)$ and $\sqrt{M}(\hat{\boldsymbol{\rho}} - \boldsymbol{\rho}_0)$ to be normal, we apply the Lyapunov CLT conditional on $\sqrt{M}(\hat{\boldsymbol{\rho}} - \boldsymbol{\rho}_0)$. The Lyapunov condition (conditional $\sqrt{M}(\hat{\boldsymbol{\rho}} - \boldsymbol{\rho}_0)$) is verified by Assumption \ref{as:finitevariance}.\footnote{See for example, \cite{van2000asymptotic}, Section 23.4} Thus the asymptotic distribution of $\sqrt{M}\bar{\mathbf{m}}_M({\boldsymbol{\theta}}_0,  {\boldsymbol{\rho}}_0)$, conditional on $\sqrt{M}(\hat{\boldsymbol{\rho}} - \boldsymbol{\rho}_0)$ is normal, which implies that the joint asymptotic distribution of $\sqrt{M}\bar{\mathbf{m}}_M({\boldsymbol{\theta}}_0,  {\boldsymbol{\rho}}_0)$ and $\sqrt{M}(\hat{\boldsymbol{\rho}} - \boldsymbol{\rho}_0)$ is normal. Consequently $\sqrt{M}\bar{\mathbf{m}}_M({\boldsymbol{\theta}}_0,  {\boldsymbol{\rho}}_0)
+ \frac{\partial \bar{\mathbf{m}}^{\ast}_M(\boldsymbol{\theta}_0, \boldsymbol{\rho}^+)}{\partial {\boldsymbol{\rho}^{\prime}}}\sqrt{M}(\hat{\boldsymbol{\rho}} - \boldsymbol{\rho}_0)$ is asymptotically normally distributed. As a result, $\sqrt{M}(\hat{\boldsymbol{\theta}} - \boldsymbol{\theta}_0)$ is asymptotically normally distributed. 

Estimating the asymptotic variance of $\sqrt{M}(\hat{\boldsymbol{\theta}} - \boldsymbol{\theta}_0)$ requires an estimate of $\boldsymbol{\Gamma}_0$, which can be complex. In Online Appendix C.2, we present an approach to estimate this asymptotic variance without requiring an estimate of $\boldsymbol{\Gamma}_0$.

%% file: appendix_online.tex
\renewcommand{\partname}{}
\begin{center}
    {\LARGE Estimating Peer Effects Using Partial Network Data\\}
    
Vincent Boucher and Aristide Houndetoungan

September 2025
\end{center}
\part{Online Appendix}

\begin{singlespace}
\parttoc
\end{singlespace}

\input{OA_extensions}
\input{OA_THM1}
\input{appendix_proofs_others}

\input{appendix_newsimulations}
\newpage
\input{OA_Bayesian}

\newpage
\input{OA_application}

\cleardoublepage
\input{appendix_ARD}

\input{appendix_simulations}

%% file: OA_extensions.tex
\section{Extensions to other network formation processes and implications for survey design.}\label{sec:OA_extensions}

In Section 2.2, we assume that the network formation process is conditionally independent across links, i.e. $P(\mathbf{A}_m|\mathbf{X}_m)=\Pi_{ij} P(a_{ij,m}|\mathbf{X}_m)$. Moreover, our asymptotic framework imposes that groups are bounded (see Assumption 1). Finally, our Assumption 5 requires that the observable information about the network structure $\mathcal{A}_m$ is sufficient to point-identify the structural parameters and allow simulating network draws (see Definition 1).

We provide important examples that are covered by these assumptions. However, in some contexts, researchers may be interested in more general or flexible network formation processes. In this section, we discuss some important examples and show how they can or could be used within our framework. We conclude by summarizing some implications for survey design.

\subsection{Alternative network formation processes}

\subsubsection{Unobserved degree heterogeneity}\label{sec:graham}

\citeOA{graham2017econometric} presents a network formation process for \emph{undirected} networks in which individuals are characterized by unobserved degree heterogeneity. As mentioned, all of our results hold for undirected networks with the appropriate notation changes. The network formation process in \citeOA{graham2017econometric} is as follows:
$$P(a_{ij,m}=1|\mathbf{X}_m, \boldsymbol{\nu}_m)=\frac{\exp\{\mathbf{w}_{ij,m}\boldsymbol{\rho}+\nu_{i,m}+\nu_{j,m}\}}{1+\exp\{\mathbf{w}_{ij,m}\boldsymbol{\rho}+\nu_{i,m}+\nu_{j,m}\}}$$
for all pairs $ij:i<j$, where $\boldsymbol{\nu}_m = (\nu_{1,m}, \dots, \nu_{N_m,m})^{\prime}$.

Note that the model is conditionally independent across links since $P(\mathbf{A}_m|\mathbf{X}_m, \boldsymbol{\nu}_m)=\Pi_{ij:i<j}P(a_{ij,m}=1|\mathbf{X}_m, \boldsymbol{\nu}_m)$. However, since groups are bounded (our Assumption 1), we cannot consistently estimate $\nu_{i,m}$ since for each individual, we only observe $n_m-1$ pairs. (Our Assumption 5 fails.) While \citeOA{graham2017econometric} provides a way to estimate $\nu_{i,m}$, it only works when researchers face a single large network, which is ruled out by our Assumption 1.

However, \citeOA{graham2017econometric} presents an estimator (the ``Tetrad Logit'') that allows for the consistent estimation of $\boldsymbol{\rho}$ conditional on the degree sequence (i.e., for each individual, the number of links they have). Specifically, he shows that:
\begin{equation}\label{eq:graham}
P(\mathbf{A}_m|\mathbf{X}_m,\bar{\mathbf{a}}_m)=\frac{\exp\lbrace\sum_{ij:i<j} a_{ij}\mathbf{w}_{ij}\boldsymbol{\rho} \rbrace}{\sum_{\mathbf{B}:\bar{\mathbf{b}}=\bar{\mathbf{a}}} \exp\lbrace\sum_{ij:i<j} b_{ij}\mathbf{w}_{ij}\boldsymbol{\rho}\rbrace},    
\end{equation}
where $\bar{\mathbf{a}}_m=\mathbf{A}_m\mathbf{1}_m$ and $\bar{\mathbf{b}}_m=\mathbf{B}_m\mathbf{1}_m$ are the degree sequences of the adjacency matrices $\mathbf{A}_m$ and $\mathbf{B}_m$. Thus, if one observes the degree sequence (i.e. $\bar{\mathbf{a}}_m\in\mathcal{A}_m$) and $a_{ij}$ in a saturated sample (see below), they can potentially estimate (\ref{eq:graham}).

Theorem 1 in \citeOA{graham2017econometric} shows that it is possible provided that one observes a ``saturated'' sample of the population. Specifically, he assumes that the network is fully observed for a \emph{subset} of individuals, i.e. $\hat{\mathbf{A}}_{\hat{N}_m}\equiv\{a_{ij,m}\}_{j<i,i,j\in \hat{N}_m}$, where $\hat{N}_m\subset N_m$. Note that from $\hat{\mathbf{A}}_{\hat{N}_m}$, we can recover $\hat{\bar{\mathbf{a}}}_{\hat{N}_m}$, the degree sequence of sampled agents within the observed saturated sample (typically different from the true degree sequence).\footnote{Theorem 1 in \cite{graham2017econometric} also requires weak conditions on the asymptotic degree sequence, see Assumption 4 in \citeOA{graham2017econometric}, as well as standard compacity and support conditions.}

So, if $\mathcal{A}_m\ni\hat{{\mathbf{A}}}_{\hat{N}_m}$ for all $m$, then $\boldsymbol{\rho}$ can be consistently estimated.\footnote{Here, consistency and asymptotic normality is simpler to achieve than in \cite{graham2017econometric} since our asymptotic framework assumes that $M$ grows to infinity (our Assumption 1). We therefore do not require CLT for $U$-statistic since a standard CLT for independent and non-identically distributed data, such as Lyapunov CLT, is sufficient} In other words, this means that if, for a \emph{subset} of individuals, the network structured among them is observed without error, then our Assumption 5 holds for the network formation process in (\ref{eq:graham}).

However, this does not imply that we can compute a consistent estimator of the distribution of the true network (our Definition 1). This is because we need to predict the links of the individuals that are \emph{not} in the saturated sample: $i\in N_m\setminus\hat{N}_m$ for all $m$. To do so, we also need the individuals' true degree sequence $\bar{\mathbf{a}}_m$. Thus, we define $\mathcal{A}_m=\{\hat{{\mathbf{A}}}_{\hat{N}_m},{\bar{\mathbf{a}}}_m\}$ for all $m$.



\subsubsection{Boucher and Mourifi\'e (2017)}\label{sec:bouchermourifie}

\citeOA{boucher2017my} present a network formation model for \emph{undirected} networks. Their network formation model, which is a special case of an exponential random graph model (ERGM) which is \emph{not} conditionally independent across links.

Their model is as follows:
$$P(a_{ij,m}=1|\mathbf{X}_m,\mathbf{A}_{m,-ij})=\frac{\exp\{\mathbf{w}_{ij,m}\tilde{\boldsymbol{\rho}} + (n_{i,m}+n_{j,m})\rho_1+\psi(d(\tilde{\mathbf{x}}_{i,m},\tilde{\mathbf{x}}_{j,m}))\rho_2 \}}{1+\exp\{\mathbf{w}_{ij,m}\tilde{\boldsymbol{\rho}} + (n_{i,m}+n_{j,m})\rho_1+\psi(d(\tilde{\mathbf{x}}_{i,m},\tilde{\mathbf{x}}_{j,m}))\rho_2 \}}$$
for all pairs $ij:i<j$, where $n_{i,m}$ and $n_{j,m}$ represent the number of links that $i$ and $j$ have and is a function of $\mathbf{A}_{m,-ij}$. Importantly, $\tilde{\mathbf{x}}_{i,m}$ and $\tilde{\mathbf{x}}_{j,m}$ are non-stochastic ``positions'' of $i$ and $j$ on an underlying Euclidean space (e.g., geographical distance), and $d(\tilde{\mathbf{x}}_{i,m},\tilde{\mathbf{x}}_{j,m})$ is a distance (and $\psi$ is some increasing function, see below). Under the restriction that $d(\tilde{\mathbf{x}}_{i,m},\tilde{\mathbf{x}}_{j,m})\geq d_0>0$ for all $i$ and $j$ and that $\rho_2<\underline{\rho}<0$, \citeOA{boucher2017my} show that $\boldsymbol{\rho}=[\tilde{\boldsymbol{\rho}},\rho_1,\rho_2]$ is consistently estimated by a simple pseudo-logistic regression.

While they do not consider bounded groups, their setup is compatible with our framework. To do so, however, we need to adapt our framework and ensure that individuals (and groups) are also drawn on some non-stochastic Euclidean space (e.g., geographical location). Then, their estimator is valid if we assume that:
\begin{enumerate}
    \item Groups are drawn in a way that the distance (on the non-stochastic space) between each pair of individuals within the group is bounded below. For example, in the case of small villages and geographical location, this implies a minimal physical distance between any two individuals' homes.
    \item The distance between any two individuals in the same group is bounded above (e.g., villages are bounded by geography).
    \item $\psi(d)=d$ when $d\leq \bar{d}$, while $\psi(d)=\infty$ if $d>\bar{d}$. Essentially, this ensures that no link can be created between individuals of different groups. For example, villages are geographically far enough so that links between any two villages are not valuable. 
\end{enumerate}

Here, while $\boldsymbol{\rho}$ can be estimated by observing, for a random sample of pairs, their linking status and the number of links they have, this is not sufficient in order to predict the linking status of all pairs. To do so, we need to observe the number of links for all individuals (the degree sequence). Let $\hat{N}^2_m\subset \{ij\in N_m\times N_m:j<i\}$ be some non-empty subset of pairs of individuals, we need $\mathcal{A}_m=\{\{a_{ij,m}\}_{i,j\in \hat{N}^2_m},\bar{a}_m\}$, where $\bar{a}_m=\mathbf{A}_m\mathbf{1}_m$.

\subsubsection{Other ERGM}\label{sec:app_ergm}

Exponential Random Graph models (ERGM) are such that:\footnote{We focus on cases for which the term inside the exponential is linear in $\boldsymbol{\rho}$ for simplicity.}
$$
P(\mathbf{A}_m|\mathbf{X}_m)=\frac{\exp\{ \mathbf{q}(\mathbf{A}_m,\mathbf{X}_m)\boldsymbol{\rho}\}}{\sum_{\mathbf{B}_m}\exp\{ \mathbf{q}(\mathbf{B}_m,\mathbf{X}_m)\boldsymbol{\rho}\}},
$$
where $\mathbf{q}$ is a known function and the sum in the denominator is over all the possible network structures $\mathbf{B}_m$ for the group $m$.  The Microfoundations for ERGM can be found in \cite{mele2017} and \cite{hsieh2019specification}. Since ERGM are from the exponential family, $\mathbf{q}(\mathbf{A}_m,\mathbf{X}_m)$ are the sufficient statistics for $\boldsymbol{\rho}$. This means that consistent estimation of $\boldsymbol{\rho}$ requires consistent estimation of these sufficient statistics. This in turns implies that the sampling process that generates $\mathcal{A}_m$ must allow for this. We give two simple examples below.

\subsubsection{ERGM: Reciprocal links}

This simplest possible ERGM is such that:
$$
P(\mathbf{A}_m|\mathbf{X}_m)\propto \exp\{\sum_{ij} (a_{ij}w_{ij}\tilde{\boldsymbol{\rho}} +  \rho_1 a_{ij}a_{ji})\},
$$
where $\boldsymbol{\rho} = [\rho_1, ~\tilde{\boldsymbol{\rho}}]$.
When $\rho_1=0$, the model reduces to the baseline model in Equation (\ref{eq:gennetfor}). Here, $\rho_1>0$ implies that reciprocal links (when $i$ is linked to $j$ and $j$ is linked to $i$) are more likely than what would be expected from a model with conditionally independent links.

Note that we can easily compute the joint distribution of $(a_{ij,m},a_{ji,m})$ as follows:
\allowdisplaybreaks
\begin{eqnarray*}
    P((a_{ij,m}=1,a_{ji,m}=1)|\mathbf{X}_m)&\propto&\exp\{w_{ij}\tilde{\boldsymbol{\rho}} + w_{ji}\tilde{\boldsymbol{\rho}} +2\rho_1 \}\\
    P((a_{ij,m}=1,a_{ji,m}=0)|\mathbf{X}_m)&\propto&\exp\{w_{ij}\tilde{\boldsymbol{\rho}}\}\\
    P((a_{ij,m}=0,a_{ji,m}=1)|\mathbf{X}_m)&\propto&\exp\{w_{ji}\tilde{\boldsymbol{\rho}}\}\\
    P((a_{ij,m}=0,a_{ji,m}=0)|\mathbf{X}_m)&\propto&1,
\end{eqnarray*}

We can then rewrite: $P(\mathbf{A}_m|\mathbf{X}_m)=\Pi_{ij:i<j}P((a_{ij,m},a_{ji,m}|\mathbf{X}_m)$. Then, the estimation of such a model requires sampling \emph{pairs of individuals}, and observing their linking status without error: $\mathcal{A}_m=\{ (a_{ij,m},a_{ji,m}), i,j\in \hat{N}^2\subseteq N^2 \}$.
Note that this information is also sufficient to compute a consistent estimator of the distribution of the true network (our Definition 1) using Bayes' rule ($\kappa(\mathcal{A}_m)=\mathcal{A}_m$). That is, for unsampled pairs ($ij\notin \hat{N}^2_m$), the observed linking status for sampled pairs is uninformative and links are predicted using the pair-level joint distribution given $\hat{\boldsymbol{\rho}}$. For sampled pairs, the distribution is degenerated since their linking status is observed.

\subsubsection{ERGM: Transitive triads}

A typical feature of the data that is hard to replicate using the model in Equation (\ref{eq:gennetfor}) is the fraction of transitive triads. That is, if $i$ is linked to $j$ and $j$ is linked to $k$, then the probability that $i$ and $k$ are linked is higher than what would be predicted by a model with conditionally independent links. Consider the following ERGM:
$$
P(\mathbf{A}_m|\mathbf{X}_m)\propto \exp\{\sum_{ij} (a_{ij}w_{ij}\tilde{\boldsymbol{\rho}} +  \rho_1 a_{ij}a_{ji} + \rho_2\sum_k a_{ij}a_{jk}a_{ki})\},
$$
which now includes the number of directed triangles (cycles of length 3). If $\rho_2>0$, then network configurations in which $i$ is linked to $j$, $j$ is linked to $k$, and $k$ is linked to $i$ are more likely (everything else equal).

Here, the sufficient statistics are: the fraction of links (given a set of observed pair characteristics $\mathbf{w}_{ij}$), the fraction of reciprocal links, and the fraction of closed directed triangles. This \emph{requires} sampling triads of individuals, which substantially complicates the sampling design. We are not aware of any such application. We also note that even with the consistent estimation of sufficient statistics, the estimation of $\boldsymbol{\rho}$ is not straightforward and computationally intensive. These considerations are left for future research.

\subsubsection{Aggregated relational data}\label{sec:ard_netext}

Aggregated relational data (ARD) are obtained from survey questions such as, ``How many friends with trait `X' do you have?'' Here, $\mathcal{A}$ can be represented by an $N\times K$ matrix of integer values, where $K$ is the number of traits that individuals were asked about.

Building on \citeOA{mccormick2015latent}, \citeOA{breza2017using} proposed a novel approach for the estimation of network formation models using only ARD. They assume:
\begin{equation}\label{eq:exard}
P(a_{ij,m}=1)=\frac{\exp\{\nu_i+\nu_j+\zeta \mathbf{z}_i'\mathbf{z}_j\}}{
1+\exp\{\nu_i+\nu_j+\zeta \mathbf{z}_i'\mathbf{z}_j\}}.    
\end{equation}
Here, $\boldsymbol{\rho}=[\{\nu_i,\mathbf{z}_i\}_{i},\zeta]$ is not observed by the econometrician. The parameters $\nu_i$ and $\nu_j$ can be interpreted as $i$ and $j$'s propensities to create links, irrespective of the identity of the other individual involved. The other component, $\zeta \mathbf{z}_i'\mathbf{z}_j$, is meant to capture homophily (like attracts like) on an abstract latent space (e.g., \citeOA{hoff2002latent}). This model differs from the ones presented above and in the text in two fundamental ways.

First, ARD does not provide information on any specific links; \footnote{That is, unless ARD includes the degree distribution with some individuals reporting having no links at all.} therefore, one could disregard the ARD information (i.e. $\kappa(\mathcal{A}_m)=\kappa_0$ for all $\mathcal{A}_m$ and $m$) and define the predicted distribution estimator of the true network as: $$\hat{P}(a_{ij,m}=1|\hat{\boldsymbol{\rho}},\mathbf{X}_m)=\frac{\exp\{\hat{\nu}_i+\hat{\nu}_j+\hat \zeta \hat{\mathbf{z}}_i'\hat{\mathbf{z}}_j\}}{
1+\exp\{\hat{\nu}_i+\hat{\nu}_j+\hat \zeta \hat{\mathbf{z}}_i'\hat{\mathbf{z}}_j\}},$$
where $\hat \nu_i$, $\hat{\mathbf{z}}_i$, and $\hat \zeta$ are the estimators (e.g., posterior means) of $\nu_i$, ${\mathbf{z}}_i$, and $\zeta$, respectively.

Second, (and perhaps more importantly) consistent estimation of $\boldsymbol{\rho}$ is only possible as the group size $N_m$ goes to infinity \citepOA{breza2019consistently}, which contradicts our Assumption \ref{as:manymarkets}. Thus, Assumption \ref{as:partial} does not hold. In the online Appendix \ref{sec:ard}, we show that our SGMM estimator (see Section \ref{sec:iv}) still performs well in finite samples for groups of moderate sizes. 

The failure of Assumption \ref{as:partial} implies that, formally, our SGMM estimator cannot be applied to ARD data and one must therefore rely on our Bayesian estimator.

\subsection{Implications for survey design}

As seen throughout the paper, the researchers' ability to obtain a consistent estimator of the network formation process (see Definition \ref{def:estimator} and Assumption \ref{as:partial}) depends on the available information about the network (i.e. $\mathcal{A}$) as well as on the flexibility of the network formation process. Highly structured network formation processes are powerful and require less information. More flexible network formation processes (such as the ones discussed above) require more data. In particular, the type of network formation model chosen will in general affect the data collection. Here, we would like to mention some general messages:
\begin{enumerate}
    \item Asking individuals about the number of links they have allows estimating much more flexible models and should be considered. In particular, it is necessary to the estimation of models in \cite{graham2017econometric} and \cite{boucher2017my}, and strongly improves the estimation of the model in \cite{breza2017using}.
    \item Whenever possible, surveys should try obtain $\mathbf{GX}$ and $\mathbf{Gy}$ from the respondents. This could be done by asking questions such as ``Which fraction of your friends are X''. When $\mathbf{GX}$ and $\mathbf{Gy}$ are (at least partly) observed, the estimation of the linear in means is less reliant on the simulated draws.\footnote{See Proposition \ref{prop:all_observed} of the Online Appendix \ref{sec:appendix_otherresults}.} Moreover, coupled with information about the number of links individuals have, these variables act as ARD data.
    \item When researchers only require a random sample of pairs, \cite{conley2010learning} propose the following sampling strategy: For each respondent, generate a random sample of names from the within-group population and ask the respondent about their linking status with each of the individuals in the sample. 
\end{enumerate}

%% file: OA_THM1.tex
\section{Proof of Theorem 1: Additional Material}\label{sec:OA_thm1}
\begin{customthm}{\ref{lemma:continuity}}[Differentiability]
    $\mathbb{E}[\mathbf{m}_{m}(\boldsymbol{\theta},\boldsymbol{\rho})]$ is continuously differentiable in $(\boldsymbol{\theta}, \boldsymbol{\rho})$.
\end{customthm}

\begin{proof}
Since $\mathbf{m}_{m}(\boldsymbol{\theta},\boldsymbol{\rho})$ is continuously differentiable in $\boldsymbol{\theta}$ and absolutely integrable, along with its derivative with respect to $\boldsymbol{\theta}$, it follows that $\mathbb{E}[\mathbf{m}_{m}(\boldsymbol{\theta},\boldsymbol{\rho})]$ is continuously differentiable in $\boldsymbol{\theta}$ by the Leibniz integral rule. However, $\mathbf{m}_{m}(\boldsymbol{\theta},\boldsymbol{\rho})$ is continuously differentiable in $\boldsymbol{\rho}$ only for \emph{almost all} $\boldsymbol{\rho}$. We now show that $\mathbb{E}[\mathbf{m}_{m}(\boldsymbol{\theta},\boldsymbol{\rho})]$ is continuously differentiable for all $\boldsymbol{\rho}$.

Consider $\mathbf{B}_m\dot{\mathbf{G}}_m(\boldsymbol{\rho}) = \mathbf{B}_m f(\dot{\mathbf{A}}_m(\boldsymbol{\rho}))$ for some (conformable) matrix $\mathbf{B}_m$. We have:
\begin{equation}\label{eq:EB}
    \hat{\mathbb{E}}_m (\mathbf{B}_m\dot{\mathbf{G}}_m(\boldsymbol{\rho})) = \sum_{\hat{\mathbf{A}}_m} \mathbf{B}_m f(\hat{\mathbf{A}}_m) P (\mathbf{A}_m(\boldsymbol{\rho}) = \hat{\mathbf{A}}_m \mid \mathbf{X}_m, \kappa(\mathcal{A}_m)),
\end{equation}
where the sum is taken over all the possible network configurations $\hat{\mathbf{A}}_m$, and where \break$P(\mathbf{A}_m(\boldsymbol{\rho})=\hat{\mathbf{A}}_m|\mathbf{X}_m,\kappa(\mathcal{A}_m))=\Pi_{ij}P(a_{ij,m}(\boldsymbol{\rho})=\hat{a}_{ij,m}|\mathbf{X}_m,\kappa(\mathcal{A}_m))$, as defined in Equation (\ref{eq:gennetfor}). Thus $\hat{\mathbb{E}}_m (\mathbf{B}_m\dot{\mathbf{G}}_m(\boldsymbol{\rho}))$ is continuously differentiable in $\boldsymbol{\rho}$ by Assumption \ref{as:derivatives_P}.
By adapting this argument for the appropriate definition of the matrix $\mathbf{B}_m$, and for $\hat{\mathbb{E}}_m$, and $\mathbb{E}^{(0)}_m$, this shows that $\mathbb{E}[\mathbf{m}_{m,rst}(\boldsymbol{\theta},\boldsymbol{\rho})|\mathbf{X}_m,\kappa(\mathcal{A}_m)]$ is continuously differentiable in $(\boldsymbol{\theta}, \boldsymbol{\rho})$. As a result, $\mathbb{E}[\mathbf{m}_{m}(\boldsymbol{\theta},\boldsymbol{\rho})]$ is continuously differentiable in $(\boldsymbol{\theta}, \boldsymbol{\rho})$ by the Leibniz integral rule.\end{proof}

\begin{customthm}{\ref{lemma:uniformconvergence}}[Uniform convergence] We establish the following results.
\begin{enumerate}[label=(\alph*), nosep]
    \item $\mathbb{E}[\bar{\mathbf{m}}_M(\boldsymbol{\theta},\hat{\boldsymbol{\rho}})] - \mathbb{E}[\bar{\mathbf{m}}_M(\boldsymbol{\theta},{\boldsymbol{\rho}}_0)]$ converges uniformly to $\mathbf{0}$ in $\boldsymbol{\theta}$ as $M \to \infty$.
    \item $\bar{\mathbf{m}}_M(\boldsymbol{\theta},\hat{\boldsymbol{\rho}}) - \mathbb{E}[\bar{\mathbf{m}}_M(\boldsymbol{\theta},\hat{\boldsymbol{\rho}})]$ converges uniformly in probability to $\mathbf{0}$ in $\boldsymbol{\theta}$ as $M \to \infty$.
    \item $\bar{\mathbf{m}}_M(\boldsymbol{\theta},\hat{\boldsymbol{\rho}}) - \mathbb{E}[\bar{\mathbf{m}}_M(\boldsymbol{\theta},{\boldsymbol{\rho}}_0)]$ converges uniformly in probability to $\mathbf{0}$ in $\boldsymbol{\theta}$ as $M \to \infty$. 
\end{enumerate}
    
\end{customthm}

\begin{proof} 
By applying the mean value theorem to the difference $\hat{P}(\mathbf{A}_m \mid \hat{\boldsymbol{\rho}}, \mathbf{X}_m, \kappa(\mathcal{A}_m)) - P(\mathbf{A}_m \mid \boldsymbol{\rho}_0, \mathbf{X}_m, \kappa(\mathcal{A}_m))$, one can see that the consistency of $\hat{\boldsymbol{\rho}}$ (Assumption~\ref{as:partial}) and the differentiability of $\hat{P}(\mathbf{A}_m \mid \boldsymbol{\rho}, \mathbf{X}_m, \kappa(\mathcal{A}_m))$ in $\boldsymbol{\rho}$ (Assumption~\ref{as:derivatives_P}) imply Definition~\ref{def:estimator}.\footnote{The proof remains valid for any estimator $\hat{\boldsymbol{\rho}}$ that satisfies Definition~\ref{def:estimator}, even if Assumptions~\ref{as:partial} and~\ref{as:derivatives_P} do not hold.} Under Definition~\ref{def:estimator}, the difference $\mathbb{E}[\bar{\mathbf{m}}_M(\boldsymbol{\theta}, \hat{\boldsymbol{\rho}}) \mid \hat{\boldsymbol{\rho}}, \mathbf{X}_m, \kappa(\mathcal{A}_m)] - \mathbb{E}[\bar{\mathbf{m}}_M(\boldsymbol{\theta}, \boldsymbol{\rho}_0) \mid \mathbf{X}_m, \kappa(\mathcal{A}_m)]$ converges pointwise to $\mathbf{0}$ for each $\boldsymbol{\theta}$. Since this difference is bounded (by Assumptions~\ref{as:manymarkets} and~\ref{as:nonsingular}), the dominated convergence theorem implies that its expectation also converges to $\mathbf{0}$; that is, $\mathbb{E}[\bar{\mathbf{m}}_M(\boldsymbol{\theta},\hat{\boldsymbol{\rho}})] - \mathbb{E}[\bar{\mathbf{m}}_M(\boldsymbol{\theta},{\boldsymbol{\rho}}_0)]$ converges to $\mathbf{0}$ pointwise for each $\boldsymbol{\theta}$. To establish uniform convergence, we apply Lemma 2.9 in  \cite{newey1994large}. The required conditions are satisfied: (i) the space of $\boldsymbol\theta$ is compact (Assumption \ref{as:regularity}), (ii) $\mathbb{E}[\bar{\mathbf{m}}_M(\boldsymbol{\theta},\hat{\boldsymbol{\rho}})] - \mathbb{E}[\bar{\mathbf{m}}_M(\boldsymbol{\theta},{\boldsymbol{\rho}}_0)]$ converges to $\mathbf{0}$ for all $\boldsymbol\theta$, and (iii) the derivative of $\mathbb{E}[\bar{\mathbf{m}}_M(\boldsymbol{\theta},\hat{\boldsymbol{\rho}})] - \mathbb{E}[\bar{\mathbf{m}}_M(\boldsymbol{\theta},{\boldsymbol{\rho}}_0)]$ with respect to $\boldsymbol{\theta}$ is bounded in probability (Assumptions \ref{as:manymarkets} and \ref{as:nonsingular}). As a result, $\mathbb{E}[\bar{\mathbf{m}}_M(\boldsymbol{\theta},\hat{\boldsymbol{\rho}})] - \mathbb{E}[\bar{\mathbf{m}}_M(\boldsymbol{\theta},{\boldsymbol{\rho}}_0)]$ uniformly converges to zero in $\boldsymbol{\theta}$.
This completes the proof of Statement \ref{lemma:uniformconvergence:A}.

For Statement \ref{lemma:uniformconvergence:B}, we first establish pointwise convergence by showing that the variance of $\bar{\mathbf{m}}_M(\boldsymbol{\theta},\hat{\boldsymbol{\rho}})$ vanishes asymptotically. By the law of iterated variances, we have: 
\begingroup
\allowdisplaybreaks
\begin{align}
    &\mathbb V(\bar{\mathbf{m}}_M(\boldsymbol{\theta}, \hat{\boldsymbol{\rho}})) = \mathbb V\left\{\mathbb E\big(\bar{\mathbf{m}}_M(\boldsymbol{\theta}, \hat{\boldsymbol{\rho}})|\mathbf X, \kappa(\mathcal{A}), \hat{\boldsymbol{\rho}}\big)\right\} + \mathbb E\left\{\mathbb V\big(\bar{\mathbf{m}}_M(\boldsymbol{\theta}, \hat{\boldsymbol{\rho}})|\mathbf X, \kappa(\mathcal{A}), \hat{\boldsymbol{\rho}}\big)\right\},\nonumber\\
    &\lim \mathbb V(\bar{\mathbf{m}}_M(\boldsymbol{\theta}, \hat{\boldsymbol{\rho}})) = \mathbb V\left\{\plim \mathbb E\big(\bar{\mathbf{m}}_M(\boldsymbol{\theta}, \hat{\boldsymbol{\rho}})|\mathbf X, \kappa(\mathcal{A}), \hat{\boldsymbol{\rho}}\big)\right\}~+ \nonumber\\&\quad\quad\quad\quad \mathbb E\left\{\plim \mathbb V\big(\bar{\mathbf{m}}_M(\boldsymbol{\theta}, \hat{\boldsymbol{\rho}})|\mathbf X, \kappa(\mathcal{A}), \hat{\boldsymbol{\rho}}\big)\right\},\nonumber\\
    \begin{split}
    &\lim \mathbb V(\bar{\mathbf{m}}_M(\boldsymbol{\theta}, \hat{\boldsymbol{\rho}})) = \mathbb V\left\{\plim \mathbb E\big(\bar{\mathbf{m}}_M(\boldsymbol{\theta}, \hat{\boldsymbol{\rho}})|\mathbf X, \kappa(\mathcal{A}), \hat{\boldsymbol{\rho}}\big)\right\} ~+ \\&\quad\quad\quad\quad\mathbb E\left\{\plim \dfrac{1}{M^2}\sum_{m = 1}^M\mathbb V\big(\mathbf{m}_m(\boldsymbol{\theta}, \hat{\boldsymbol{\rho}})|\mathbf X_m, \mathcal{A}_m, \hat{\boldsymbol{\rho}}\big)\right\}.
    \end{split} \label{eq:Vmbar}
\end{align}
\endgroup
The second equality holds by the dominated convergence theorem.\footnote{Specifically, $\mathbb{E}(\bar{\mathbf{m}}_M(\boldsymbol{\theta}, \hat{\boldsymbol{\rho}}) \mid \mathbf{X}, \kappa(\mathcal{A}), \hat{\boldsymbol{\rho}})$ and $\mathbb{V}(\bar{\mathbf{m}}_M(\boldsymbol{\theta}, \hat{\boldsymbol{\rho}}) \mid \mathbf{X}, \kappa(\mathcal{A}), \hat{\boldsymbol{\rho}})$ are bounded by Assumptions \ref{as:manymarkets} and \ref{as:nonsingular}, as well as by the fact that the conditional variance of $\boldsymbol{\varepsilon}_m$ is uniformly bounded (Assumption \ref{as:finitevariance}). Consequently, we can interchange the expectation and probability limit operators.
} Equation \eqref{eq:Vmbar} holds by the fact that \( \mathbf{m}_m(\boldsymbol{\theta}, \hat{\boldsymbol{\rho}}) \) are independent across \( m \), conditional on \( \mathbf{X}_m, \mathcal{A}_m, \hat{\boldsymbol{\rho}} \). 

From Equation \eqref{eq:Vmbar}, it is thus sufficient to show that \( \plim \mathbb{E}(\bar{\mathbf{m}}_M(\boldsymbol{\theta}, \hat{\boldsymbol{\rho}}) | \mathbf{X}, \kappa(\mathcal{A}), \hat{\boldsymbol{\rho}}) \) is nonstochastic and that \( \plim \dfrac{1}{M^2}\sum_{m = 1}^M \mathbb{V}(\mathbf{m}_m(\boldsymbol{\theta}, \hat{\boldsymbol{\rho}}) | \mathbf{X}_m, \mathcal{A}_m, \hat{\boldsymbol{\rho}}) = \mathbf{0} \).
First, the variance of \( \mathbb{E}[\bar{\mathbf{m}}_M(\boldsymbol{\theta},{\boldsymbol{\rho}}) | \mathbf{X}_m, \kappa(\mathcal{A}_m)] \) vanishes asymptotically by being the variance of an average of independent elements. Therefore, it converges in $\mathcal{L}^2$ and, thus, in probability to its expectation, \( \mathbb{E}[\bar{\mathbf{m}}_M(\boldsymbol{\theta},{\boldsymbol{\rho}})] \), which is nonstochastic. Consequently, \( \plim \mathbb{E}(\bar{\mathbf{m}}_M(\boldsymbol{\theta}, \hat{\boldsymbol{\rho}}) | \mathbf{X}, \kappa(\mathcal{A}), \hat{\boldsymbol{\rho}}) \) is also nonstochastic and, thus,  $\mathbb V\{\plim \mathbb E(\bar{\mathbf{m}}_M(\boldsymbol{\theta}, \hat{\boldsymbol{\rho}})|\mathbf X, \kappa(\mathcal{A}), \hat{\boldsymbol{\rho}})\} = \mathbf 0$.

Second,  $\mathbb V(\mathbf{m}_m(\boldsymbol{\theta}, \hat{\boldsymbol{\rho}})|\mathbf X_m, \mathcal{A}_m, \hat{\boldsymbol{\rho}}) < \infty$ because $\mathbf{m}_m(\boldsymbol{\theta}, \hat{\boldsymbol{\rho}})$ is linear in $\boldsymbol\varepsilon_m$ which has a bounded variance (Assumption \ref{as:finitevariance}). Thus, $\plim \dfrac{1}{M^2}\sum_{m = 1}^M\mathbb V(\mathbf{m}_m(\boldsymbol{\theta}, \hat{\boldsymbol{\rho}})|\mathbf X_m, \mathcal{A}_m, \hat{\boldsymbol{\rho}}) = 0$ and $\mathbb E\{\plim \dfrac{1}{M^2}\sum_{m = 1}^M\mathbb V(\mathbf{m}_m(\boldsymbol{\theta}, \hat{\boldsymbol{\rho}})|\mathbf X_m, \mathcal{A}_m, \hat{\boldsymbol{\rho}})\} = \mathbf 0$. 

Consequently, $\bar{\mathbf{m}}_M(\boldsymbol{\theta},\hat{\boldsymbol{\rho}}) - \mathbb{E}[\bar{\mathbf{m}}_M(\boldsymbol{\theta},{\boldsymbol{\rho}}_0)]$ converges in probability to $\mathbf 0$. The convergence is uniform in $\boldsymbol{\theta}$ from Lemma 2.9 of \citeOA{newey1994large} as in Statements \ref{lemma:uniformconvergence:A}. This completes the proof of Statement \ref{lemma:uniformconvergence:B} and, consequently, the lemma, given that Statements \ref{lemma:uniformconvergence:A} and \ref{lemma:uniformconvergence:B} together imply Statement \ref{lemma:uniformconvergence:C}.
\end{proof}

\begin{lemma}\label{sec:stochasticequi}
    The stochastic equicontinuity condition \autoref{cond:sequi} is verified.
\end{lemma}
\begin{proof}
As $\boldsymbol{\theta}_0$ in \autoref{cond:sequi} is fixed, we ignore it in our notations and define $\tilde{\mathbf{m}}_{m}({\boldsymbol{\rho}}) = \mathbf{m}_{m}({\boldsymbol{\theta}}_0,{\boldsymbol{\rho}})$. We follow \citeOA{andrews1994empirical} and define
$$\nu_M({\boldsymbol{\rho}})=\frac{1}{\sqrt{M}}\sum_m [\tilde{\mathbf{m}}_m({\boldsymbol{\rho}})-\mathbb{E}(\tilde{\mathbf{m}}_m({\boldsymbol{\rho}}))],$$
so that conditions \autoref{cond:sequi} is equivalent to $\nu_M(\hat{\boldsymbol{\rho}}) - \nu_M({\boldsymbol{\rho}}_0)=o_p(1)$.
Consider the following pseudo-metric, for any dimension $k$ of the moment function
$$
d_k({\boldsymbol{\rho}}_1,{\boldsymbol{\rho}}_2)=\sup_{m}(\mathbb{E}[\tilde{\mathbf{m}}_{m,[k]}({\boldsymbol{\rho}}_1)-\tilde{\mathbf{m}}_{m,[k]}({\boldsymbol{\rho}}_2)]^2)^{1/2}.
$$
We say that the process $\nu_M$ is \emph{stochastically equicontinuous} if, $\forall \epsilon>0$, $\exists \delta>0$ such that
$$
\plim\sup_{d_k({\boldsymbol{\rho}}_1,{\boldsymbol{\rho}}_2)<\delta} | \nu_{M,[k]}({\boldsymbol{\rho}}_1)-\nu_{M,[k]}({\boldsymbol{\rho}}_2) |<\epsilon,
$$
for each dimension $[k]$.
To see that stochastic equicontinuity implies \autoref{cond:sequi}, note that, for any $\epsilon > 0$:
\begingroup
\allowdisplaybreaks
\begin{eqnarray*}
&&\lim P(|\nu_{M,[k]}(\hat{\boldsymbol{\rho}})-\nu_{M,[k]}({\boldsymbol{\rho}}_0)|>\epsilon) \\&\leq& \lim P(|\nu_{M,[k]}(\hat{\boldsymbol{\rho}})-\nu_{M,[k]}({\boldsymbol{\rho}}_0)|>\epsilon, d_k(\hat{\boldsymbol{\rho}},{\boldsymbol{\rho}}_0)\leq\delta)  + \lim P( d_k(\hat{\boldsymbol{\rho}},{\boldsymbol{\rho}}_0)>\delta) \\
&\leq& \lim P\left( \sup_{d(\boldsymbol{\rho}_1,\boldsymbol{\rho}_2)<\delta} | \nu_{M,[k]}(\hat{\boldsymbol{\rho}})-\nu_{M,[k]}({\boldsymbol{\rho}}_0) | > \epsilon\right)
\end{eqnarray*}
\endgroup
The last inequality holds because $\lim P( d_k(\hat{\boldsymbol{\rho}},{\boldsymbol{\rho}}_0)>\delta) = 0$ by the consistency of $\hat{\boldsymbol{\rho}}$. Stochastic equicontinuity implies that $\delta$ can be chosen so that $\lim P\Big( \sup_{d(\boldsymbol{\rho}_1,\boldsymbol{\rho}_2)<\delta} | \nu_{M,[k]}(\hat{\boldsymbol{\rho}})-\nu_{M,[k]}({\boldsymbol{\rho}}_0) | > \epsilon\Big)$ is as small as desired. Thus,  
$\lim P(|\nu_{M,[k]}(\hat{\boldsymbol{\rho}})-\nu_{M,[k]}({\boldsymbol{\rho}}_0)|>\epsilon)$  
can also be made arbitrarily small, that is,  
$\nu_M(\hat{\boldsymbol{\rho}}) - \nu_M({\boldsymbol{\rho}}_0) = o_p(1)$,  
which corresponds to our condition \autoref{cond:sequi}. It is thus sufficient to show that $\nu_M$ is stochastically equicontinuous.

Following \citeOA{andrews1994empirical}, Section 5, we say that $\tilde{\mathbf{m}}_m$ is Type $IV(p=2)$ if the parameter space is bounded (Assumption \ref{as:regularity}) and
\begin{equation}\label{eq:IV_cond}
\sup_m\left( \mathbb{E}\Big(\sup_{\boldsymbol{\rho}_1:\Vert {\boldsymbol{\rho}}_1-{\boldsymbol{\rho}}_2\Vert<\delta} (\tilde{\mathbf{m}}_{m,[k]}({\boldsymbol{\rho}}_1)-\tilde{\mathbf{m}}_{m,[k]}({\boldsymbol{\rho}}_2))^2\Big)\right)^{1/2}\leq C\delta^\psi,
\end{equation}
for all ${\boldsymbol{\rho}}_2$ and all $\delta>0$ in a neighborhood of $0$, for some finite positive constants $C$ and $\psi$, and for all dimensions $[k]$.

We can express $\tilde{\mathbf{m}}_{m,[k]}(\boldsymbol{\rho})$ as a linear function of $\boldsymbol{\varepsilon}_m$ (e.g., see Equation \eqref{eq:m:lineps} in Appendix \ref{sec:proofconsistencySGMM}). Thus, $\tilde{\mathbf{m}}_{m,[k]}({\boldsymbol{\rho}}_1)-\tilde{\mathbf{m}}_{m,[k]}({\boldsymbol{\rho}}_2) = u_{1,m,[k]}(\boldsymbol{\rho}_1, \boldsymbol{\rho}_2) + \mathbf{u}_{2,m,[k]}(\boldsymbol{\rho}_1, \boldsymbol{\rho}_2)\boldsymbol{\varepsilon}_m$ for some scalar $u_{1,m,[k]}(\boldsymbol{\rho}_1, \boldsymbol{\rho}_2)$ and row vector $\mathbf{u}_{2,m,[k]}(\boldsymbol{\rho}_1, \boldsymbol{\rho}_2)$. Additionally, by Assumptions \ref{as:manymarkets} and \ref{as:nonsingular}, $|u_{1,m,[k]}(\boldsymbol{\rho}_1, \boldsymbol{\rho}_2)|$ and $\lVert \mathbf{u}_{2,m,[k]}(\boldsymbol{\rho}_1, \boldsymbol{\rho}_2) \rVert$ are uniformly bounded by some $\bar{u}_{1,[k]}$ and $\bar{u}_{2,[k]}$, respectively, where $\bar{u}_{1,[k]}$ and $\bar{u}_{2,[k]}$ do not depend on $\boldsymbol{\rho}_1$ and $\boldsymbol{\rho}_2$.  
Therefore, $(\tilde{\mathbf{m}}_{m,[k]}({\boldsymbol{\rho}}_1)-\tilde{\mathbf{m}}_{m,[k]}({\boldsymbol{\rho}}_2))^2$ is dominated by $(\bar{u}_{1,[k]} + \bar{u}_{2,[k]} \lVert \boldsymbol{\varepsilon}_m \rVert)^2$, for any sub-multiplicative norm $\lVert \cdot \rVert$. 

Since $(\bar{u}_{1} + \bar{u}_{2} \lVert \boldsymbol{\varepsilon}_m \rVert)^2$ is integrable (see Assumption \ref{as:finitevariance}),  
we can apply the dominated convergence theorem and interchange the expectation and the second supremum symbol in \eqref{eq:IV_cond}. A sufficient condition for \eqref{eq:IV_cond} is thus:  
\begin{equation}\label{eq:suffIV}
\sup_m\Big(\sup_{\boldsymbol{\rho}_1:\Vert {\boldsymbol{\rho}}_1-{\boldsymbol{\rho}}_2\Vert<\delta} \mathbb{E} ((\tilde{\mathbf{m}}_{m,[k]}({\boldsymbol{\rho}}_1)-\tilde{\mathbf{m}}_{m,[k]}({\boldsymbol{\rho}}_2))^2)\Big)^{1/2}\leq C\delta^\psi.
\end{equation}

Now, note that $\mathbb{E} (\tilde{\mathbf{m}}_{m,[k]}({\boldsymbol{\rho}}_1)-\tilde{\mathbf{m}}_{m,[k]}({\boldsymbol{\rho}}_2))^2$ is continuously differentiable in ${\boldsymbol{\rho}}_1$ with bounded derivatives following the argument in Lemma \ref{lemma:continuity}. See in particular Equation (\ref{eq:EB}).\footnote{The derivative is bounded because the linking probabilities have bounded derivatives (Assumption \ref{as:derivatives_P}), and $\mathbf{X}_m$, $\mathbb{E}(\lVert\boldsymbol{\varepsilon}_m\rVert^2)$, and the entries of the network matrices are bounded (Assumptions \ref{as:manymarkets}, \ref{as:nonsingular}, and \ref{as:finitevariance}).}

Then, by the Mean Value Theorem, we have  
$\mathbb{E} (\tilde{\mathbf{m}}_{m,[k]}({\boldsymbol{\rho}}_1)-\tilde{\mathbf{m}}_{m,[k]}({\boldsymbol{\rho}}_2))^2 = \mathcal{D}(\boldsymbol{\rho}^+, \boldsymbol{\rho}_2) (\boldsymbol{\rho}_1 - {\boldsymbol{\rho}}_2)$, where $\mathcal{D}(\boldsymbol{\rho}^+, \boldsymbol{\rho}_2)$ is the derivative of $\mathbb{E} (\tilde{\mathbf{m}}_{m,[k]}({\boldsymbol{\rho}}_1)-\tilde{\mathbf{m}}_{m,[k]}({\boldsymbol{\rho}}_2))^2$ with respect to ${\boldsymbol{\rho}}_1$ at some $\boldsymbol{\rho}^+$ lying between $\boldsymbol{\rho}_1$ and ${\boldsymbol{\rho}}_2$. Thus, $\tilde{\mathbf{m}}_m$ is of Type IV with $p=2$, $\psi = 1/2$, and $C$ as the upper bound of $\mathcal{D}(\boldsymbol{\rho}^+, \boldsymbol{\rho}_2)$. As a result, Condition (\ref{eq:suffIV}), and thus Condition (\ref{eq:IV_cond}), hold.

By Theorem 4 in \citeOA{andrews1994empirical}, stochastic equicontinuity \autoref{cond:sequi} holds if Ossianders' condition (his condition D) holds, $\lim\frac{1}{M}\sum_m\mathbb{E}\sup_{\boldsymbol{\rho}}|\tilde{\mathbf{m}}_m|^{2+\eta}<\infty$ for some $\eta > 0$ (his condition B), and if groups $m$ are independent (Assumption \ref{as:manymarkets}) implied by his condition C). By Theorem 5 in \citeOA{andrews1994empirical}, Ossianders' condition holds if $\tilde{\mathbf{m}}_m$ is Type $IV(p=2)$, which we just shown. His condition B is  verified because $\mathbb{E}(\lVert\boldsymbol{\varepsilon}_m\rVert^{2 + \eta})$ is bounded for some $\eta > 0$ (Assumption \ref{as:finitevariance}). Thus, as above, we can employ the dominated convergence theorem and interchange the expectation and the second supremum. Since $\mathbb E|\tilde{\mathbf{m}}_m|^{2+\eta}$ is bounded, then condition B follows.
\end{proof}

\subsection{Identification\label{app:ident}}
In this section, we show that the identification can be expressed as an identification condition on a concentrated objective function.

As \( \mathbf{W}_0 \) is positive definite (Assumption \ref{as:Wm}), identification is equivalent to stating that \( Q_0(\boldsymbol{\theta}) \) has a unique minimizer. Since \( Q_0(\boldsymbol{\theta}) \) depends on the true value \( \boldsymbol{\rho}_0 \) and not its estimator, all simulated networks in this section are drawn from the true network distribution. We therefore omit \( \boldsymbol{\rho}_0 \) from the notation for simplicity.

We define: 
\begin{align*}
    \mathbf{B}_m(\alpha)&=\frac{1}{RST}\sum_{rst}\dot{\mathbf{Z}}_m^{(r)\prime}(\mathbf{I}_m-\alpha\ddot{\mathbf{G}}_m^{(s)})(\mathbf{I}_m-\alpha\dddot{\mathbf{G}}_m^{(t)})^{-1}\dddot{\mathbf{V}}_m^{(t)},\\
    \mathbf{D}_m(\alpha)&=\break\frac{1}{RS}\sum_{rs}\dot{\mathbf{Z}}_m^{(r)\prime}(\mathbf{I}_m-\alpha\ddot{\mathbf{G}}_m^{(s)}).
\end{align*}
We have $\textstyle\bar{\mathbf{m}}_M(\boldsymbol{\theta},\boldsymbol{\rho}_0)=\frac{1}{M}\sum_m[\mathbf{D}_m(\alpha)\mathbf{y}_m-\mathbf{B}_m(\alpha)\tilde{\boldsymbol{\theta}}]$. The first-order condition of the minimization of $Q_0 (\boldsymbol{\theta})$ with respect to $\tilde{\boldsymbol{\theta}}$ is: $$
\left[\lim \frac{1}{M}\sum_m\mathbb E\left[\mathbf{B}_m(\alpha)\right]\tilde{\boldsymbol{\theta}}\right]'\mathbf{W}_0\left[\lim \frac{1}{M}\sum_m\big(\mathbb E\left[\mathbf{D}_m(\alpha)\mathbf{y}_m\right]-\mathbb E\left[\mathbf{B}_m(\alpha)\right]\tilde{\boldsymbol{\theta}}\big)\right] = 0.
$$
A sufficient condition for the last equation to have a unique solution in $\tilde{\boldsymbol{\theta}}$ is that $\textstyle\bar{\mathbf{B
}}_0(\alpha):=\lim \frac{1}{M}\sum_m\mathbb E\left[\mathbf{B}_m(\alpha)\right]$ is a full rank matrix for all $\alpha$. Under this condition, the solution $\tilde{\boldsymbol{\theta}}$ can be expressed as:
$$\hat{\tilde{\boldsymbol{\theta}}}(\alpha)=[\bar{\mathbf{B
}}_0'(\alpha)\mathbf{W}_0\bar{\mathbf{B
}}_0]^{-1}\bar{\mathbf{B
}}_0'(\alpha)\mathbf{W}_0\bar{\mathbf F}_0(\alpha),$$
where $\textstyle\bar{\mathbf F}_0(\alpha) = \plim \frac{1}{M}\sum_m\mathbb E\left[\mathbf{D}_m(\alpha)\mathbf{y}_m\right]$. By replacing $\mathbf y_m = (\mathbf I_m - \alpha_0 \mathbf G_m)^{-1}(\mathbf V_m \tilde{\boldsymbol\theta}_0 + \boldsymbol{\varepsilon}_m)$ in the expression of $\textstyle\bar{\mathbf F}_0(\alpha)$, we obtain:
\begin{align*}
    \bar{\mathbf F}_0(\alpha) &= \lim \frac{1}{M}\sum_m\frac{1}{RS}\sum_{rs}\mathbb E \left(\dot{\mathbf{Z}}_m^{(r)\prime}(\mathbf{I}_m-\alpha\ddot{\mathbf{G}}_m^{(s)})(\mathbf{I}_m-\alpha_0\mathbf{G}_m)^{-1}\mathbf{V}_m\tilde{\boldsymbol\theta}_0\right).
\end{align*}
Since $\bar{\mathbf F}_0(\alpha_0) = \bar{\mathbf{B
}}_0(\alpha_0)\tilde{\boldsymbol\theta}_0$, it follows that $\hat{\tilde{\boldsymbol{\theta}}}(\alpha_0) = \tilde{\boldsymbol\theta}_0$, which means that $\tilde{\boldsymbol\theta}_0$ is identified if $\alpha_0$ is identified; the underlying condition being that $\bar{\mathbf{B
}}_0(\alpha)$ is full rank.

By replacing the solution $\hat{\tilde{\boldsymbol{\theta}}}(\alpha)$ in the objective function, we can concentrate $Q_0(\boldsymbol{\theta})$ around $\alpha$ as
$Q_0^c(\alpha) = \bar{\mathbf Q}^c(\alpha)'\mathbf{W}_0\bar{\mathbf Q}^c(\alpha)$,
Where $$\bar{\mathbf Q}^c(\alpha) = \bar{\mathbf F}_0(\alpha) - \bar{\mathbf{B
}}_0(\alpha)[\bar{\mathbf{B
}}_0'(\alpha)\mathbf{W}_0\bar{\mathbf{B
}}_0(\alpha)]^{-1}\bar{\mathbf{B
}}_0'(\alpha)\mathbf{W}_0\bar{\mathbf F}_0(\alpha).$$

Let $\mathbf{W}_0^{1/2}$ be the positive definite square root of  $\mathbf{W}_0$. We have $\mathbf{W}_0^{1/2}\bar{\mathbf Q}^c(\alpha) = [\mathbf I_{w} - \mathbf P_{\mathbf B}(\alpha)]\mathbf{W}_0^{1/2}\bar{\mathbf F}_0(\alpha)$, where $\mathbf P_{\mathbf B}(\alpha) := \mathbf{W}_0^{1/2}\bar{\mathbf{B
}}_0(\alpha)[\bar{\mathbf{B
}}_0'(\alpha)\mathbf{W}_0\bar{\mathbf{B
}}_0(\alpha)]^{-1}\bar{\mathbf{B
}}_0'(\alpha)\mathbf{W}_0^{1/2}$ is a projection matrix onto the space of the column of $\mathbf{W}_0^{1/2}\bar{\mathbf{B
}}_0(\alpha)$ and $\mathbf I_{w}$ if the identity matrix of the same dimension as $\mathbf{W}_0$. The concentrated objective function can be written as:  
$$Q_0^c(\alpha) = [\mathbf{W}_0^{1/2}\bar{\mathbf F}_0(\alpha)]^{\prime}[\mathbf I_{w} - \mathbf P_{\mathbf B}(\alpha)]\mathbf{W}_0^{1/2}\bar{\mathbf F}_0(\alpha).$$

For identification to hold, the equation \( Q_0^c(\alpha) = 0 \) must not have multiple solutions. Unfortunately, simplifying this condition is challenging due to the nonlinearity of $Q_0^c(\alpha)$. A similar issue arises with the maximum likelihood estimator even when the network is fully observed. In this case, the identification condition also leads to a nonlinear equation in the peer effect parameter \citepOA[see][Assumption 9]{lee2004asymptotic}.  

Nevertheless, it is straightforward to sketch the empirical counterpart of \( Q_0^c(\alpha) \) since it is a function of a single variable. In numerous simulation exercises, we observe that \( Q_0^c(\alpha) \) is strictly convex, even when all entries of \( \mathbf{A}_m \) are simulated from an estimated distribution (without setting some entries to observed data). This evidence is encouraging and suggests that the solution to \( Q_0^c(\alpha) = 0 \) is likely unique.

\subsection{Asymptotic variance estimation}\label{sec:varestim}
Estimating the asymptotic variance of $\sqrt{M}(\hat{\boldsymbol{\theta}} - \boldsymbol{\theta}_0)$ can be challenging, as it requires computing the derivative of the expected moment function to estimate $\boldsymbol{\Gamma}_0 = \plim \frac{\partial \bar{\mathbf{m}}^{\ast}_M(\boldsymbol{\theta}_0, \boldsymbol{\rho}^+)}{\partial {\boldsymbol{\rho}^{\prime}}}$ (see Appendix \ref{sec:asymptnorm}). We now present a simple method for estimating this asymptotic variance by adapting \citeOA{houndetoungan2024inference}.

Taking the first derivative of the objective function at the second stage (for finite $R$,$S$,$T$ and conditional on $\hat{\mathbf{\rho}}$) with respect to $\boldsymbol{\theta}$, we have
$$
\left(\frac{\partial\bar{\mathbf{m}}_M(\hat{\boldsymbol{\theta}}, {\boldsymbol{\rho}})}{\partial\boldsymbol{\theta}^{\prime}}\right)^{\prime}\mathbf{W}_M[\bar{\mathbf{m}}_M(\hat{\boldsymbol{\theta}}, {\boldsymbol{\rho}})]=\mathbf{0}.
$$
By applying the mean value theorem to $\bar{\mathbf{m}}_M(\boldsymbol{\theta}, {\boldsymbol{\rho}})$, we have
$$\sqrt{M}(\hat{\boldsymbol{\theta}}-\boldsymbol{\theta}_0)=-\left[\mathbf H_M(\hat{\boldsymbol \theta})^{\prime} \mathbf{W}_M \mathbf H_M({\boldsymbol \theta}^{\ast})\right]^{-1}\mathbf H_M(\hat{\boldsymbol \theta})^{\prime}  \mathbf{W}_M\dfrac{1}{\sqrt{M}}\sum_{m = 1}^M\mathbf{m}_m(\boldsymbol{\theta}_0, \hat{\boldsymbol \rho}),   $$
where $\mathbf H_M(\boldsymbol \theta) = \dfrac{\partial\bar{\mathbf{m}}_M({\boldsymbol{\theta}}, \hat{\boldsymbol{\rho}})}{\partial\boldsymbol{\theta}^{\prime}}$ and $\boldsymbol \theta^{\ast}$ is some point between $\hat{\boldsymbol{\theta}}$ and $\boldsymbol{\theta}_0$.  

Let $\textstyle\boldsymbol{\Omega}_M = \Var(\frac{1}{\sqrt{M}}\sum_{m = 1}^M\mathbf{m}_m(\boldsymbol{\theta}_0, \hat{\boldsymbol \rho}))$. We assume the following:
\begin{assumption}
$\plim \boldsymbol{\Omega}_M = \boldsymbol{\Omega}_0$ and $\plim \mathbf{H}_M({\boldsymbol \theta}_0) = \mathbf{H}_0$ exist and are finite matrices.
\end{assumption}

Under this assumption, the asymptotic variance of $\sqrt{M}(\hat{\boldsymbol{\theta}}-\boldsymbol{\theta}_0)$ is:
$$\mathbb{V}_0(\sqrt{M}(\hat{\boldsymbol{\theta}}-\boldsymbol{\theta}_0)) = (\mathbf H_0^{\prime} \mathbf{W}_0 \mathbf H_0)^{-1}\mathbf H_0^{\prime}  \mathbf{W}_0 \boldsymbol{\Omega}_0 \mathbf{W}_0 \mathbf H_0 (\mathbf H_0^{\prime} \mathbf{W}_0 \mathbf H_0)^{-1}.$$

The expression for the asymptotic variance is similar to the variance of the standard GMM estimator. The key difference is that $\boldsymbol{\Omega}_0$, which is the asymptotic variance of $\frac{1}{\sqrt{M}}\sum_{m = 1}^M\mathbf{m}_m(\boldsymbol{\theta}_0, \hat{\boldsymbol \rho})$, accounts for the uncertainty in $\boldsymbol\varepsilon_m$, the first-stage estimator $\hat{\boldsymbol{\rho}}$, and the finite number of simulated networks from the estimated network distribution.

To estimate $\mathbb{V}_0(\sqrt{M}(\hat{\boldsymbol{\theta}}-\boldsymbol{\theta}_0))$, one can replace $\mathbf{H}_0$ and $\mathbf{W}_0$ with their usual estimators. Specifically, $\mathbf{H}_0$ can be estimated by $\mathbf{H}_M(\hat{\boldsymbol{\theta}})$ and $\mathbf{W}_0$ can be estimated by $\mathbf{W}_M$. Let $\mathcal S$ be the set of simulated networks from the network distribution and the true network. To estimate $\boldsymbol{\Omega}_0$, we rely on the following decomposition.
\begingroup
\allowdisplaybreaks
\begin{align}
    \boldsymbol{\Omega}_M &= \mathbb V \left(\frac{1}{\sqrt{M}}\sum_{m = 1}^M\mathbf{m}_m(\boldsymbol{\theta}_0, \hat{\boldsymbol \rho})\right),\nonumber\\
    \boldsymbol{\Omega}_M &= \mathbb E\left\{\mathbb V \left(\frac{1}{\sqrt{M}}\sum_{m = 1}^M\mathbf{m}_m(\boldsymbol{\theta}_0, \hat{\boldsymbol \rho})| \mathbf X_m, \mathcal S\right)\right\} + \mathbb V\left\{\mathbb E \left(\frac{1}{\sqrt{M}}\sum_{m = 1}^M\mathbf{m}_m(\boldsymbol{\theta}_0, \hat{\boldsymbol \rho})| \mathbf X_m, \mathcal S\right)\right\}, \nonumber\\
    \boldsymbol{\Omega}_M &= \mathbb E\left(\frac{1}{M}\sum_{m = 1}^M\mathbf V_m\right) + \mathbb V\left(\frac{1}{\sqrt{M}}\sum_{m = 1}^M\mathcal E_m\right). \nonumber
\end{align}
\endgroup
where $\mathbf V_m = \mathbb V \left(\mathbf{m}_m(\boldsymbol{\theta}_0, \hat{\boldsymbol \rho})| \mathbf X_m, \hat{\boldsymbol{\rho}}, \kappa(\mathcal{A}_m)\right)$ and $\mathcal E_m = \mathbb E \left(\mathbf{m}_m(\boldsymbol{\theta}_0, \hat{\boldsymbol \rho})| \mathbf X_m, \hat{\boldsymbol{\rho}}, \kappa(\mathcal{A}_m)\right)$.

Note that both $\mathbf{V}_m$ and $\mathcal{E}_m$ can be easily computed and estimated. They represent the conditional variance and the conditional expectation of the moment function, given $\mathcal{S}$. 

Let $\mathbf V_M = \frac{1}{M} \sum_{m = 1}^{M} \mathbf{V}_m$ and $\mathcal{E}_M = \frac{1}{\sqrt{M}} \sum_{m = 1}^{M} \mathcal{E}_m$. It follows that:
$$\boldsymbol{\Omega}_0 = \plim  \mathbf V_M + \lim \mathbb{V} \left(\mathcal E_M\right).$$

The first term is due to the error term of the model $\boldsymbol{\varepsilon}_m$, whereas the second term reflects the uncertainty associated with the estimation of $\hat{\boldsymbol{\rho}}$ and the simulated networks.
In practice, $\plim  \mathbf V_M$ can be estimated by the average of the conditional variance of the moment function without accounting for the uncertainty in the simulated network. To estimate $\lim \mathbb{V} \left(\mathcal E_M\right)$, we repeatedly generate many $\mathcal{S}$, each associated with a new $\boldsymbol{\rho}$ simulated from the estimator of the distribution of $\hat{\boldsymbol{\rho}}$. For each $\mathcal{S}$, we compute the associated $\mathcal E_M$. Finally, the sample variance of the generated values of $\mathcal E_M$ serves as an estimator of $\lim \mathbb{V} \left(\mathcal E_M\right)$.\footnote{Our R package offers tools to compute this variance.}

%% file: appendix_proofs_others.tex
\section{Additional Technical Results }\label{sec:appendix_otherresults}

\subsection{Simple estimators}
When the network is fully observed, the moment function of the standard instrumental variables approach is linear in parameters \citep{bramoulle2009identification}. Consequently, the estimator can be computed without requiring numerical optimization and identification conditions can be easily tested. Our SGMM estimator does not exhibit such simplicity when $\mathbf{G}_m\mathbf{y}_m$ is not observed. In this section, we discuss other straightforward estimators that result from a linear moment function. We first discuss the case where $\mathbf{G}_m\mathbf{X}_m$ and $\mathbf{G}_m\mathbf{y}_m$ are observed.
\begin{proposition}\label{prop:all_observed}
[Conditions] Suppose that $\mathbf{G}_m\mathbf{X}_m$ and $\mathbf{G}_m\mathbf{y}_m$ are observed. Let $\boldsymbol{H}_m$ be a matrix such that (1) at least one column of $\mathbf{H}^k_m\mathbf{X}_m$ is (strongly) correlated with $\mathbf{G}_m\mathbf{y}_m$, conditional on $[\mathbf{1}_m,\mathbf{X}_m,\mathbf{G}_m\mathbf{X}_m]$ for $k\geq 2$, and (2) $\mathbb{E}[\boldsymbol{\varepsilon}_m|\mathbf{X}_m,\mathbf{A}_m,\mathbf{H}_m]=\mathbf{0}$ for all $m$. Finally,  define the matrix $\mathbf{Z}_m=[\mathbf{1}_m,\mathbf{X}_m,\mathbf{G}_m\mathbf{X}_m,\mathbf{H}^2_m\mathbf{X}_m,\mathbf{H}^3_m\mathbf{X}_m...]$.

[Results] Then, under classical assumptions (e.g., \citeOA{cameron2005microeconometrics}, Proposition 6.1), the (linear) GMM estimator based on the moment function $\frac{1}{M}\sum_m\mathbf{Z}_m'\boldsymbol{\varepsilon}_m$ is consistent and asymptotically normally distributed with the usual asymptotic variance-covariance matrix.
\end{proposition}

Condition (1) is the relevancy condition, whereas condition (2) is the exogeneity condition.\footnote{Although (for simplicity) in Proposition \ref{prop:all_observed}, we use the entire matrix $\mathbf{X}_m$ to generate the instruments $\mathbf{H}_m\mathbf{X}_m$, in practice, one should avoid including instruments (i.e., columns of $\mathbf{H}_m\mathbf{X}_m$) that are weakly correlated with $\mathbf{G}_m\mathbf{y}_m$.} Although Proposition \ref{prop:all_observed} holds for any matrix $\mathbf{H}_m$ such that conditions (1) and (2) hold, the most sensible example in our context is when $\mathbf{H}_m$ is constructed using a draw from $\hat{P}(\mathbf{A}_m|\hat{\boldsymbol{\rho}},\mathbf{X}_m,\kappa(\mathcal{A}_m))$.

Importantly, the moment conditions remain valid even when the researcher uses a \emph{mispecified} estimator of the distribution ${P}(\mathbf{A}_m|\mathbf{X}_m,\kappa(\mathcal{A}_m))$, as long as the specification error on ${P}(\mathbf{A}_m|\mathbf{X}_m,\kappa(\mathcal{A}_m))$ does not induce a correlation with $\boldsymbol{\varepsilon}_m$.\footnote{We would like to thank Chih-Sheng Hsieh and Arthur Lewbel for discussions on this important point.} This could be of great practical importance, especially if the estimation of $\hat{P}(\mathbf{A}_m|\hat{\boldsymbol{\rho}},\mathbf{X}_m,\kappa(\mathcal{A}_m))$ suffers from a small sample bias.

Second, we present a simple, but asymptotically biased, linear GMM estimator when $\mathbf{G}_m\mathbf{X}_m$ is observed and $\mathbf{G}_m\mathbf{y}_m$ is not. The presentation of such an estimator is useful because simulations show that the asymptotic bias turns out to be negligible in many cases, especially for moderate group sizes. Moreover, the estimator is computationally attractive because the estimator can be written in a closed form.

\begin{proposition}\label{prop:bias_nocontext}
[Conditions] Assume that $\mathbf{G}_m\mathbf{X}_m$ is observed. Let $\ddot{\mathbf{S}}_m=[\mathbf{1}_m,\mathbf{X}_m,\mathbf{G}_m\mathbf{X}_m,\break\ddot{\mathbf{G}}_m\mathbf{X}_m,\ddot{\mathbf{G}}_m\mathbf{y}_m]$ and $\dot{\mathbf{Z}}_m=[\mathbf{1}_m,\mathbf{X}_m,\mathbf{G}_m\mathbf{X}_m,\ddot{\mathbf{G}}_m\mathbf{X}_m,\dot{\mathbf{G}}^2_m\mathbf{X}_m,\dot{\mathbf{G}}^3_m\mathbf{X}_m,...]$. We consider the moment function $\frac{1}{M} \sum_{m = 1}^M \dot{\mathbf{Z}}'_m (\mathbf y_m - \ddot{\mathbf{S}}_m\ddot{\boldsymbol{\theta}})$ and $\check{\boldsymbol{\theta}}$ be the associated GMM estimator of  $\ddot{\boldsymbol{\theta}}$. We define the sensitivity matrix ${\mathbf{R}}=\left(\frac{\sum_m\ddot{\mathbf{S}}_m'\dot{\mathbf{Z}}_m}{M}\mathbf{W}_M\frac{\sum_m \dot{\mathbf{Z}}_m'\ddot{\mathbf{S}}_m}{M}\right)^{-1}\frac{\sum_m\ddot{\mathbf{S}}_m'\dot{\mathbf{Z}}_m}{M}\mathbf{W}_M$.

[Result] Under classical assumptions (see proof), the asymptotic bias of $\hat{\boldsymbol{\theta}}$ is given by $\alpha_0  \plim\frac{\mathbf{R}\sum_m \dot{\mathbf{Z}}_m'({\mathbf{G}}_m-\ddot{\mathbf{G}}_m)\mathbf{y}_m}{M}$. Moreover, letting $\mathbf{W}_M$ be an identity matrix minimizes the asymptotic bias in the sense of minimizing the Frobenius norm of $\mathbf{R}$.
\end{proposition}
Although there are no obvious ways to obtain a consistent estimate of the asymptotic bias (because $\mathbf{y}_m$ is a function of $\mathbf{G}_m$ and $\alpha_0$), simulations show that the bias is very small in practice.

The intuition behind Proposition \ref{prop:bias_nocontext} comes from the literature on error-in-variable models with repeated observations (e.g., \citeOA{bound2001measurement}). The instrumental variable uses two independent draws from the (estimated) distribution of the true network. One draw is used to proxy the unobserved variable (i.e., $\mathbf{G}_m\mathbf{y}_m$), whereas the other is used to proxy the instrument (i.e., $\mathbf{G}_m\mathbf{X}_m$). This approach greatly reduces the bias compared with a situation in which only one draw would be used.

The argument in Proposition \ref{prop:bias_nocontext} is very similar to the one in \citeOA{andrews2017measuring}, although here perturbation with respect to the true model is not \textit{local}.\footnote{See page 1562 in \citeOA{andrews2017measuring}.} We also show that we expect the identity matrix weight to minimize the asymptotic bias. Our result, therefore, provides a theoretical justification for the simulations in \citeOA{onishi2021sample} who show that using the identity matrix to weight the moments greatly reduces the bias in the context studied by \citeOA{andrews2017measuring}.

\subsubsection{Proof of Proposition \ref{prop:bias_nocontext}}
\noindent\textbf{Part 1: Asymptotic bias}\\
Let $\ddot{\boldsymbol{\theta}}_0$ be the true value the parameter when regressors are defined as $\ddot{\mathbf{S}}_m=[\mathbf{1}_m,\mathbf{X}_m,\mathbf{G}_m\mathbf{X}_m, \break\ddot{\mathbf{G}}_m\mathbf{X}_m,\ddot{\mathbf{G}}_m\mathbf{y}_m]$; that is, the true coefficient vector associated with $[\mathbf{1}_m,\mathbf{X}_m\mathbf{G}_m\mathbf{X}_m,\ddot{\mathbf{G}}_m\mathbf{y}_m]$ is $\boldsymbol{\theta}_0$ while the true coefficient vector associated with $\ddot{\mathbf{G}}_m\mathbf{X}_m$ is zero.

We now use matrix notation as the sample level to avoid summations over $m$ and the index $m$. For example, $\dot{\mathbf{Z}}'\ddot{\mathbf{S}}= \sum_m \dot{\mathbf{Z}}_m'\ddot{\mathbf{S}}_m$. The Linear GMM estimator can be written as
\begin{align*}
    &\check{\boldsymbol{\theta}}=\left(\frac{\ddot{\mathbf{S}}'\dot{\mathbf{Z}}}{M}\mathbf W_M\frac{\dot{\mathbf{Z}}'\ddot{\mathbf{S}}}{M}\right)^{-1}\frac{\ddot{\mathbf{S}}'\dot{\mathbf{Z}}}{M}\mathbf W_M \dot{\mathbf{Z}}'\left(\frac{\ddot{\mathbf{S}}\ddot{\boldsymbol{\theta}}_0 + \boldsymbol{\eta}+\boldsymbol{\varepsilon}}{M}\right)
\end{align*}
where $\boldsymbol{\eta}=\alpha_0(\mathbf{G}-\ddot{\mathbf{G}})\mathbf{y}$ is due to the approximation of $\mathbf{G}\mathbf{y}$ by $\ddot{\mathbf{G}}\mathbf{y}$ in $\ddot{\mathbf{S}}$. Therefore 
$\check{\boldsymbol{\theta}}= \ddot{\boldsymbol{\theta}}_0 + {\mathbf{R}} \left(\frac{ \dot{\mathbf{Z}}'\boldsymbol{\eta}+\dot{\mathbf{Z}}'\boldsymbol{\varepsilon}}{M}\right)$ and the asymptotic bias of $\check{\boldsymbol{\theta}}$ is $\plim(\check{\boldsymbol{\theta}} - \ddot{\boldsymbol{\theta}}_0) = \alpha_0  \plim\frac{{\mathbf{R}}\dot{\mathbf{Z}}'({\mathbf{G}}-\ddot{\mathbf{G}})\mathbf{y}}{M}.$
\bigskip
\noindent\textbf{Part 2: Choice of $\mathbf{W}$} (we omit the index $M$ for simplicity)\\
Let $\boldsymbol{\Delta}=\mathbf{G}-\ddot{\mathbf{G}}$, $\mathbf{K}=\dot{\mathbf{Z}}'\boldsymbol{\Delta}\mathbf{G}^2/M$, if $\boldsymbol{\gamma}=\mathbf{0}$, and $\mathbf{K}=\dot{\mathbf{Z}}'\boldsymbol{\Delta}/M$ otherwise. Consider $\Vert \mathbf{RK}\Vert_F=\sqrt{trace(\mathbf{K}'\mathbf{R}'\mathbf{R}\mathbf{K})}=\sqrt{trace(\mathbf{K}\mathbf{K}'\mathbf{R}'\mathbf{R})}$. We have
$$
(1/M^2)\mathbf{R}\mathbf{R}'=[\ddot{\mathbf{S}}'\dot{\mathbf{Z}}\mathbf{W}\dot{\mathbf{Z}}'\ddot{\mathbf{S}}]^{-1}\ddot{\mathbf{S}}'\dot{\mathbf{Z}}\mathbf{W}\mathbf{W}\dot{\mathbf{Z}}'\ddot{\mathbf{S}}[\ddot{\mathbf{S}}'\dot{\mathbf{Z}}\mathbf{W}\dot{\mathbf{Z}}'\ddot{\mathbf{S}}]^{-1}.
$$
Let $\mathbf{W}=\mathbf{C}'\mathbf{C}$ and let $\mathbf{B}=\ddot{\mathbf{S}}'\dot{\mathbf{Z}}\mathbf{C}'$. We have
$$
(1/M^2)\mathbf{R}\mathbf{R}'=(\mathbf{B}'\mathbf{B})^{-1}\mathbf{B}'\mathbf{C}\mathbf{C}'
\mathbf{B}(\mathbf{B}'\mathbf{B})^{-1}.
$$
Now, define $\mathbf{J}'=(\mathbf{B}'\mathbf{B})^{-1}\mathbf{B}'\mathbf{C}-(\mathbf{B}'(\mathbf{C}')^{-1}\mathbf{C}^{-1}\mathbf{B})^{-1}\mathbf{B}'(\mathbf{C}')^{-1}$. We have
$$
(1/M^2)\mathbf{R}\mathbf{R}'=\mathbf{J}'\mathbf{J}+(\mathbf{B}'(\mathbf{C}')^{-1}\mathbf{C}^{-1}\mathbf{B})^{-1}=\mathbf{J}'\mathbf{J}+ (\ddot{\mathbf{S}}'\dot{\mathbf{Z}}\dot{\mathbf{Z}}'\ddot{\mathbf{S}})^{-1}.
$$
Therefore, we have
$$
(1/M^2)\Vert \mathbf{R}\Vert_F=\sqrt{trace(\mathbf{J}'\mathbf{J}+ (\ddot{\mathbf{S}}'\dot{\mathbf{Z}}\dot{\mathbf{Z}}'\ddot{\mathbf{S}})^{-1})}=\sqrt{trace(\mathbf{J}'\mathbf{J})+ trace((\ddot{\mathbf{S}}'\dot{\mathbf{Z}}\dot{\mathbf{Z}}'\ddot{\mathbf{S}})^{-1})}.
$$
When $\mathbf{W}=\mathbf{I}$, we have that $\mathbf{J}=\mathbf{0}$ and the Frobenius norm of $\mathbf{R}$ is given by \break$M^2\sqrt{trace((\ddot{\mathbf{S}}'\dot{\mathbf{Z}}\dot{\mathbf{Z}}'\ddot{\mathbf{S}})^{-1})}$.

\subsection{Corollaries to Theorem \ref{prop:non_observed}}

\begin{corollary}\label{prop:GXobs}
Assume that $\mathbf{G}_m\mathbf{X}_m$ is observed but that $\mathbf{G}_m\mathbf{y}_m$ is not observed. Let $\dot{\mathbf{Z}}_m^{(r)}=[\mathbf{1}_m,\mathbf{X}_m,{\mathbf{G}}_m\mathbf{X}_m,\dot{\mathbf{G}}^{(r)}_m{\mathbf{G}}\mathbf{X},(\dot{\mathbf{G}}_m^{(r)})^2{\mathbf{G}}_m\mathbf{X}_m,...]$,

\noindent $\ddot{\mathbf{Z}}_m^{(r,s)}=[\mathbf{1}_m,\mathbf{X}_m,\ddot{\mathbf{G}}^{(s)}_m\mathbf{X}_m,\dot{\mathbf{G}}^{(r)}_m\ddot{\mathbf{G}}_m^{(s)}\mathbf{X}_m,(\dot{\mathbf{G}}_m^{(r)})^2\ddot{\mathbf{G}}^{(s)}_m\mathbf{X}_m,...]$, $\mathbf{V}_m=[\mathbf{1}_m,\mathbf{X}_m,\mathbf{G}_m\mathbf{X}_m]$, and $\ddot{\mathbf{V}}_m^{(s)}=[\mathbf{1}_m,\mathbf{X}_m,\ddot{\mathbf{G}}^{(s)}_m\mathbf{X}_m]$. Then, the results from Theorem \ref{prop:non_observed} hold for the following (simulated) moment function:
\begin{eqnarray}\label{eq:Gyobs}
\bar{\mathbf{m}}_{M}(\boldsymbol{\theta})=\frac{1}{M}\sum_m\frac{1}{R}\sum_{r}\dot{\mathbf{Z}}_m^{(r)\prime}(\mathbf{I}_m-\alpha\dot{\mathbf{G}}^{(r)}_m)\mathbf{y}_m -
\frac{1}{RS}\sum_{r,s}(\dot{\mathbf{Z}}_m^{(r)\prime}\mathbf{V}_m-\ddot{\mathbf{Z}}_m^{(r,s)\prime}\ddot{\mathbf{V}}_m^{(s)})\tilde{\boldsymbol{\theta}}\nonumber\\ -\frac{1}{RS}\sum_{r,s}  \ddot{\mathbf{Z}}_m^{(r,s)\prime}(\mathbf{I}_m-\alpha\dot{\mathbf{G}}^{(r)}_m)(\mathbf{I}_m-\alpha\ddot{\mathbf{G}}_m^{(s)})^{-1}\ddot{\mathbf{V}}^{(s)}_m\tilde{\boldsymbol{\theta}}
\end{eqnarray}
under the same conditions.
\end{corollary}

\begin{corollary}\label{prop:Gyobs}
Assume that $\mathbf{G}_m\mathbf{y}_m$ is observed but that $\mathbf{G}_m\mathbf{X}_m$ is not observed. Let $\dot{\mathbf{Z}}_m^{(r)}=[\mathbf{1}_m,\mathbf{X}_m,\dot{\mathbf{G}}^{(r)}_m\mathbf{X}_m,(\dot{\mathbf{G}}_m^{(r)})^2\mathbf{X}_m,...]$, and $\ddot{\mathbf{V}}_m^{(s)}=[\mathbf{1}_m,\mathbf{X}_m,\ddot{\mathbf{G}}^{(s)}_m\mathbf{X}_m]$.  Then, the results from Theorem \ref{prop:non_observed} hold for the following (simulated) moment function:
\begin{eqnarray}\label{eq:GXobs}
\frac{1}{R}\sum_{r}\dot{\mathbf{Z}}_m^{(r)\prime}(\mathbf{I}_m-\alpha\mathbf{G}_m)\mathbf{y}_m -
\frac{1}{RS}\sum_{r,s} \dot{\mathbf{Z}}^{(r)\prime}_m\ddot{\mathbf{V}}_m^{(s)}\tilde{\boldsymbol{\theta}}
\end{eqnarray}
under the same conditions.
\end{corollary}

%% file: appendix_newsimulations.tex
\section{Full Simulation Results for the SGMM Estimator}\label{append:fullsimu}

Tables \ref{tan:sim:noFE}--\ref{tab:simmclas:FE} report simulation results for the SGMM estimator. We also report simulations for cases where $\mathbf{G}_m\mathbf{y}_m$ and/or $\mathbf{G}_m\mathbf{X}_m$ is observed, including cases that account for unobserved group heterogeneity. The estimator still performs well in these settings. Precision improves significantly when $\mathbf{G}_m\mathbf{y}_m$ is observed, more so than when $\mathbf{G}_m\mathbf{X}_m$ is observed. This occurs because $\mathbf{G}_m\mathbf{y}_m$ is a nonlinear function of the true network $\mathbf{G}_m$ and peer effect coefficient $\alpha_0$. Therefore, its approximation is more challenging than that of $\mathbf{G}_m\mathbf{X}_m$, which is exogenous.  For all simulations, we set $R = 100$ and $S = T = 1$.

\renewcommand{\arraystretch}{.9}
\begin{table}[!htbp]
\footnotesize
\centering
\caption{Full simulation results under missing links without group fixed effects}

\label{tan:sim:noFE}
\begin{tabular}{p{2.1cm}d{3}d{3}d{3}d{3}d{3}d{3}d{3}d{3}}
\toprule
Proportion of missing links & \multicolumn{2}{c}{0\%} & \multicolumn{2}{c}{25\%} & \multicolumn{2}{c}{50\%} & \multicolumn{2}{c}{75\%} \\ \\
Statistic & \multicolumn{1}{c}{Mean} & \multicolumn{1}{c}{Std} & \multicolumn{1}{c}{Mean} & \multicolumn{1}{c}{Std} & \multicolumn{1}{c}{Mean} & \multicolumn{1}{c}{Std} & \multicolumn{1}{c}{Mean} & \multicolumn{1}{c}{Std} \\\midrule
                            & \multicolumn{8}{c}{Classical IV: $\mathbf{Gy}$ observed and $\mathbf{GX}$ unobserved; Instruments: $\mathbf{GX}^2$} \\[1ex]
$\alpha = 0.538$            & 0.537            &0.008     & 0.531            &0.014      & 0.525            &0.023     & 0.518           &0.055     \\
$c = 3.806$                 & 3.807            &0.132     & 4.378            &0.168      & 4.885            &0.267     & 5.336           &0.638     \\
$\beta_1 = -0.072$          & -0.072           &0.009     & -0.089           &0.011      & -0.106           &0.015     & -0.124          &0.030     \\
$\beta_2 = 0.133$           & 0.133            &0.027     & 0.136            &0.030      & 0.141            &0.030     & 0.143           &0.033     \\
$\gamma_1 = 0.086$          & 0.086            &0.005     & 0.063            &0.005      & 0.046            &0.006     & 0.033           &0.010     \\
$\gamma_2 = -0.003$         & -0.003           &0.037     & -0.009           &0.036      & -0.013           &0.040     & -0.013          &0.052     \\\midrule
                            & \multicolumn{8}{c}{Classical IV: $\mathbf{Gy}$ and $\mathbf{GX}$ unobserved; Instruments: $\mathbf{GX}^2$}          \\[1ex]
$\alpha = 0.538$            & 0.537            &0.008     & 0.442            &0.014      & 0.362            &0.021     & 0.293           &0.043     \\
$c = 3.806$                 & 3.807            &0.132     & 6.598            &0.325      & 8.931            &0.412     & 10.789          &0.474     \\
$\beta_1 = -0.072$          & -0.072           &0.009     & -0.176           &0.021      & -0.273           &0.028     & -0.358          &0.033     \\
$\beta_2 = 0.133$           & 0.133            &0.027     & 0.151            &0.058      & 0.168            &0.072     & 0.186           &0.084     \\
$\gamma_1 = 0.086$          & 0.086            &0.005     & 0.030            &0.008      & -0.006           &0.011     & -0.027          &0.022     \\
$\gamma_2 = -0.003$         & -0.003           &0.037     & -0.028           &0.045      & -0.046           &0.054     & -0.050          &0.078     \\\midrule
                            & \multicolumn{8}{c}{SGMM: $\mathbf{Gy}$ and $\mathbf{GX}$ observed; $T = 100$}                                                                         \\[1ex]
$\alpha = 0.538$            & 0.537            &0.008     & 0.537            &0.012      & 0.538            &0.015     & 0.539           &0.021     \\
$c = 3.806$                 & 3.807            &0.132     & 3.811            &0.133      & 3.801            &0.136     & 3.805           &0.150     \\
$\beta_1 = -0.072$          & -0.072           &0.009     & -0.073           &0.009      & -0.072           &0.009     & -0.072          &0.010     \\
$\beta_2 = 0.133$           & 0.133            &0.027     & 0.132            &0.028      & 0.134            &0.026     & 0.132           &0.026     \\
$\gamma_1 = 0.086$          & 0.086            &0.005     & 0.086            &0.006      & 0.086            &0.007     & 0.086           &0.010     \\
$\gamma_2 = -0.003$         & -0.003           &0.037     & -0.003           &0.036      & -0.003           &0.037     & -0.003          &0.038     \\\midrule
                            & \multicolumn{8}{c}{SGMM: $\mathbf{Gy}$ observed and $\mathbf{GX}$ unobserved; $S = T = 100$}                                                          \\[1ex]
$\alpha = 0.538$            & 0.537            &0.008     & 0.538            &0.012      & 0.538            &0.019     & 0.540           &0.033     \\
$c = 3.806$                 & 3.807            &0.132     & 3.812            &0.150      & 3.806            &0.178     & 3.802           &0.225     \\
$\beta_1 = -0.072$          & -0.072           &0.009     & -0.073           &0.010      & -0.072           &0.011     & -0.072          &0.014     \\
$\beta_2 = 0.133$           & 0.133            &0.027     & 0.132            &0.030      & 0.134            &0.030     & 0.132           &0.032     \\
$\gamma_1 = 0.086$          & 0.086            &0.005     & 0.086            &0.006      & 0.086            &0.010     & 0.085           &0.016     \\
$\gamma_2 = -0.003$         & -0.003           &0.037     & -0.003           &0.042      & -0.003           &0.056     & -0.004          &0.081     \\\midrule
                            & \multicolumn{8}{c}{SGMM: $\mathbf{Gy}$ unobserved and $\mathbf{GX}$ observed; $S = T = 100$}                                                          \\[1ex]
$\alpha = 0.538$            & 0.537            &0.008     & 0.538            &0.015      & 0.538            &0.027     & 0.540           &0.064     \\
$c = 3.806$                 & 3.807            &0.132     & 3.817            &0.263      & 3.819            &0.350     & 3.803           &0.492     \\
$\beta_1 = -0.072$          & -0.072           &0.009     & -0.073           &0.016      & -0.073           &0.021     & -0.072          &0.030     \\
$\beta_2 = 0.133$           & 0.133            &0.027     & 0.132            &0.047      & 0.133            &0.057     & 0.135           &0.066     \\
$\gamma_1 = 0.086$          & 0.086            &0.005     & 0.086            &0.009      & 0.086            &0.014     & 0.085           &0.031     \\
$\gamma_2 = -0.003$         & -0.003           &0.037     & -0.002           &0.050      & -0.005           &0.076     & -0.004          &0.136     \\\midrule
                            & \multicolumn{8}{c}{SGMM: $\mathbf{Gy}$ and $\mathbf{GX}$ unobserved; $R = S = T = 100$}                                                               \\[1ex]
$\alpha = 0.538$            & 0.537            &0.008     & 0.538            &0.016      & 0.539            &0.029     & 0.542           &0.073     \\
$c = 3.806$                 & 3.807            &0.132     & 3.816            &0.314      & 3.821            &0.423     & 3.794           &0.580     \\
$\beta_1 = -0.072$          & -0.072           &0.009     & -0.073           &0.019      & -0.073           &0.026     & -0.071          &0.036     \\
$\beta_2 = 0.133$           & 0.133            &0.027     & 0.132            &0.055      & 0.133            &0.069     & 0.136           &0.079     \\
$\gamma_1 = 0.086$          & 0.086            &0.005     & 0.086            &0.009      & 0.086            &0.016     & 0.084           &0.035     \\
$\gamma_2 = -0.003$         & -0.003           &0.037     & -0.002           &0.052      & -0.005           &0.083     & -0.005          &0.153     \\ \bottomrule
\multicolumn{9}{l}{%
  \begin{minipage}{12cm}%
  \vspace{0.05cm}
    \footnotesize{Note: We perform 1,000 simulations. 'Std' denotes the standard deviation.}
  \end{minipage}%
}
\end{tabular}
\end{table}

\begin{table}[!htbp]
\footnotesize
\centering
\caption{Full simulation results under missing links with group fixed effects}
\label{tan:sim:FE}
\begin{tabular}{p{2.1cm}d{3}d{3}d{3}d{3}d{3}d{3}d{3}d{3}}
\toprule
Proportion of missing links & \multicolumn{2}{c}{0\%} & \multicolumn{2}{c}{25\%} & \multicolumn{2}{c}{50\%} & \multicolumn{2}{c}{75\%} \\ \\
Statistic & \multicolumn{1}{c}{Mean} & \multicolumn{1}{c}{Std} & \multicolumn{1}{c}{Mean} & \multicolumn{1}{c}{Std} & \multicolumn{1}{c}{Mean} & \multicolumn{1}{c}{Std} & \multicolumn{1}{c}{Mean} & \multicolumn{1}{c}{Std} \\\midrule
                            & \multicolumn{8}{c}{Classical IV: $\mathbf{Gy}$ observed and $\mathbf{GX}$ unobserved; Instruments: $\mathbf{GX}^2$} \\[1ex]
$\alpha = 0.538$            & 0.538            &0.009     & 0.531            &0.015      & 0.523            &0.027     & 0.522           &0.067     \\
$\beta_1 = -0.072$          & -0.072           &0.009     & -0.090           &0.011      & -0.107           &0.016     & -0.123          &0.035     \\
$\beta_2 = 0.133$           & 0.133            &0.028     & 0.136            &0.030      & 0.139            &0.032     & 0.143           &0.033     \\
$\gamma_1 = 0.086$          & 0.086            &0.005     & 0.063            &0.005      & 0.046            &0.006     & 0.032           &0.011     \\
$\gamma_2 = -0.003$         & -0.001           &0.038     & -0.009           &0.039      & -0.012           &0.042     & -0.012          &0.052     \\\midrule
                            & \multicolumn{8}{c}{Classical IV: $\mathbf{Gy}$ and $\mathbf{GX}$ unobserved; Instruments: $\mathbf{GX}^2$}          \\[1ex]
$\alpha = 0.538$            & 0.538            &0.009     & 0.432            &0.016      & 0.343            &0.025     & 0.272           &0.049     \\
$\beta_1 = -0.072$          & -0.072           &0.009     & -0.178           &0.022      & -0.273           &0.029     & -0.353          &0.033     \\
$\beta_2 = 0.133$           & 0.133            &0.028     & 0.149            &0.058      & 0.165            &0.074     & 0.182           &0.081     \\
$\gamma_1 = 0.086$          & 0.086            &0.005     & 0.033            &0.009      & 0.000            &0.013     & -0.019          &0.024     \\
$\gamma_2 = -0.003$         & -0.001           &0.038     & -0.028           &0.051      & -0.044           &0.061     & -0.045          &0.081     \\\midrule
                            & \multicolumn{8}{c}{SGMM: $\mathbf{Gy}$ and $\mathbf{GX}$ observed; $T = 100$}                                                                         \\[1ex]
$\alpha = 0.538$            & 0.538            &0.009     & 0.537            &0.013      & 0.538            &0.017     & 0.536           &0.027     \\
$\beta_1 = -0.072$          & -0.072           &0.009     & -0.073           &0.009      & -0.072           &0.010     & -0.073          &0.011     \\
$\beta_2 = 0.133$           & 0.133            &0.028     & 0.133            &0.027      & 0.133            &0.027     & 0.133           &0.027     \\
$\gamma_1 = 0.086$          & 0.086            &0.005     & 0.086            &0.006      & 0.086            &0.008     & 0.087           &0.012     \\
$\gamma_2 = -0.003$         & -0.001           &0.038     & -0.002           &0.039      & -0.004           &0.041     & -0.002          &0.040     \\\midrule
                            & \multicolumn{8}{c}{SGMM: $\mathbf{Gy}$ observed and $\mathbf{GX}$ unobserved; $S = T = 100$}                                                          \\[1ex]
$\alpha = 0.538$            & 0.538            &0.009     & 0.537            &0.013      & 0.538            &0.022     & 0.539           &0.042     \\
$\beta_1 = -0.072$          & -0.072           &0.009     & -0.073           &0.010      & -0.072           &0.012     & -0.072          &0.016     \\
$\beta_2 = 0.133$           & 0.133            &0.028     & 0.133            &0.030      & 0.133            &0.032     & 0.133           &0.032     \\
$\gamma_1 = 0.086$          & 0.086            &0.005     & 0.086            &0.007      & 0.086            &0.011     & 0.085           &0.019     \\
$\gamma_2 = -0.003$         & -0.001           &0.038     & -0.003           &0.046      & -0.004           &0.059     & 0.002           &0.085     \\\midrule
                            & \multicolumn{8}{c}{SGMM: $\mathbf{Gy}$ unobserved and $\mathbf{GX}$ observed; $S = T = 100$}                                                          \\[1ex]
$\alpha = 0.538$            & 0.538            &0.009     & 0.537            &0.018      & 0.538            &0.031     & 0.542           &0.076     \\
$\beta_1 = -0.072$          & -0.072           &0.009     & -0.074           &0.017      & -0.075           &0.022     & -0.073          &0.031     \\
$\beta_2 = 0.133$           & 0.133            &0.028     & 0.133            &0.047      & 0.133            &0.059     & 0.136           &0.066     \\
$\gamma_1 = 0.086$          & 0.086            &0.005     & 0.085            &0.010      & 0.084            &0.016     & 0.081           &0.036     \\
$\gamma_2 = -0.003$         & -0.001           &0.038     & -0.003           &0.057      & -0.007           &0.082     & 0.001           &0.140     \\\midrule
                            & \multicolumn{8}{c}{SGMM: $\mathbf{Gy}$ and $\mathbf{GX}$ unobserved; $R = S = T = 100$}                                                               \\[1ex]
$\alpha = 0.538$            & 0.538            &0.009     & 0.537            &0.018      & 0.538            &0.033     & 0.539           &0.084     \\
$\beta_1 = -0.072$          & -0.072           &0.009     & -0.074           &0.020      & -0.075           &0.027     & -0.074          &0.037     \\
$\beta_2 = 0.133$           & 0.133            &0.028     & 0.133            &0.055      & 0.133            &0.071     & 0.137           &0.078     \\
$\gamma_1 = 0.086$          & 0.086            &0.005     & 0.085            &0.011      & 0.084            &0.017     & 0.082           &0.039     \\
$\gamma_2 = -0.003$         & -0.001           &0.038     & -0.003           &0.060      & -0.007           &0.087     & 0.003           &0.153     \\\bottomrule
\multicolumn{9}{l}{%
  \begin{minipage}{12cm}%
  \vspace{0.1cm}
    \footnotesize{Note: We perform 1,000 simulations. 'Std' denotes the standard deviation.}
  \end{minipage}%
}
\end{tabular}
\end{table}

\begin{table}[!htbp]
\footnotesize
\centering
\caption{Full simulation results under misclassified links without group fixed effects}
\label{tab:simmclas:noFE}
\begin{tabular}{p{2.3cm}d{3}d{3}d{3}d{3}d{3}d{3}d{3}d{3}}
\toprule
False pos. rate         & \multicolumn{2}{c}{0\%}  & \multicolumn{2}{c}{0\%}  & \multicolumn{2}{c}{15\%} & \multicolumn{2}{c}{15\%}\\
False neg. rate         & \multicolumn{2}{c}{15\%} & \multicolumn{2}{c}{30\%} & \multicolumn{2}{c}{0\%}  & \multicolumn{2}{c}{15\%}\\
Statistic & \multicolumn{1}{c}{Mean} & \multicolumn{1}{c}{Std} & \multicolumn{1}{c}{Mean} & \multicolumn{1}{c}{Std} & \multicolumn{1}{c}{Mean} & \multicolumn{1}{c}{Std} & \multicolumn{1}{c}{Mean} & \multicolumn{1}{c}{Std}\\\midrule
&\multicolumn{8}{c}{Classical IV: $\mathbf{Gy}$ observed and $\mathbf{GX}$ unobserved; Instruments: $\mathbf{GX}^2$}                                           \\[1ex]
$\alpha = 0.538$    & 0.534           & 0.011       & 0.529           & 0.015       & 0.611           & 0.112       & 0.612            & 0.143        \\
$c = 3.806$         & 4.154           & 0.151       & 4.489           & 0.190       & 3.904           & 1.534       & 3.981            & 2.186        \\
$\beta_1 = -0.072$  & -0.082          & 0.010       & -0.093          & 0.012       & -0.087          & 0.068       & -0.083           & 0.097        \\
$\beta_2 = 0.133$   & 0.134           & 0.028       & 0.138           & 0.029       & 0.136           & 0.033       & 0.136            & 0.036        \\
$\gamma_1 = 0.086$  & 0.072           & 0.005       & 0.060           & 0.005       & 0.044           & 0.020       & 0.032            & 0.023        \\
$\gamma_2 = -0.003$ & -0.008          & 0.038       & -0.011          & 0.038       & -0.010          & 0.075       & -0.009           & 0.080        \\\midrule
&\multicolumn{8}{c}{Classical IV: $\mathbf{Gy}$ and $\mathbf{GX}$ unobserved; Instruments: $\mathbf{GX}^2$}                                                    \\[1ex]
$\alpha = 0.538$    & 0.479           & 0.012       & 0.424           & 0.015       & 0.366           & 0.168       & 0.275            & 0.171        \\
$c = 3.806$         & 5.538           & 0.276       & 7.100           & 0.359       & 8.870           & 1.563       & 9.876            & 1.502        \\
$\beta_1 = -0.072$  & -0.135          & 0.018       & -0.196          & 0.024       & -0.421          & 0.034       & -0.425           & 0.036        \\
$\beta_2 = 0.133$   & 0.143           & 0.048       & 0.158           & 0.060       & 0.191           & 0.087       & 0.191            & 0.088        \\
$\gamma_1 = 0.086$  & 0.049           & 0.007       & 0.022           & 0.009       & 0.064           & 0.049       & 0.036            & 0.050        \\
$\gamma_2 = -0.003$ & -0.022          & 0.043       & -0.035          & 0.048       & 0.022           & 0.194       & 0.026            & 0.197        \\\midrule
&\multicolumn{8}{c}{SGMM: $\mathbf{Gy}$ and $\mathbf{GX}$ observed; $T = 100$}                                                                                 \\[1ex]
$\alpha = 0.538$    & 0.537           & 0.011       & 0.538           & 0.013       & 0.538           & 0.020       & 0.537            & 0.023        \\
$c = 3.806$         & 3.809           & 0.133       & 3.808           & 0.139       & 3.802           & 0.151       & 3.814            & 0.156        \\
$\beta_1 = -0.072$  & -0.073          & 0.009       & -0.072          & 0.010       & -0.072          & 0.010       & -0.073           & 0.011        \\
$\beta_2 = 0.133$   & 0.132           & 0.027       & 0.133           & 0.026       & 0.132           & 0.027       & 0.133            & 0.025        \\
$\gamma_1 = 0.086$  & 0.086           & 0.006       & 0.086           & 0.006       & 0.086           & 0.009       & 0.086            & 0.010        \\
$\gamma_2 = -0.003$ & -0.003          & 0.037       & -0.004          & 0.037       & -0.003          & 0.038       & -0.004           & 0.038        \\\midrule
&\multicolumn{8}{c}{SGMM: $\mathbf{Gy}$ observed and $\mathbf{GX}$ unobserved; $S = T = 100$}                                                                  \\[1ex]
$\alpha = 0.538$    & 0.538           & 0.011       & 0.537           & 0.014       & 0.539           & 0.037       & 0.540            & 0.047        \\
$c = 3.806$         & 3.806           & 0.145       & 3.806           & 0.165       & 3.805           & 0.243       & 3.805            & 0.273        \\
$\beta_1 = -0.072$  & -0.072          & 0.010       & -0.072          & 0.011       & -0.072          & 0.015       & -0.072           & 0.017        \\
$\beta_2 = 0.133$   & 0.132           & 0.029       & 0.134           & 0.029       & 0.133           & 0.034       & 0.133            & 0.034        \\
$\gamma_1 = 0.086$  & 0.086           & 0.006       & 0.087           & 0.007       & 0.085           & 0.018       & 0.086            & 0.023        \\
$\gamma_2 = -0.003$ & -0.004          & 0.042       & -0.004          & 0.046       & -0.006          & 0.086       & -0.007           & 0.105        \\\midrule
&\multicolumn{8}{c}{SGMM: $\mathbf{Gy}$ unobserved and $\mathbf{GX}$ observed; $S = T = 100$}                                                                  \\[1ex]
$\alpha = 0.538$    & 0.538           & 0.012       & 0.537           & 0.018       & 0.539           & 0.077       & 0.536            & 0.103        \\
$c = 3.806$         & 3.803           & 0.232       & 3.799           & 0.292       & 3.812           & 0.501       & 3.812            & 0.622        \\
$\beta_1 = -0.072$  & -0.072          & 0.015       & -0.072          & 0.018       & -0.072          & 0.032       & -0.073           & 0.038        \\
$\beta_2 = 0.133$   & 0.132           & 0.041       & 0.135           & 0.048       & 0.135           & 0.067       & 0.134            & 0.070        \\
$\gamma_1 = 0.086$  & 0.086           & 0.007       & 0.087           & 0.010       & 0.085           & 0.039       & 0.088            & 0.050        \\
$\gamma_2 = -0.003$ & -0.005          & 0.047       & -0.004          & 0.055       & -0.005          & 0.170       & -0.010           & 0.209        \\\midrule
&\multicolumn{8}{c}{SGMM: $\mathbf{Gy}$ and $\mathbf{GX}$ unobserved; $R = S = T = 100$}                                                                       \\[1ex]
$\alpha = 0.538$    & 0.538           & 0.013       & 0.537           & 0.018       & 0.543           & 0.090       & 0.539            & 0.124        \\
$c = 3.806$         & 3.800           & 0.273       & 3.794           & 0.350       & 3.802           & 0.596       & 3.792            & 0.753        \\
$\beta_1 = -0.072$  & -0.072          & 0.017       & -0.072          & 0.022       & -0.071          & 0.038       & -0.072           & 0.047        \\
$\beta_2 = 0.133$   & 0.132           & 0.048       & 0.135           & 0.058       & 0.135           & 0.081       & 0.134            & 0.085        \\
$\gamma_1 = 0.086$  & 0.086           & 0.008       & 0.087           & 0.010       & 0.083           & 0.045       & 0.087            & 0.061        \\
$\gamma_2 = -0.003$ & -0.005          & 0.047       & -0.004          & 0.057       & -0.008          & 0.205       & -0.014           & 0.251        \\\bottomrule
\multicolumn{9}{l}{%
  \begin{minipage}{13cm}%
  \vspace{0.1cm}
    \small{Note: We perform 1,000 simulations. 'Std' denotes the standard deviation. `False pos. rate` refers to the proportion of false positives among actual negatives, which include true negatives and false positives. `False neg. rate` refers to the proportion of false negatives among actual positives, which include true positives and false negatives. A positive indicates a friendship, while a negative indicates a non-friendship.}
  \end{minipage}%
}
\end{tabular}
\end{table}

\begin{table}[!htbp]
\centering
\footnotesize
\caption{Full simulation results under misclassified links with group fixed effects}
\label{tab:simmclas:FE}
\begin{tabular}{p{2.3cm}d{3}d{3}d{3}d{3}d{3}d{3}d{3}d{3}}
\toprule
False pos. rate         & \multicolumn{2}{c}{0\%}  & \multicolumn{2}{c}{0\%}  & \multicolumn{2}{c}{15\%} & \multicolumn{2}{c}{15\%}\\
False neg. rate         & \multicolumn{2}{c}{15\%} & \multicolumn{2}{c}{30\%} & \multicolumn{2}{c}{0\%}  & \multicolumn{2}{c}{15\%}\\
Statistic & \multicolumn{1}{c}{Mean} & \multicolumn{1}{c}{Std} & \multicolumn{1}{c}{Mean} & \multicolumn{1}{c}{Std} & \multicolumn{1}{c}{Mean} & \multicolumn{1}{c}{Std} & \multicolumn{1}{c}{Mean} & \multicolumn{1}{c}{Std}\\\midrule
&\multicolumn{8}{c}{Classical IV: $\mathbf{Gy}$ observed and $\mathbf{GX}$ unobserved; Instruments: $\mathbf{GX}^2$}                                           \\[1ex]
$\alpha = 0.538$    & 0.532           & 0.020       & 0.525           & 0.026       & 0.655           & 0.172       & 0.645            & 0.217        \\
$\beta_1 = -0.072$  & -0.082          & 0.010       & -0.092          & 0.011       & -0.093          & 0.054       & -0.097           & 0.131        \\
$\beta_2 = 0.133$   & 0.135           & 0.028       & 0.138           & 0.029       & 0.138           & 0.033       & 0.135            & 0.036        \\
$\gamma_1 = 0.086$  & 0.072           & 0.004       & 0.059           & 0.005       & 0.042           & 0.018       & 0.033            & 0.023        \\
$\gamma_2 = -0.003$ & -0.007          & 0.039       & -0.010          & 0.039       & -0.009          & 0.083       & -0.006           & 0.095        \\\midrule
&\multicolumn{8}{c}{Classical IV: $\mathbf{Gy}$ and $\mathbf{GX}$ unobserved; Instruments: $\mathbf{GX}^2$}                                                    \\[1ex]
$\alpha = 0.538$    & 0.472           & 0.019       & 0.412           & 0.024       & 0.380           & 0.156       & 0.289            & 0.165        \\
$\beta_1 = -0.072$  & -0.111          & 0.013       & -0.147          & 0.016       & -0.280          & 0.023       & -0.282           & 0.025        \\
$\beta_2 = 0.133$   & 0.140           & 0.037       & 0.148           & 0.043       & 0.173           & 0.058       & 0.170            & 0.061        \\
$\gamma_1 = 0.086$  & 0.061           & 0.006       & 0.041           & 0.006       & 0.064           & 0.030       & 0.045            & 0.033        \\
$\gamma_2 = -0.003$ & -0.015          & 0.042       & -0.021          & 0.045       & 0.011           & 0.145       & 0.009            & 0.146        \\\midrule
&\multicolumn{8}{c}{SGMM: $\mathbf{Gy}$ and $\mathbf{GX}$ observed; $T = 100$}                                                                                 \\[1ex]
$\alpha = 0.538$    & 0.538           & 0.019       & 0.537           & 0.023       & 0.535           & 0.042       & 0.539            & 0.051        \\
$\beta_1 = -0.072$  & -0.072          & 0.010       & -0.072          & 0.010       & -0.072          & 0.011       & -0.072           & 0.011        \\
$\beta_2 = 0.133$   & 0.133           & 0.026       & 0.133           & 0.027       & 0.133           & 0.026       & 0.131            & 0.027        \\
$\gamma_1 = 0.086$  & 0.086           & 0.005       & 0.086           & 0.006       & 0.087           & 0.010       & 0.086            & 0.011        \\
$\gamma_2 = -0.003$ & -0.004          & 0.039       & -0.003          & 0.039       & -0.001          & 0.039       & -0.005           & 0.040        \\\midrule
&\multicolumn{8}{c}{SGMM: $\mathbf{Gy}$ observed and $\mathbf{GX}$ unobserved; $S = T = 100$}                                                                  \\[1ex]
$\alpha = 0.538$    & 0.538           & 0.019       & 0.537           & 0.024       & 0.533           & 0.075       & 0.532            & 0.109        \\
$\beta_1 = -0.072$  & -0.072          & 0.010       & -0.072          & 0.011       & -0.073          & 0.016       & -0.073           & 0.019        \\
$\beta_2 = 0.133$   & 0.133           & 0.028       & 0.134           & 0.029       & 0.134           & 0.033       & 0.131            & 0.035        \\
$\gamma_1 = 0.086$  & 0.086           & 0.005       & 0.086           & 0.006       & 0.087           & 0.018       & 0.088            & 0.026        \\
$\gamma_2 = -0.003$ & -0.004          & 0.044       & -0.004          & 0.048       & -0.002          & 0.100       & 0.001            & 0.115        \\\midrule
&\multicolumn{8}{c}{SGMM: $\mathbf{Gy}$ unobserved and $\mathbf{GX}$ observed; $S = T = 100$}                                                                  \\[1ex]
$\alpha = 0.538$    & 0.538           & 0.020       & 0.536           & 0.027       & 0.544           & 0.105       & 0.552            & 0.142        \\
$\beta_1 = -0.072$  & -0.072          & 0.011       & -0.072          & 0.013       & -0.072          & 0.022       & -0.071           & 0.026        \\
$\beta_2 = 0.133$   & 0.133           & 0.031       & 0.134           & 0.035       & 0.134           & 0.043       & 0.131            & 0.046        \\
$\gamma_1 = 0.086$  & 0.086           & 0.006       & 0.086           & 0.007       & 0.084           & 0.025       & 0.083            & 0.034        \\
$\gamma_2 = -0.003$ & -0.005          & 0.046       & -0.003          & 0.052       & -0.002          & 0.132       & -0.003           & 0.157        \\\midrule
&\multicolumn{8}{c}{SGMM: $\mathbf{Gy}$ and $\mathbf{GX}$ unobserved; $R = S = T = 100$}                                                                       \\[1ex]
$\alpha = 0.538$    & 0.538           & 0.021       & 0.536           & 0.028       & 0.540           & 0.126       & 0.537            & 0.178        \\
$\beta_1 = -0.072$  & -0.072          & 0.013       & -0.072          & 0.015       & -0.072          & 0.027       & -0.073           & 0.033        \\
$\beta_2 = 0.133$   & 0.133           & 0.036       & 0.134           & 0.042       & 0.135           & 0.055       & 0.131            & 0.059        \\
$\gamma_1 = 0.086$  & 0.086           & 0.006       & 0.086           & 0.008       & 0.085           & 0.030       & 0.087            & 0.042        \\
$\gamma_2 = -0.003$ & -0.004          & 0.047       & -0.003          & 0.054       & -0.003          & 0.164       & -0.002           & 0.198        \\\bottomrule
\multicolumn{9}{l}{%
  \begin{minipage}{13cm}%
  \vspace{0.1cm}
    \small{Note: We perform 1,000 simulations. 'Std' denotes the standard deviation. `False pos. rate` refers to the proportion of false positives among actual negatives, which include true negatives and false positives. `False neg. rate` refers to the proportion of false negatives among actual positives, which include true positives and false negatives. A positive indicates a friendship, while a negative indicates a non-friendship.}
  \end{minipage}%
}
\end{tabular}
\end{table}

\renewcommand{\arraystretch}{1}

%% file: OA_Bayesian.tex
\section{Bayesian estimator}\label{sec:OA_bayesian}

\input{appendix_bayesian}

\subsection{Network Sampling}\label{sec:networksampling}
This section explains how we sample the network in Algorithm \ref{algo:mcmc} using Gibbs sampling. As discussed above, a natural solution is to update only one entry of the adjacency matrix at every step $t$ of the MCMC. The entry $(i, j)$ is updated according to its conditional posterior distribution. For each entry, however, we need to compute $\mathcal{P}(\mathbf{y}|0,\mathbf{A}_{-ij})$ and $\mathcal{P}(\mathbf{y}|1,\mathbf{A}_{-ij})$, which are the respective likelihoods of replacing $a_{ij}$ by 0 or by 1. The likelihood computation requires the determinant of $(\mathbf{I}-\alpha\mathbf{G})$, which has a complexity $O(N^3)$ where N is the dimension of $\mathbf{G}$. This implies that we must compute $2N(N-1)$ times $\text{det}(\mathbf{I}-\alpha\mathbf{G})$ to update the adjacency matrix at each step of the MCMC. As $\bG$ is row-normalized, alternating any off-diagonal entry $(i,j)$ in $\bA$ between $0$ and $1$ perturbs all off-diagonal entries of the row $i$ in $(\mathbf{I}-\alpha\mathbf{G})$. We show that $\bA_{ij}$ and $\text{det}(\mathbf{I}-\alpha\mathbf{G})$ can be updated by computing a determinant of an auxiliary matrix that requires only updating two entries.

Assume that we want to update the entry $(i,j)$. Let $h$ be the function defined in $\mathbb{N}$ such that $\forall ~ x \in \mathbb{N}^*$, $h(x) = x$, and $h(0) = 1$. Let $\mathbf{L}$ be an N $\times$ N diagonal matrix, where $\mathbf{L}_{ii} = h(n_i)$, and $n_i$ stands for the degree of $i$, while $\mathbf{L}_{kk}=1$ for all $k\neq i$, and $\mathbf{W}$ is the matrix $\bG$ where the row $i$ of $\mathbf{W}$ is replaced by the row $i$ of $\bA$. Then, as the determinant is linear in each row, we can obtain $\mathbf{I}-\alpha\mathbf{G}$ by dividing the row $i$ of $\mathbf{L} - \alpha\mathbf{W}$ by $h(n_i)$. We get:
$$
\text{det}(\mathbf{I}-\alpha\mathbf{G}) = \dfrac{1}{h(n_i)}\text{det}(\mathbf{L}-\alpha\mathbf{W}).
$$
\noindent When $a_{ij}$ changes (from $0$ to $1$, or $1$ to $0$), note that only the entries $(i,i)$ and $(i,j)$ change in $\mathbf{L}-\alpha\mathbf{W}$. Two cases can be distinguished.
\begin{itemize}
    \item If $a_{ij} = 0$ before the update, then the new degree of $i$ will be $n_i + 1$. Thus, the entry $(i,i)$ in $\mathbf{L}-\alpha\mathbf{W}$ will change from $h(n_i)$ to $h(n_i + 1)$ (as the diagonal of $\mathbf{W}$ equals $0$), and the entry $(i,j)$ will change from $0$ to $-\alpha$. The new determinant is therefore given by
    $$\text{det}(\mathbf{I}-\alpha\mathbf{G}^*) = \dfrac{1}{h(n_i + 1)}\text{det}(\mathbf{L}^*-\alpha\mathbf{W}^*),$$ where $\mathbf{G}^*$, $\mathbf{L}^*$, and $\alpha\mathbf{W}^*$ are the new matrices once $a_{ij}$ has been updated.
    \item If $a_{ij} = 1$ before the update, then the new degree of $k$ will be $n_i - 1$. Thus, the entry $(i,i)$ in $\mathbf{L}-\alpha\mathbf{W}$ will change from $h(n_i)$ to $h(n_i - 1)$, and the entry $(i,j)$ will change from $-\alpha$ to $0$. The new determinant is therefore given by
    $$\text{det}(\mathbf{I}-\alpha\mathbf{G}^*) = \dfrac{1}{h(n_i - 1)}\text{det}(\mathbf{L}^*-\alpha\mathbf{W}^*).$$
\end{itemize}
\noindent Then, to update $\text{det}(\mathbf{L}-\alpha\mathbf{W})$ when only the entries $(i,i)$ and $(i,j)$ change, we adapt the Lemma 1 in \citeOA{hsieh2019structural} as follows:

\begin{proposition}\label{prop:det}
\noindent Let $\mathbf{e}_i$ be the $i$'th unit basis vector in $\mathbb{R}^N$. Let $\mathbf{M}$ denote an N $\times$ N matrix and $\mathbf{B}_{ij}(\mathbf{Q}, \epsilon)$ an N $\times$ N matrix as a function of an N $\times$ N matrix $\mathbf{Q}$ and a real value $\epsilon$, such that
\begin{equation}
    \mathbf{B}_{ij}(\mathbf{Q}, \epsilon) = \dfrac{\mathbf{Q} \mathbf{e}_i \mathbf{e}_j^{\prime} \mathbf{Q}}{1 + \epsilon \mathbf{e}_j^{\prime} \mathbf{Q} \mathbf{e}_i}.
\end{equation}
\noindent Adding a perturbation $\epsilon_1$ in the (i, i)th position and a perturbation $\epsilon_2$ in the (i, j)th position to the matrix $\mathbf{M}$ can be written as $\tilde{\mathbf{M}} = \mathbf{M} + \epsilon_1 \mathbf{e}_i \mathbf{e}_i^{\prime} +  \epsilon_2 \mathbf{e}_i \mathbf{e}_j^{\prime}$. 
\begin{enumerate}
    \item  The inverse of the perturbed matrix can be written as
    \begin{equation*}
    \tilde{\mathbf{M}}^{-1} = \mathbf{M}^{-1} - \epsilon_1\mathbf{B}_{ii}(\mathbf{M}^{-1}, \epsilon_1) -  \epsilon_2\mathbf{B}_{ij}\left(\mathbf{M}^{-1} - \epsilon_1\mathbf{B}_{ii}(\mathbf{M}^{-1}, \epsilon_1), \epsilon_2\right).
    \end{equation*}
    \item The determinant of the perturbed matrix can be written as
    \begin{equation*}
    \text{det}\left(\tilde{\mathbf{M}}\right) = \left(1 + \epsilon_2 \mathbf{e}_j^{\prime}\left(\mathbf{M}^{-1} - \epsilon_1\mathbf{B}_{ii}(\mathbf{M}^{-1}, \epsilon_1) \mathbf{e}_i\right)\right)(1 + \epsilon_1\mathbf{e}_i^{\prime}\mathbf{M}^{-1}\mathbf{e}_i)\text{det}\left(\mathbf{M}\right).
    \end{equation*}
\end{enumerate}
\end{proposition}
\begin{proof}
\begin{enumerate}
\item By the Sherman--Morrison formula \citepOA[][]{mele2017}, we have
\begin{equation*}
\left(\mathbf{M} + \epsilon \mathbf{e}_i \mathbf{e}_j^{\prime}\right)^{-1} = \mathbf{M}^{-1} - \epsilon \dfrac{\mathbf{M}^{-1} \mathbf{e}_i \mathbf{e}_j^{\prime} \mathbf{M}^{-1}}{1 + \epsilon \mathbf{e}_j^{\prime} \mathbf{M}^{-1} \mathbf{e}_i} = \mathbf{M}^{-1} - \epsilon\mathbf{B}_{ij}(\mathbf{M}, \epsilon).
\end{equation*}
Thus, 
\begin{eqnarray*}
 \tilde{\mathbf{M}}^{-1} &= \left((\mathbf{M} + \epsilon_1 \mathbf{e}_i \mathbf{e}_i^{\prime}) +  \epsilon_2 \mathbf{e}_i \mathbf{e}_j^{\prime}\right)^{-1},\\
 \tilde{\mathbf{M}}^{-1}   &= (\mathbf{M} + \epsilon_1 \mathbf{e}_i \mathbf{e}_i^{\prime})^{-1} - \epsilon_2\mathbf{B}_{ij}((\mathbf{M} + \epsilon_1 \mathbf{e}_i \mathbf{e}_i^{\prime})^{-1}, \epsilon_2),\\
 \tilde{\mathbf{M}}^{-1}   &=  \mathbf{M}^{-1} - \epsilon_1\mathbf{B}_{ii}(\mathbf{M}^{-1}, \epsilon_1) -  \epsilon_2\mathbf{B}_{ij}\left(\mathbf{M}^{-1} - \epsilon_1\mathbf{B}_{ii}(\mathbf{M}^{-1}, \epsilon_1), \epsilon_2\right).
\end{eqnarray*}
\item By the matrix determinant lemma \citepOA[][]{johnson1985matrix}, we have
\begin{equation*}
\text{det}\left(\mathbf{M} + \epsilon \mathbf{e}_i \mathbf{e}_j^{\prime}\right) = (1 + \epsilon \mathbf{e}_j^{\prime}\mathbf{M}^{-1}\mathbf{e}_i)\text{det}\left(\mathbf{M}\right).
\end{equation*}
\noindent It follows that
\begin{eqnarray*}
 \text{det}\left(\tilde{\mathbf{M}}\right) &= \text{det}\left((\mathbf{M} + \epsilon_1 \mathbf{e}_i \mathbf{e}_i^{\prime}) +  \epsilon_2 \mathbf{e}_i \mathbf{e}_j^{\prime}\right),\\
 \text{det}\left(\tilde{\mathbf{M}}\right)   &= (1 + \epsilon_2\mathbf{e}_j^{\prime} (\mathbf{M} + \epsilon_1 \mathbf{e}_i \mathbf{e}_i^{\prime})^{-1} \mathbf{e}_i) \text{det}\left(\mathbf{M} + \epsilon_1 \mathbf{e}_i \mathbf{e}_i^{\prime}\right),\\
 \text{det}\left(\tilde{\mathbf{M}}\right)  &=  \left(1 + \epsilon_2 \mathbf{e}_j^{\prime}\left(\mathbf{M}^{-1} - \epsilon_1\mathbf{B}_{ii}(\mathbf{M}^{-1}, \epsilon_1) \mathbf{e}_i\right)\right)(1 + \epsilon_1\mathbf{e}_i^{\prime}\mathbf{M}^{-1}\mathbf{e}_i)\text{det}\left(\mathbf{M}\right).
\end{eqnarray*}
\end{enumerate}
\end{proof}

\noindent The method proposed above becomes computationally intensive when many entries must be updated simultaneously.  We also propose an alternative method that allows updating the block for entries in $\bA$. Let $\mathbf{D}=(\mathbf{I}-\alpha\mathbf{G})$; we can write
\begin{equation}
\displaystyle \text{det}(\mathbf{D}) = \sum_{j = 1}^N(-1)^{i+j}\mathbf{D}_{\text{ij}}\delta_{\text{ij}},
\label{det}
\end{equation}
where $i$ denotes any row of $\mathbf{D}$ and $\delta_{ij}$ is the minor\footnote{The determinant of the submatrix of $\mathbf{M}$ by removing row $i$ and column $j$.} associated with the entry $(i, j)$. The minors of row $i$ do not depend on the values of entries in row $i$. To update any block in row $i$, we therefore compute the $N$ minors associated with $i$ and use this minor within the row. We can then update many entries simultaneously without increasing the number of times that we compute $\text{det}(\mathbf{D})$.

\noindent One possibility is to update multiple links simultaneously by randomly choosing the number of entries to consider and their position in the row. As suggested by \citeOA{chib2010tailored}, this method would help the Gibbs sampling to converge more quickly. We can summarize how we update the row $i$ as follows:
\begin{enumerate}
    \item Compute the $N$ minors $\delta_{i1}, \dots, \delta_{in}$ \label{update::minor}.
	\item Let $\Omega_{\bG}$ be the entries to update in the row $i$, and $n_{\bG} = |\Omega_{\bG}|$ the number of entries in $\Omega_{\bG}$.
	\begin{enumerate}
	    \item Choose $r$, the size of the block to update, as a random integer number such that $1 \leq r \leq  n_{\bG}$. In practice, we choose $r \leq \min(5, n_{\bG})$ because the number of possibilities of links to consider grows exponentially with $r$. \label{update::block::size}
    	\item Choose the $r$ random entries from $\Omega_{\bG}$. These entries define the block to update. \label{update::entries}
    	\item Compute the posterior probabilities of all possibilities of links inside the block and update the block (there are $2^r$ possibilities). Use the minors calculated at \ref{update::minor} and the formula (\ref{det}) to quickly compute $det(\mathbf{D})$.  \label{update::post::prob}
    	\item  Remove the $r$ drawn positions from $\Omega_{\bG}$ and let  $n_{\bG} = n_{\bG} - r$.  Replicate \ref{update::block::size},  \ref{update::entries}, and \ref{update::post::prob} until $n_{\bG} = 0$.
	\end{enumerate}
\end{enumerate}

\subsection{How to build prior distributions}\label{sec:appendix_prior}
The two following examples discuss how to construct prior distributions depending on whether the first stage is estimated by a classical or Bayesian estimator.

\begin{example}[Priors from the Asymptotic Distribution of $\boldsymbol{\rho}$]\label{ex:prior_distr}
In a classical setting, and under the usual assumptions, the estimation of (\ref{eq:gennetfor}) produces an estimator $\hat{\boldsymbol{\rho}}$ of $\boldsymbol{\rho}_0$ and an estimator of the asymptotic variance of $\hat{\boldsymbol{\rho}}$, i.e., $\hat{\mathbf{V}}(\hat{\boldsymbol{\rho}})$. In this case, we define the prior density $\pi(\boldsymbol{\rho})$ as the density of a multivariate normal distribution with mean $\hat{\boldsymbol{\rho}}$ and variance--covariance matrix $\hat{\mathbf{V}}(\hat{\boldsymbol{\rho}})$.
\end{example}

\begin{example}[Priors from the Posterior Distribution of $\boldsymbol{\rho}$]\label{ex:prior_post}
In a Bayesian setting, the estimation of $\boldsymbol{\rho}$ from the network formation model (\ref{eq:gennetfor}) results in draws from the posterior distribution of $\boldsymbol{\rho}$. It is therefore natural to use such a posterior distribution as the prior distribution of $\mathbf{A}$ for the estimation based on (\ref{eq:likelihood}). Performing such a sequential Bayesian updating approach comes with a well-known numerical issue.\footnote{See \citeOA{thijssen2020approximating} for a recent discussion.}

Indeed, the evaluation of the acceptance ratio in Step 1 of Algorithm \ref{algo:mcmc} below requires the evaluation of the density of $\boldsymbol{\rho}$ at different values. Ideally, one would use the draws from the posterior distribution of $\boldsymbol{\rho}$ from the first step (network formation model) and perform a nonparametric kernel density estimation of the posterior distribution. However, when the dimension of $\boldsymbol{\rho}$ is large, the kernel density estimation may be infeasible in practice.

This is especially true for very flexible network formation models, such as that proposed by \citeOA{breza2017using} for which the number of parameters to estimate is $O(N_m)$. In such a case, it might be more reasonable to use a more parametric approach or to impose additional restrictions on the dependence structure of $\boldsymbol{\rho}$ across dimensions.\footnote{For example, if we assume that the posterior distribution of $\boldsymbol{\rho}$ is jointly normal, the estimation of the mean and variance-covariance matrix is straightforward, even in a high-dimensional setting. Simulations suggest that this approach performs well in practice. See the Vignette accompanying our R package.}
\end{example}

\subsection{Simulation Study}
In this section, we assess the performance of the Bayesian estimator through a simulation study. We consider the case of missing links described in Section~\ref{sec:montecarlo} for the SGMM estimator, with the same distribution for the variables in $\mathbf{X}_m$ and the network, and the same true parameter value:  $(\alpha, \boldsymbol{\beta}^{\prime}, \boldsymbol{\gamma}^{\prime}) = (0.538,~3.806,~-0.072,~0.132,~0.086,~-0.003)$, $\sigma = 0.707$, and $\boldsymbol{\rho}^{\prime} = (-2.349,~-0.700,~0.404)$.  We also set $M = 50$ and $N_m \in \{30,~50\}$. We consider scenarios where $50\%$ or $75\%$ of the entries $a_{ij,m}$ are missing at random.

For each of the resulting scenarios, we perform 19 simulations, varying the starting value for $\alpha$ in the MCMC process from $-0.9$ to $0.9$ in increments of $0.1$. The starting values for the other model parameters are set to zero, except for $\sigma$, which is set to one. The starting values for the unobserved network entries are randomly generated from a Bernoulli distribution with parameter $0.2$.

The simulation results displayed in Figure~\ref{fig:Bayesian} indicate that the MCMC converges to a stationary distribution regardless of the starting value for $\alpha$. The burn-in period increases with the number of nodes in the network, and especially with the number of missing entries in the network. However, given the proportion of missing entries in the network, convergence occurs quickly—before 2{,}000 iterations.
\begin{figure}[!h]
    \centering
    \includegraphics[scale = 0.8]{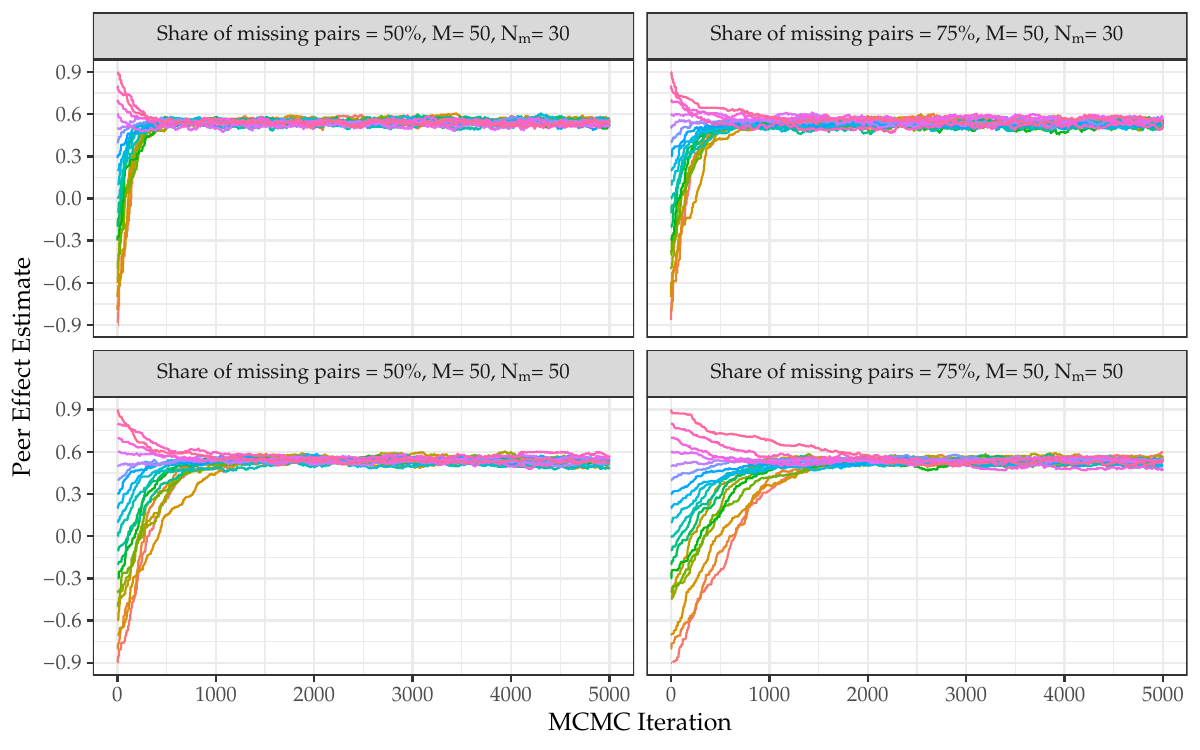}
    \caption{Bayesian estimation of peer effects with mismeasured links}
    \label{fig:Bayesian}
\end{figure}

%% file: appendix_bayesian.tex

\subsection{Posterior Distributions for Algorithm \ref{algo:mcmc}.}\label{sec:mcmcdetails}
To compute the posterior distributions, we set prior distributions on $\Tilde{\alpha}$, $\blambda$, and $\sigma^2$, where $\Tilde{\alpha} = \log(\frac{\alpha}{1-\alpha})$ and $\blambda = [\bbeta, \bgamma]$. In Algorithm \ref{algo:mcmc}, we therefore sample $\Tilde{\alpha}$ and compute $\alpha$, such that $\alpha = \dfrac{\exp(\Tilde{\alpha})}{1 + \exp(\Tilde{\alpha})}$. Using this functional form for computing $\alpha$ ensures that $\alpha \in (0,1)$. The prior distributions are set as follows: 
\begin{align*}
   \Tilde{\alpha} &\sim \mathcal{N}(\mu_{\Tilde{\alpha}}, \sigma^2_{\Tilde{\alpha}}),\\
    \blambda | \sigma^2&\sim \mathcal{N}(\boldsymbol{\mu_{\blambda}}, \sigma^2\boldsymbol{\Sigma_{\blambda}}),\\
    \sigma^2 &\sim IG(\dfrac{a}{2}, \dfrac{b}{2}).
\end{align*}
For the simulations and estimations in this paper, we set $\mu_{\Tilde{\alpha}} = -1$, $\sigma^{-2}_{\Tilde{\alpha}} = 2$,  $\boldsymbol{\mu_{\blambda}} = \boldsymbol{0}$, $\boldsymbol{\Sigma_{\blambda}}^{-1} = \dfrac{1}{100}\boldsymbol{I}_K$, $a = 4$, and $b = 4$, where $\boldsymbol{I}_K$ is the identity matrix of dimension K and $K = \dim(\blambda)$.\\
Following Algorithm \ref{algo:mcmc}, $\alpha$ is updated at each iteration $t$ of the MCMC by drawing $
\Tilde{\alpha}^*$ from the proposal $\mathcal{N}(\Tilde{\alpha}_{t-1}, \xi_t)$, where the jumping scale $\xi_t$ is also updated at each $t$ following \citeOA{atchade2005adaptive} for an acceptance rate of $a^*$ targeted at 0.44. As the proposal is symmetrical, $\alpha^* = \dfrac{\exp(\Tilde{\alpha}^*)}{1 + \exp(\Tilde{\alpha}^*)}$ is accepted with the probability
$$
    \min\left\lbrace 1,\frac{\mathcal{P}(\mathbf{y}|\mathbf{A}_{t}, \blambda_{t-1},\alpha^*)P(\Tilde{\alpha}^*)}{\mathcal{P}(\mathbf{y}|\mathbf{A}_t, \boldsymbol{\theta}_{t-1})P(\Tilde{\alpha_t})}\right\rbrace.
$$
The parameters $\blambda_t = [\bbeta_t$, $\bgamma_t]$ and $\sigma^2_t$ are drawn from their posterior conditional distributions, given as follows:
\begin{eqnarray*}
    \blambda_t |\by, \bA_t, \alpha_t, \sigma^2_{t-1}  &\sim\mathcal{N}(\hat{\bmu}_{{\blambda}_t}, \sigma^2_{t-1}\hat{\boldsymbol{\Sigma}}_{{\blambda}_t}),\\
    \sigma^2_t|  \by, \bA_t, \btheta_t &\sim IG\left(\dfrac{\hat{a}_t}{2}, \dfrac{\hat{b}_t}{2}\right),
\end{eqnarray*}
where,
\begin{eqnarray*}
    \hat{\boldsymbol{\Sigma}}_{{\blambda}_t}^{-1} &= \mathbf{V}_t^{\prime}\mathbf{V}_t + \boldsymbol{\Sigma}_{\blambda}^{-1},\\
    \hat{\bmu}_{{\blambda}_t} &= \hat{\boldsymbol{\Sigma}}_{{\blambda}_t}\left(\mathbf{V}_t^{\prime}(\by - \alpha_t\boldsymbol{G}_t\by) + \boldsymbol{\Sigma}_{\blambda}^{-1}\bmu_{\blambda}\right),\\
    \hat{a}_t &= a + N,\\
    \hat{b}_t &= b + (\blambda_t - \bmu_{\blambda})^{\prime}\boldsymbol{\Sigma}^{-1}_{\blambda}(\blambda_t - \bmu_{\blambda})  + (\by - \alpha_t\bG_t\by - \mathbf{V}_t\blambda_t)^{\prime}(\by - \alpha_t\bG_t\by - \mathbf{V}_t\blambda_t),\\
    \mathbf{V}_t &= [\boldsymbol{1}, ~ \bX, ~ \bG_t\bX].
\end{eqnarray*}

%% file: OA_application.tex
\section{Application}\label{sec:OAapplication}

\input{appendix_addhealth}

\begin{figure}[!htbp]
    \centering
    \caption{MCMC Simulations --  Peer Effect Model}
    \label{fig:mcmc:theta}
    \includegraphics[scale = 0.40]{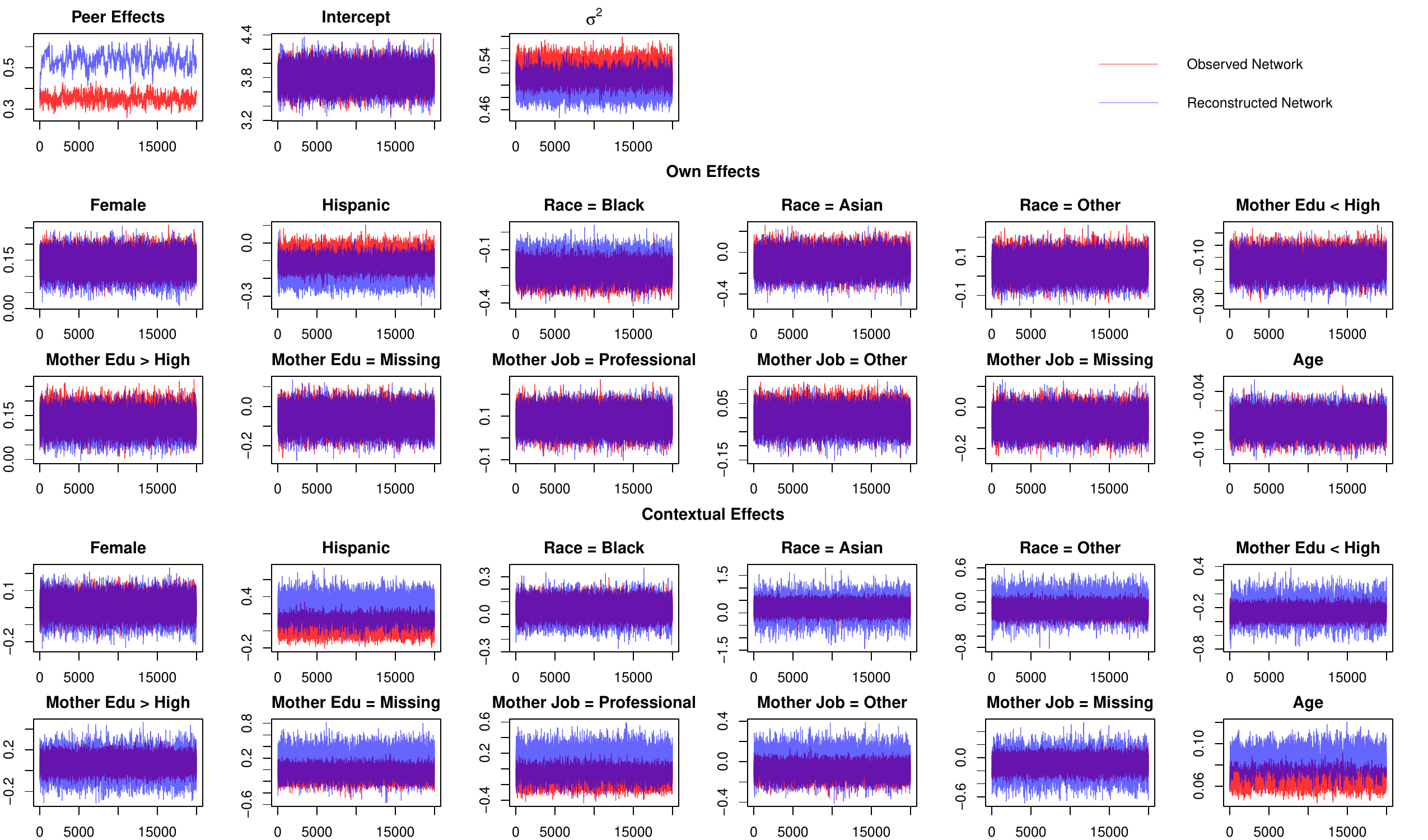}
\end{figure}

\begin{figure}[!htbp]
    \centering
    \caption{MCMC Simulations --  Network Formation Model}
    \label{fig:mcmc:rho}
    \includegraphics[scale = 0.55]{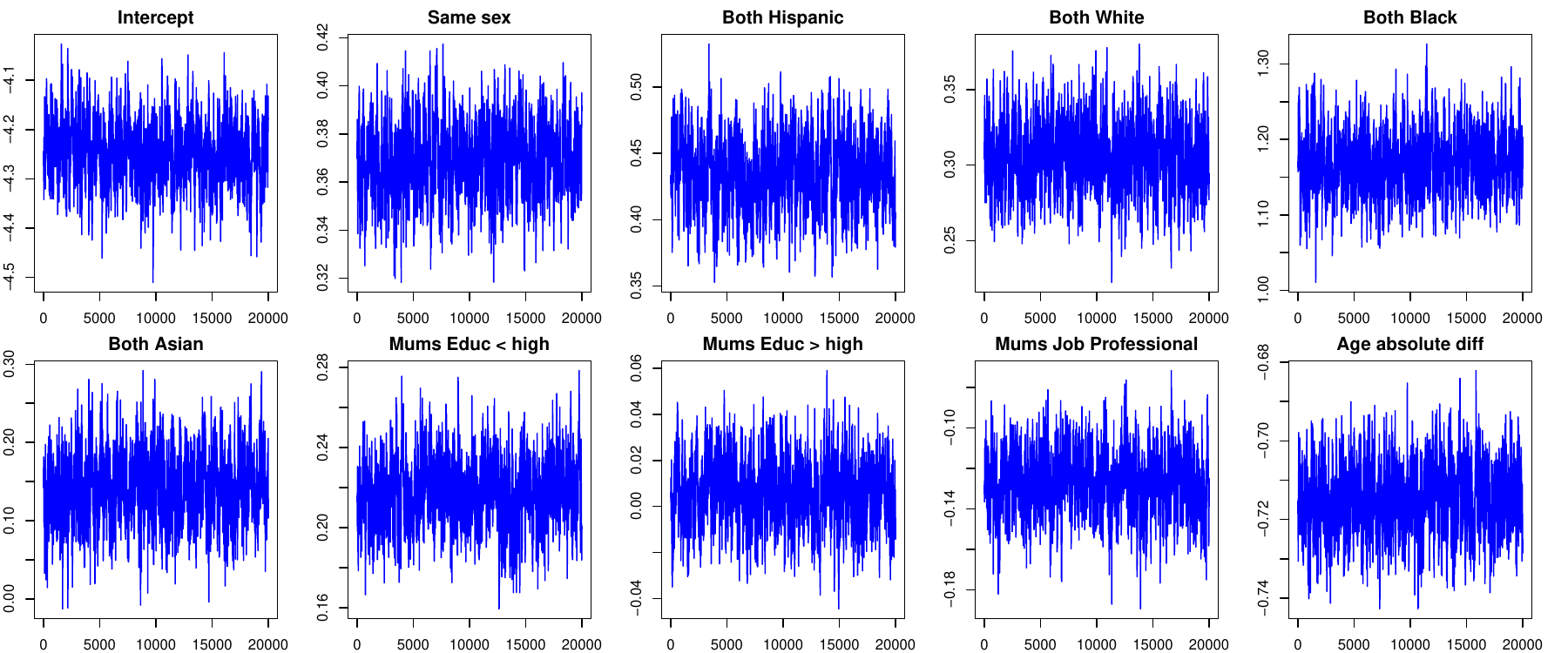}
\end{figure}

%% file: appendix_addhealth.tex

\subsection{Error codes only}\label{sec:manski1}
Each student nominates their best friends up to 5 males and 5 females. Because we know the sex of nominated friends, even when the identifier is coded with errors, we associate each missing link to a male or female student. We then have two sets of network data for each student $i$: the set of data from $i$ to their male schoolmates and the set of data from $i$ to their female schoolmates. A set is considered fully observed if it has no missing values. We estimate the network formation only using the fully observed sets. The sets with partial or no observed data are inferred (even the data we do not doubt in those sets are inferred).

This approach raises a selection problem that we address by weighting each selected set, following \citeOA{manski1977estimation}. The intuition of the weights lies in the fact that the sets with many links have lower probabilities to be selected (because error codes are more likely). The weight is the inverse of the selection probability. For a selected set $\mathcal{S}_{is}$ (of network data from $i$ to schoolmates of sex $s$),  the selection probability can be estimated as the proportion of 
 sets without missing data among the sets of network data to schoolmates of sex $s$ having the same number of links than $\mathcal{S}_{is}$.
 
For the Bayesian estimator, we jointly estimate the peer effect model and the network formation model (i.e., using Step 1' on Page 28). Thus, in the MCMC, $\boldsymbol{\rho}$ and the sets with partial or no network data are inferred using information from the weighted sets and the peer effect model.

\subsection{Error codes and top coding}\label{sec:manski2}
 We consider the same selected sets as in the case of missing data only. However, we doubt the exactitude of a link $a_{ij} \in \mathcal{S}_{is}$ if $a_{ij} = 0$ and the number of links in $\mathcal{S}_{is}$ is five. Therefore, if the number of links in $\mathcal{S}_{is}$ is five, we adjust the weight associated with each $a_{ij}$. For $a_{ij} = 0$, we multiply the weight obtained in the case of missing data only by $(\lvert \mathcal{S}_{is} \rvert - \ell(\mathcal{S}_{is}))/(\lvert \mathcal{S}_{is} \rvert - \hat{\ell}(\mathcal{S}_{is}))$, and for $a_{ij} = 1$, we multiply the weight by $\ell(\mathcal{S}_{is})/\hat{\ell}(\mathcal{S}_{is})$, where $\ell(\mathcal{S}_{is})$ is the estimate of the true number of links from $i$ to their schoolmates of sex $s$, $\hat{\ell}(\mathcal{S}_{is})$ if the number of links declared in $\mathcal{S}_{is}$, and $\lvert \mathcal{S}_{is} \rvert$ is the number of data in $\mathcal{S}_{is}$ (number of students having the sex $s$ in the school minus one).

We denote $s = m$ for males and $s = f$ for females. Fours scenarios are possible: $\{\hat{\ell}(\mathcal{S}_{im}) < 5, ~\hat{\ell}(\mathcal{S}_{if}) < 5\}$, $\{\hat{\ell}(\mathcal{S}_{im}) = 5, ~\hat{\ell}(\mathcal{S}_{if}) < 5\}$, $\{\hat{\ell}(\mathcal{S}_{im}) < 5, ~\hat{\ell}(\mathcal{S}_{if}) = 5\}$, and $\{\hat{\ell}(\mathcal{S}_{im}) = 5, ~\hat{\ell}(\mathcal{S}_{if}) = 5\}$. In the last three cases, $\ell(\mathcal{S}_{im}) + \ell(\mathcal{S}_{if})$ is left-censored and we know the lower bound. Assuming that the number of links $i$ follows a Poisson distribution of mean $n_i^e$, we estimate $n_i^e$ using a censored Poisson regression on the declared number of links. We assume that $n_i^e$ is an exponential linear function of $i$'s characteristics (age, sex, ...), and we also include school-fixed effects to control for school size. 

The estimate of $n_i^e$ is $\ell(\mathcal{S}_{im}) + \ell(\mathcal{S}_{if})$, and it allows us to compute $\ell(\mathcal{S}_{im})$ and $\ell(\mathcal{S}_{if})$. For the  case $\{\hat{\ell}(\mathcal{S}_{im}) = 5, ~\hat{\ell}(\mathcal{S}_{if}) = 5\}$, we assume that $\ell(\mathcal{S}_{im}) = \ell(\mathcal{S}_{if}) = 0.5(\ell(\mathcal{S}_{im}) + \ell(\mathcal{S}_{if}))$. In the other cases, as either $\ell(\mathcal{S}_{im})$ or $\ell(\mathcal{S}_{if})$ is known, the second member of $\ell(\mathcal{S}_{im}) + \ell(\mathcal{S}_{if})$ can be computed. 

For the Bayesian estimator, and contrary to the case with error codes only, it is more challenging to jointly estimate the peer effect model and the network formation model in a single step. We therefore first estimate the network formation model and then the Bayesian estimator (i.e., using Step 1 in Algorithm \ref{algo:mcmc}).
Thus, for the MCMC, the estimated distribution of $\boldsymbol{\rho}$ from the network formation model is used as a prior distribution. We then infer  $\boldsymbol{\rho}$ and the network data $a_{ij} = 0$ that we are doubtful about, using information from the peer effect model and the prior distribution of $\boldsymbol{\rho}$.

\subsection{Tables}\label{sec:app_tables}

\begin{table}[!htbp] 
   \centering 
  \caption{Summary statistics} 
  \label{tab:sumstats} 
\begin{tabular}{@{\extracolsep{5pt}} lld{3}ccc} 
\\[-1.8ex]\hline 
\hline \\[-1.8ex] 
\multicolumn{2}{l}{Statistic} & \multicolumn{1}{c}{Mean} & \multicolumn{1}{c}{Std. Dev.} & \multicolumn{1}{c}{Pctl(25)} & \multicolumn{1}{c}{Pctl(75)} \\ 
\hline \\[-1.8ex] 
\multicolumn{2}{l}{Female} & 0.540 & 0.498 & 0 & 1 \\ 
\multicolumn{2}{l}{Hispanic}  & 0.157 & 0.364 & 0 & 0 \\ 
\multicolumn{2}{l}{Race} &  &  &  &  \\ 
 & White & 0.612 & 0.487 & 0 & 1 \\ 
 & Black & 0.246 & 0.431 & 0 & 0 \\ 
 & Asian & 0.022 & 0.147 & 0 & 0 \\ 
 & Other & 0.088 & 0.283 & 0 & 0 \\ 
\multicolumn{2}{l}{Mother's education}  &  &  &  &  \\ 
 & High & 0.310 & 0.462 & 0 & 1 \\ 
 & <High & 0.193 & 0.395 & 0 & 0 \\ 
 & >High & 0.358 & 0.480 & 0 & 1 \\ 
 & Missing & 0.139 & 0.346 & 0 & 0 \\ 
\multicolumn{2}{l}{Mother's job}  &  &  &  &  \\ 
 & Stay-at-home & 0.225 & 0.417 & 0 & 0 \\ 
 & Professional & 0.175 & 0.380 & 0 & 0 \\ 
 & Other & 0.401 & 0.490 & 0 & 1 \\ 
 & Missing & 0.199 & 0.399 & 0 & 0 \\ 
\multicolumn{2}{l}{Age}  & 13.620 & 1.526 & 13 & 14 \\ 
\multicolumn{2}{l}{GPA}  & 2.912    &0.794    &2.333      &3.5\\ 
\hline \\[-1.8ex] 
\multicolumn{6}{l}{%
  \begin{minipage}{11.6cm}%
    \footnotesize{Note: We only keep the 33 schools having less than 200 students from the In-School sample. The variable GPA is computed by taking the average grade for English, Mathematics, History, and Science, letting $A=4$, $B=3$, $C=2$, and $D=1$. Thus, higher scores indicate better academic achievement.}
  \end{minipage}%
  }
\end{tabular} 

\end{table}

\begin{table}[htbp]
\small
\centering
\caption{Empirical results (Bayesian method)} 
\label{tab:appendapplicationB} 
    \begin{adjustbox}{max width=0.8\textwidth}
\begin{tabular}{lp{6cm}d{4}d{4}d{4}d{4}d{4}d{4}}
\\[-1.8ex]
\toprule
                    &                             & \multicolumn{2}{c}{Model 1} & \multicolumn{2}{c}{Model 2} & \multicolumn{2}{c}{Model 3} \\
\multicolumn{2}{l}{Statistic}                     & \multicolumn{1}{c}{Mean}        & \multicolumn{1}{c}{Std}       & \multicolumn{1}{c}{Mean}        & \multicolumn{1}{c}{Std}      & \multicolumn{1}{c}{Mean}        & \multicolumn{1}{c}{Std}       \\ \midrule
                    &                              \multicolumn{7}{c}{\textbf{Peer effect model}}                         \\
\multicolumn{2}{l}{Peer effects}                & 0.350       &  0.024        & 0.524       &  0.036        & 0.538       &  0.037        \\
\multicolumn{2}{l}{\textbf{Own effects}}        &             &               &             &               &             &               \\
\multicolumn{2}{l}{Female}                      & 0.144       &  0.029        & 0.135       &  0.030        & 0.133       &  0.031        \\
\multicolumn{2}{l}{Hispanic}                    & -0.083      &  0.042        & -0.148      &  0.048        & -0.151      &  0.047        \\
\multicolumn{2}{l}{Race (White)}                &             &               &             &               &             &               \\
                   & Black                      & -0.230      &  0.045        & -0.190      &  0.055        & -0.189      &  0.055        \\
                   & Asian                      & -0.091      &  0.089        & -0.113      &  0.091        & -0.110      &  0.091        \\
                   & Other                      & 0.055       &  0.051        & 0.039       &  0.052        & 0.039       &  0.052        \\
\multicolumn{2}{l}{Mother's education (High)}   &             &               &             &               &             &               \\
                   & \textless High             & -0.122      &  0.039        & -0.138      &  0.040        & -0.139      &  0.040        \\
                   & \textgreater High          & 0.140       &  0.034        & 0.123       &  0.034        & 0.121       &  0.034        \\
                   & Missing                    & -0.060      &  0.050        & -0.069      &  0.051        & -0.070      &  0.051        \\
\multicolumn{2}{l}{Mother's job (Stay-at-home)} &             &               &             &               &             &               \\
                   & Professional               & 0.080       &  0.045        & 0.075       &  0.044        & 0.079       &  0.044        \\
                   & Other                      & 0.003       &  0.035        & -0.014      &  0.035        & -0.012      &  0.035        \\
                   & Missing                    & -0.066      &  0.047        & -0.074      &  0.048        & -0.073      &  0.048        \\
\multicolumn{2}{l}{Age}                         & -0.073      &  0.010        & -0.071      &  0.010        & -0.072      &  0.010        \\
\multicolumn{2}{l}{\textbf{Contextual effects}} &             &               &             &               &             &               \\
\multicolumn{2}{l}{Female}                      & 0.011       &  0.049        & -0.003      &  0.060        & -0.003      &  0.060        \\
\multicolumn{2}{l}{Hispanic}                    & 0.060       &  0.069        & 0.272       &  0.102        & 0.276       &  0.105        \\
\multicolumn{2}{l}{Race (White)}                &             &               &             &               &             &               \\
                   & Black                      & 0.050       &  0.058        & 0.025       &  0.073        & 0.033       &  0.074        \\
                   & Asian                      & 0.209       &  0.186        & 0.110       &  0.365        & 0.209       &  0.385        \\
                   & Other                      & -0.137      &  0.089        & -0.044      &  0.163        & -0.051      &  0.167        \\
\multicolumn{2}{l}{Mother's education (High)}   &             &               &             &               &             &               \\
                   & \textless High             & -0.269      &  0.070        & -0.228      &  0.141        & -0.221      &  0.149        \\
                   & \textgreater High          & 0.072       &  0.059        & 0.063       &  0.097        & 0.057       &  0.102        \\
                   & Missing                    & -0.077      &  0.093        & 0.107       &  0.167        & 0.124       &  0.174        \\
\multicolumn{2}{l}{Mother's job (Stay-at-home)} &             &               &             &               &             &               \\
                   & Professional               & -0.110      &  0.08)        & 0.102       &  0.124        & 0.090       &  0.134        \\
                   & Other                      & -0.101      &  0.060        & -0.003      &  0.100        & -0.017      &  0.103        \\
                   & Missing                    & -0.093      &  0.085        & -0.075      &  0.157        & -0.109      &  0.165        \\
\multicolumn{2}{l}{Age}                         & 0.066       &  0.006        & 0.083       &  0.008        & 0.086       &  0.009        \\
SE$^2$             &                            & 0.523       &               & 0.496       &               & 0.499       &               \\\midrule
                    & \multicolumn{7}{c}{\textbf{Network formation model}}                                                  \\
\multicolumn{2}{l}{Same sex}                      &             &               & 0.310       &  0.011        & 0.370       &  0.014        \\
\multicolumn{2}{l}{Both Hispanic}                 &             &               & 0.416       &  0.020        & 0.436       &  0.026        \\
\multicolumn{2}{l}{Both White}                    &             &               & 0.312       &  0.018        & 0.304       &  0.023        \\
\multicolumn{2}{l}{Both Black}                    &             &               & 1.076       &  0.030        & 1.173       &  0.038        \\
\multicolumn{2}{l}{Both Asian}                    &             &               & 0.164       &  0.034        & 0.144       &  0.043        \\
\multicolumn{2}{l}{Mother's education $<$ High} &             &               & 0.226       &  0.013        & 0.218       &  0.017        \\
\multicolumn{2}{l}{Mother's education $>$ High} &             &               & 0.007       &  0.012        & 0.005       &  0.014        \\
\multicolumn{2}{l}{Mother's job Professional}      &             &               & -0.116      &  0.012        & -0.128      &  0.016        \\
\multicolumn{2}{l}{Age absolute difference}       &             &               & -0.700      &  0.007        & -0.714      &  0.009        \\
\multicolumn{2}{l}{}                              &             &               &             &               &             &               \\
\multicolumn{2}{l}{Average number of friends}     & 3.251       &               & 4.665       &               & 5.618       & \\
\bottomrule

\multicolumn{8}{l}{%
  \begin{minipage}{17cm}%
  \vspace{0.1cm}
    \footnotesize{Note: Model 1 considers the observed network as given. Model 2 infers the missing links due to friendship nominations coded with error, and Model 3 infers the missing links due to friendship nominations coded with error and controls for top coding.
    For each model, Column "Mean" indicates the posterior mean, and Column "Std" indicates the posterior standard deviations in parentheses.\\
    \textit{N} = 3,126. Observed links = 17,993. 
    Proportion of inferred network data: error code = 60.0\%, error code and top coding = 65.0\%. The explained variable is computed by taking the average grade for English, Mathematics, History, and Science, letting $A=4$, $B=3$, $C=2$, and $D=1$. Thus, higher scores indicate better academic achievement.}
  \end{minipage}%
}\\
\end{tabular} 
\end{adjustbox}

\end{table}

\begin{table}[htbp]
\small
\centering
\caption{Empirical results (SGMM Method)} 
\label{tab:appendapplicationG} 
    \begin{adjustbox}{max width=0.8\textwidth}
\begin{tabular}{lp{6cm}d{4}d{4}d{4}d{4}d{4}d{4}}
\\[-1.8ex]
\toprule
                    &                             & \multicolumn{2}{c}{Model 1} & \multicolumn{2}{c}{Model 2} & \multicolumn{2}{c}{Model 3} \\
\multicolumn{2}{l}{Statistic}                     & \multicolumn{1}{c}{Mean}        & \multicolumn{1}{c}{Std}       & \multicolumn{1}{c}{Mean}        & \multicolumn{1}{c}{Std}      & \multicolumn{1}{c}{Mean}        & \multicolumn{1}{c}{Std}       \\ \midrule
                    &                              \multicolumn{7}{c}{\textbf{Peer effect model}}                         \\
\multicolumn{2}{l}{Peer effects}                & 0.455       &  0.230        & 0.753       &  0.254        & 0.683       &  0.242        \\
\multicolumn{2}{l}{\textbf{Own effects}}        &             &               &             &               &             &               \\
\multicolumn{2}{l}{Female}                      & 0.179       &  0.039        & 0.122       &  0.036        & 0.122       &  0.035        \\
\multicolumn{2}{l}{Hispanic}                    & -0.129      &  0.045        & -0.160      &  0.051        & -0.160      &  0.051        \\
\multicolumn{2}{l}{Race (White)}                &             &               &             &               &             &               \\
                & Black                         & -0.276      &  0.058        & -0.172      &  0.058        & -0.166      &  0.059        \\
                & Asian                         & -0.178      &  0.101        & -0.131      &  0.085        & -0.124      &  0.086        \\
                & Other                         & 0.087       &  0.062        & 0.023       &  0.061        & 0.023       &  0.062        \\
\multicolumn{2}{l}{Mother's education (High)}   &             &               &             &               &             &               \\
                & \textless High                & -0.134      &  0.044        & -0.121      &  0.046        & -0.124      &  0.047        \\
                & \textgreater High             & 0.109       &  0.036        & 0.121       &  0.030        & 0.122       &  0.03)        \\
                & Missing                       & -0.066      &  0.053        & -0.060      &  0.051        & -0.062      &  0.051        \\
\multicolumn{2}{l}{Mother's job (Stay-at-home)} &             &               &             &               &             &               \\
                & Professional                  & 0.145       &  0.055        & 0.065       &  0.043        & 0.071       &  0.043        \\
                & Other                         & 0.043       &  0.035        & -0.019      &  0.031        & -0.018      &  0.030        \\
                & Missing                       & -0.018      &  0.045        & -0.072      &  0.043        & -0.068      &  0.043        \\
\multicolumn{2}{l}{Age}                         & -0.042      &  0.032        & -0.072      &  0.015        & -0.068      &  0.016        \\
\multicolumn{2}{l}{\textbf{Contextual effects}} &             &               &             &               &             &               \\
\multicolumn{2}{l}{Female}                      & -0.056      &  0.074        & -0.014      &  0.068        & -0.001      &  0.068        \\
\multicolumn{2}{l}{Hispanic}                    & 0.265       &  0.121        & 0.331       &  0.169        & 0.368       &  0.175        \\
\multicolumn{2}{l}{Race (White)}                &             &               &             &               &             &               \\
                & Black                         & 0.129       &  0.125        & 0.035       &  0.113        & 0.013       &  0.108        \\
                & Asian                         & 2.409       &  1.220        & 3.236       &  2.359        & 3.466       &  2.575        \\
                & Other                         & -0.363      &  0.180        & -0.111      &  0.170        & -0.195      &  0.198        \\
\multicolumn{2}{l}{Mother's education (High)}   &             &               &             &               &             &               \\
                & \textless High                & -0.215      &  0.083        & -0.206      &  0.337        & -0.283      &  0.355        \\
                & \textgreater High             & 0.168       &  0.113        & -0.043      &  0.139        & -0.051      &  0.155        \\
                & Missing                       & 0.240       &  0.165        & -0.041      &  0.280        & -0.034      &  0.303        \\
\multicolumn{2}{l}{Mother's job (Stay-at-home)} &             &               &             &               &             &               \\
                & Professional                  & -0.239      &  0.111        & 0.182       &  0.142        & 0.186       &  0.158        \\
                & Other                         & -0.101      &  0.072        & 0.126       &  0.183        & 0.103       &  0.198        \\
                & Missing                       & -0.199      &  0.162        & 0.247       &  0.381        & 0.168       &  0.396        \\
\multicolumn{2}{l}{Age}                         & 0.075       &  0.033        & 0.110       &  0.030        & 0.103       &  0.029        \\\midrule
\multicolumn{8}{l}{\textbf{Network formation model}}                                                                                      \\
\multicolumn{2}{l}{Same sex}                    &             &               & 0.309       &  0.016        & 0.370       &  0.015        \\
\multicolumn{2}{l}{Both Hispanic}               &             &               & 0.417       &  0.027        & 0.433       &  0.025        \\
\multicolumn{2}{l}{Both White}                  &             &               & 0.312       &  0.025        & 0.304       &  0.023        \\
\multicolumn{2}{l}{Both Black}                  &             &               & 1.077       &  0.043        & 1.171       &  0.041        \\
\multicolumn{2}{l}{Both Asian}                  &             &               & 0.165       &  0.050        & 0.142       &  0.047        \\
\multicolumn{2}{l}{Mother's education $<$ High}  &             &               & 0.226       &  0.018        & 0.216       &  0.017        \\
\multicolumn{2}{l}{Mother's education $>$ High}  &             &               & 0.009       &  0.016        & 0.006       &  0.015        \\
\multicolumn{2}{l}{Mother's job Professional}    &             &               & -0.116      &  0.017        & -0.128      &  0.016        \\
\multicolumn{2}{l}{Age absolute difference}     &             &               & -0.701      &  0.010        & -0.715      &  0.009        \\
\multicolumn{2}{l}{}                            &             &               &             &               &             &               \\
\multicolumn{2}{l}{Average number of friends}   & 3.251       &               & 4.664       &               & 5.613       &     \\
\bottomrule

\multicolumn{8}{l}{%
  \begin{minipage}{16cm}%
  \vspace{0.1cm}
    \footnotesize{Note: Model 4 considers the observed network as given. Model 5 infers the missing links due to friendship nominations coded with error, and Model 6 infers the missing links due to friendship nominations coded with error and controls for top coding.
    For each model, Column "Mean" indicates the estimates, and Column "Std" indicates the posterior standard deviations in parentheses.\\
    \textit{N} = 3,126. Observed links = 17,993. 
    Proportion of inferred network data: error code = 60.0\%, error code and top coding = 65.0\%. The explained variable is computed by taking the average grade for English, Mathematics, History, and Science, letting $A=4$, $B=3$, $C=2$, and $D=1$. Thus, higher scores indicate better academic achievement.}
  \end{minipage}%
}\\
\end{tabular} 
\end{adjustbox}

\end{table}

\begin{table}[htbp]
\small
\centering
\caption{Empirical results (Standard IV estimator based on the sample of students without missing network data)}
\label{tab:appendapplicationIV} 
\begin{adjustbox}{max width=0.8\textwidth}
\begin{tabular}{lp{6cm}d{4}d{4}}
\\[-1.8ex]
\toprule
\multicolumn{2}{l}{Statistic}                     & \multicolumn{1}{c}{Mean}        & \multicolumn{1}{c}{Std}\\ \midrule
\multicolumn{2}{l}{Peer effects}                & 0.408          & (0.243)          \\
\multicolumn{2}{l}{\textbf{Own effect}}                  &                &                  \\
\multicolumn{2}{l}{Female}                      & 0.178          & (0.075)          \\
\multicolumn{2}{l}{Hispanic}                    & -0.068         & (0.099)          \\
\multicolumn{2}{l}{Race (white)}                &                &                  \\
                & Black                         & -0.170         & (0.084)          \\
                & Asian                         & -0.239         & (0.212)          \\
                & Other                         & 0.101          & (0.131)          \\
\multicolumn{2}{l}{Mother’s education (High)}   &                &                  \\
                & \textless high                & -0.041         & (0.104)          \\
                & \textgreater high             & 0.230          & (0.089)          \\
                & missing                       & 0.171          & (0.122)          \\
\multicolumn{2}{l}{Mother’s job (Stay-at-home)} &                &                  \\
                & Professional                  & 0.060          & (0.121)          \\
                & Other                         & -0.098         & (0.092)          \\
                & missing                       & -0.116         & (0.116)          \\
\multicolumn{2}{l}{Age}                         & -0.011         & (0.028)          \\
\multicolumn{2}{l}{\textbf{Contextual effects}}          &                &                  \\
\multicolumn{2}{l}{Female}                      & -0.063         & (0.153)          \\
\multicolumn{2}{l}{Hispanic}                    & 0.421          & (0.246)          \\
\multicolumn{2}{l}{Race (white)}                &                &                  \\
                & Black                         & 0.010          & (0.169)          \\
                & Asian                         & -0.242         & (0.909)          \\
                & Other                         & -0.204         & (0.304)          \\
\multicolumn{2}{l}{Mother’s education (High)}   &                &                  \\
                & \textless high                & -0.292         & (0.257)          \\
                & \textgreater high             & -0.076         & (0.196)          \\
                & missing                       & 0.077          & (0.295)          \\
\multicolumn{2}{l}{Mother’s job (Stay-at-home)} &                &                  \\
                & Professional                  & 0.168          & (0.247)          \\
                & Other                         & -0.141         & (0.191)          \\
                & missing                       & -0.126         & (0.283)          \\
\multicolumn{2}{l}{Age}                         & -0.065         & (0.05)          \\
\bottomrule

\multicolumn{4}{l}{%
  \begin{minipage}{10cm}%
  \vspace{0.1cm}
    \footnotesize{Note: These results are based on the standard IV estimation considering the sample of students without missing network data and who are not affected by the censoring issue ($N = 561$). Column "Mean" indicates the estimates, and Column "Std" indicates the posterior standard deviations in parentheses. The explained variable is computed by taking the average grade for English, Mathematics, History, and Science, letting $A=4$, $B=3$, $C=2$, and $D=1$. Thus, higher scores indicate better academic achievement.}
  \end{minipage}%
}\\
\end{tabular} 
\end{adjustbox}
\end{table} 

\clearpage

\subsection{Key Player}\label{sec:keyplayer}

Figure \ref{fig:centrality} shows a scatter plot of the vector of centralities in the observed and reconstructed networks. The figure illustrates the effects of missing network data. First, because the reconstructed network has more links, centrality is higher on average. This is essentially the social multiplier effect. Not accounting for missing links leads to an underestimation of spillover effects.
Second, some individuals, in particular those having very few links in the observed network, are in reality highly central. Therefore, targeting a policy at individuals having a high centrality in the observed network would be inefficient. In particular, Figure \ref{fig:centrality2} shows that even isolated individuals and individuals interacting in isolated pairs in the observed network (having centralities of $1$ and $1.35$ respectively) can be, in reality, highly central. Thus, a policy based on the evaluation of an observed network, coupled with the associated endogenous peer effect coefficient $\alpha$, would not only underestimate the social multiplier but would also target the wrong individuals.

\begin{figure}[htbp]
    \centering
    \includegraphics[scale=0.7]{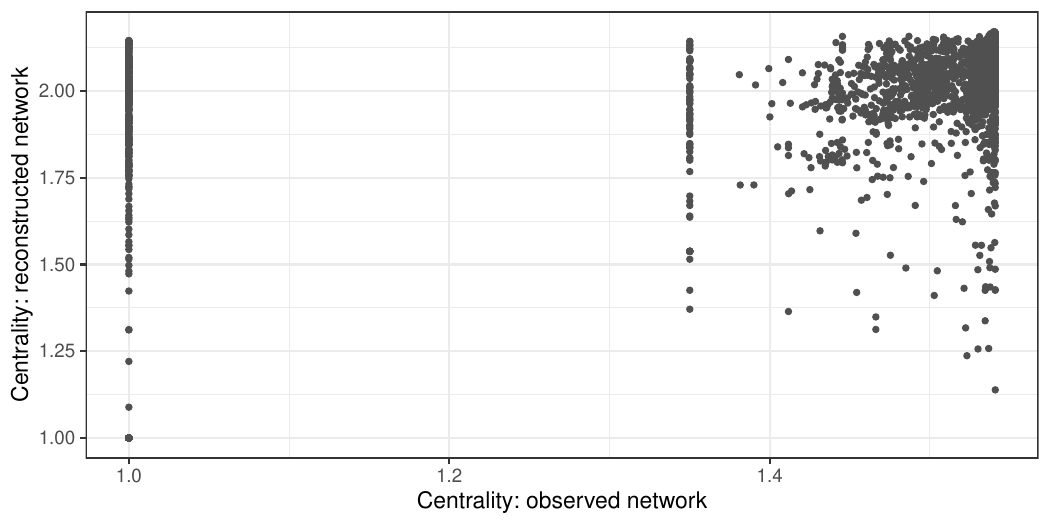}
    \caption{Centrality}
    \label{fig:centrality}
    
      \begin{minipage}{14cm}%
  \vspace{0.3cm}
    \footnotesize{Note: The centrality vector is given by $(\mathbf{I}-\hat{\alpha}\mathbf{G})^{-1}\mathbf{1}$. To compute centrality based on the observed network, we use the observed network $\mathbf{G}$ and the $\hat{\alpha}$ estimated using specification \emph{Obsv.Bayes}. To compute centrality based on the reconstructed network, we use $\hat{\alpha}$ and $\mathbf{G}$ estimated using the specification \emph{TopMiss.Bayes}. For both centrality vectors, we use the average vector centrality across 10,000 draws from their respective posterior distributions.}
  \end{minipage}%
\end{figure}

\begin{figure}[htbp]
    \centering
    \includegraphics[scale=0.7]{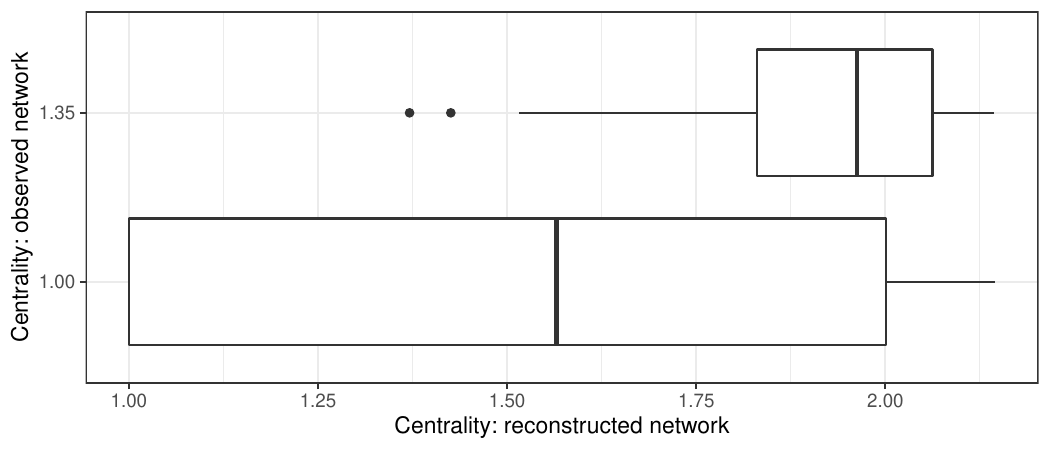}
    \caption{Centrality}
    \label{fig:centrality2}
    \begin{minipage}{14cm}%
  \vspace{0.3cm}
    \footnotesize{Note: See note of Figure \ref{fig:centrality}.}
  \end{minipage}%
\end{figure}

%% file: appendix_ARD.tex
\section{Technical Appendix: Aggregated Relational Data}\label{sec:ardsetting}

This section provides details about ARD simulation and model estimation using a MCMC method. We simulate the network for a population of 5000 individuals divided into $m = 20$ groups of $n=250$ individuals. Within each group, the probability of a link is
\begin{equation}\label{eq:app:ard}
P(a_{ij}=1)\propto\exp\{\nu_i+\nu_j+\zeta \mathbf{z}_i'\mathbf{z}_j\}.   
\end{equation}
As there is no connection between the groups, the networks are simulated and estimated independently. We first present how we simulate the data following the model (\ref{eq:ard}).

\subsection{ARD Simulation}
The parameters are defined as follows: $\zeta = 1.5 $, $~\nu_i \sim \mathcal{N}(-1.25, 0.37)$, and the $\mathbf{z}_i$ are distributed uniformly according to a von Mises--Fisher distribution. We use a hypersphere of dimension 3. We set the same values for the parameter for the 20 groups. We generate the probabilities of links in each network following \citeOA{breza2017using}.
\begin{equation}\label{eq:link:app:prob}
P(a_{ij} = 1|\nu_i,\nu_j,\zeta,\mathbf{z}_i,\mathbf{z}_j) =  \frac{\exp\{\nu_i+\nu_j+\zeta \mathbf{z}_i'\mathbf{z}_j\} \sum_{i=1}^N d_i}{\sum_{ij}\exp\{\nu_i+\nu_j+\zeta \mathbf{z}_i'\mathbf{z}_j\}},    
\end{equation}
where $d_i$ is the degree defined by $d_i \approx \frac{C_p(0)}{C_p(\zeta)}\exp{(\nu_i)}\sum_{i = 1}^N\exp(\nu_i)$, and the function $C_p(.)$ is the normalization constant in the von Mises--Fisher distribution density function. After computing the probability of a link for any pair in the population, we sample the entries of the adjacency matrix using a Bernoulli distribution with probability (\ref{eq:link:app:prob}).

To generate the ARD, we require the ``traits'' (e.g., cities) for each individual. We set $K = 12$ traits on the hypersphere. Their location $\mathbf{v}_k$ is distributed uniformly according to the von Mises-Fisher distribution. The individuals having the trait $k$ are assumed to be generated by a von Mises--Fisher distribution with the location parameter $\mathbf{v}_k$ and the intensity parameter $\eta_k \sim  |\mathcal{N}(4, 1)|$, $k = 1, \dots, 12$.

We attribute traits to individuals given their spherical coordinates. We first define $N_k$, the number of individuals having the trait $k$:
$$
N_k = \Bigg\lfloor r_{k} \dfrac{\sum_{i = 1}^{N} f_{\mathcal{M}}(\mathbf{z}_i| \mathbf{v}_k, \eta_k)}{\max_i f_{\mathcal{M}}(\mathbf{z}_i| \mathbf{v}_k, \eta_k)} \Bigg\rfloor,$$  
where $\lfloor x\rfloor$ represents the greatest integer less than or equal to $x$, $r_{k}$ is a random number uniformly distributed over $(0.8; 0.95)$, and $f_{\mathcal{M}}(\mathbf{z}_i| \mathbf{v}_k, \eta_k)$ is the von Mises--Fisher distribution density function evaluated at $\mathbf{z}_i$ with the location parameter $\mathbf{v}_k$ and the intensity parameter $\eta_k$.

The intuition behind this definition for $N_k$ is that when many $\mathbf{z}_i$ are close to $\mathbf{v}_k$, many individuals should have the trait $k$.

We can finally attribute trait $k$ to individual $i$ by sampling a Bernoulli distribution with the probability $f_{ik}$ given by 
$$f_{ik} = N_k \dfrac{f_{\mathcal{M}}(\mathbf{z}_i| \mathbf{v}_k, \eta_k)}{\sum_{i = 1}^{N} f_{\mathcal{M}}(\mathbf{z}_i| \mathbf{v}_k, \eta_k)}.
$$
The probability of having a trait depends on the proximity of the individuals to the trait's location on the hypersphere.

\subsection{Model Estimation}
In practice, we only have the ARD and the traits of each individual. \citeOA{mccormick2015latent} propose an MCMC approach to infer the parameters in the model (\ref{eq:app:ard}).

However, the spherical coordinates and the degrees in this model are not identified. The authors solve this issue by fixing some $\mathbf{v}_k$ and use the fixed positions to rotate the latent surface back to a common orientation at each iteration of the MCMC using a Procrustes transformation. In addition, the total size of a subset $b_k$ is constrained in the MCMC. 

As discussed by \citeOA{mccormick2015latent}, the number of $\mathbf{v}_k$ and $b_k$ to be set as fixed depends on the dimensions of the hypersphere. In our simulations, $\mathbf{v}_1$, $\mathbf{v}_2$, \dots, $\mathbf{v}_5$ are set as fixed to rotate back the latent space. When simulating the data, we let $\mathbf{v}_1 = (1, 0, 0)$, $\mathbf{v}_2 = (0, 1, 0)$, and $\mathbf{v}_3 = (0, 0, 1)$. This ensures that the fixed positions on the hypersphere are spaced, as suggested by the authors, to use as much of the space as possible, maximizing the distance between the estimated positions. We also constrain $b_3$ to its true value. The results do not change when we constrain a larger set of $b_k$

Following \citeOA{breza2017using}, we estimate the link probabilities using the parameters' posterior distributions. The gregariousness parameters are computed from the degrees $d_i$ and the parameter $\zeta$ using the following equation:
$$ \nu_i = \log(d_i) - \log\Big(\sum_{i = 1}^N d_i\Big) + \frac{1}{2}\log\Big(\frac{C_p(\zeta)}{C_p(0)}\Big).$$

%% file: appendix_simulations.tex
\subsection{Finite Sample Performance Using ARD}\label{sec:ard}
In this section, we study the small sample performance of the estimator presented in Section \ref{sec:iv} when the researcher only has access to ARD. First, we simulate network data using the model proposed by \citeOA{breza2017using} and simulate outcomes using the linear-in-means model (\ref{eq:linearinmean}) conditional on the simulated networks. Second, we estimate the network formation model using the Bayesian estimator proposed by \citeOA{breza2017using} (yielding $\hat{\boldsymbol{\rho}}_B$) and using the classical estimator proposed by \citeOA{alidaee2020recovering} (yielding $\hat{\boldsymbol{\rho}}_A$). Third, we estimate the linear-in-means model using the estimators presented in Proposition \ref{prop:bias_nocontext} and Theorem \ref{prop:non_observed} based on $\hat{\boldsymbol{\rho}}_A$ and $\hat{\boldsymbol{\rho}}_B$.

Recall that
\begin{equation}\label{eq:ard}
P(a_{ij}=1)\propto\exp\{\nu_i+\nu_j+\zeta \mathbf{z}_i'\mathbf{z}_j\},    
\end{equation}
where $\nu_i$, $\nu_j$, $\zeta$, $\mathbf{z}_i$, and $\mathbf{z}_j$ are not observed by the econometrician but follow parametric distributions. We refer the interested reader to \citeOA{mccormick2015latent}, \citeOA{breza2017using}, and \citeOA{breza2019consistently} for a formal discussion of the model, including its identification and consistent estimation. 



To study the finite sample performance of our instrumental strategy in this context, we simulate 20 groups, each having 250 individuals. Within each subpopulation, we simulate the ARD responses and a series of observable characteristics. The details of the Monte Carlo simulations can be found below in the Online Appendix \ref{sec:ardsetting}.

Importantly, the model in (\ref{eq:ard}) is based on a single population framework. Thus, the network formation model must be estimated separately for each of the 20 groups. With only 250 individuals in each group, we therefore expect significant small-sample bias.

We contrast the estimator proposed by \citeOA{breza2017using} with that of \citeOA{alidaee2020recovering}. Whereas \citeOA{breza2017using} present a parametric Bayesian estimator, \citeOA{alidaee2020recovering} propose a (nonparametric) penalized regression based on a low-rank assumption. One main advantage of the estimator proposed in \citeOA{alidaee2020recovering} is that it allows for a wider class of model and ensures that the estimation is fast and easily implementable.\footnote{The authors developed user-friendly packages in R and Python. See \citeOA{alidaee2020recovering} for links and details.} Note, however, that their method only yields a consistent estimator of $\hat{P}(\mathbf{A})$ if the true network is effectively low rank.

Very intuitively, the low-rank assumption implies that linking probabilities were generated from a small number of parameters. 
Importantly, the model (\ref{eq:ard}) is not necessarily low rank; for example, if the individuals' latent positions (i.e., the $\mathbf{z}_i$'s) are uniformly distributed, then the model may not be low rank and the method proposed by \citeOA{alidaee2020recovering} would perform poorly. If, however, individuals' latent positions are located around a few focal points, then the model might be low-rank because knowledge of these focal points may have high predictive power.

We compare the performance of both estimators as we vary the concentration parameter (that is, $\kappa$; see below in the Online Appendix \ref{sec:ardsetting} for details). This has the effect of changing the \emph{effective rank} of the linking probabilities: increasing $\kappa$ decreases the effective rank.\footnote{We refer the interested reader to \citeOA{alidaee2020recovering} for a formal discussion of the effective rank and its importance for their estimator.} We therefore expect the estimator proposed by \citeOA{alidaee2020recovering} to perform better for larger values of $\kappa$.

The results are presented in Tables \ref{tab:gx:obs:full} and \ref{tab:gx:unobs:full}. Table \ref{tab:gx:obs:full} presents the results for the special case where $\mathbf{GX}$ are observed in the data. The table displays the performance of our simulated GMM (see Corollary \ref{prop:GXobs}) when the network formation model is estimated by \citeOA{breza2017using} and \citeOA{alidaee2020recovering}. 

When $\kappa=0$, the network formation is not low rank. This disproportionately affects the estimator of \citeOA{alidaee2020recovering}. When $\kappa=15$, the estimators proposed by \citeOA{breza2017using} and \citeOA{alidaee2020recovering} perform similarly.

We now turn to the more general case where $\mathbf{GX}$ are not observed. Table \ref{tab:gx:unobs:full} presents the performance of our SGMM estimator (Theorem \ref{prop:non_observed}) when the network formation process is estimated using the estimators proposed by \citeOA{breza2017using} and \citeOA{alidaee2020recovering} and when we assume that the researcher knows the true distribution of the network.

We see that the performance of our estimator is strongly affected by the quality of the first-stage network formation estimator. When based on either the estimator proposed by \citeOA{breza2017using} or \citeOA{alidaee2020recovering}, for $\kappa=0$ or $\kappa=15$, our SGMM estimator performs poorly.

The poor performance of our SGMM estimator in a context where both $\mathbf{Gy}$ and $\mathbf{GX}$ are unobserved was anticipated. This occurs for two main reasons. \emph{First}, the consistency of the network formation estimator in \citeOA{breza2019consistently} holds as the size of each subpopulation goes to infinity, whereas the consistency of our estimator holds as the number of (bounded) subpopulations goes to infinity. This should affect the performance of our estimator, when based on \emph{estimated} network formation models but not when based on the true distribution of the network.

\emph{Second}, as discussed in \ref{sec:ard_netext}, ARD provides very little information about the realized network structure in the data (as opposed to censoring issues, for example; see Example \ref{ex:censored}). Then, if the true distribution is vague in the sense that most predicted probabilities are away from $0$ or $1$, we expect the estimation to be imprecise. This is what happens when $\kappa=15$, where our estimation based on the true distribution of the network is very imprecise in a context where the network affects the outcome through both $\mathbf{Gy}$ and $\mathbf{GX}$.

In the next section, we present a likelihood-based estimator, which uses more information on the data-generating process of the outcome to improve the precision of the estimation.

\renewcommand{\arraystretch}{.95}
\begin{table}[!htbp]
\centering 
\small
\caption{Simulation results with ARD and observed $\bG\bX$} 
\label{tab:gx:obs:full}
\begin{threeparttable}
\begin{tabular}{@{\extracolsep{5pt}}ld{4}ld{4}c}
\\[-1.8ex]\hline 
\hline \\[-1.8ex] 
Parameter              & \multicolumn{2}{c}{Breza et al.} & \multicolumn{2}{c}{Alidaee et al.} \\
                       & \multicolumn{1}{c}{Mean}           & \multicolumn{1}{c}{Std}        & \multicolumn{1}{c}{Mean}            & \multicolumn{1}{c}{Std}         \\
\hline\\[-1.8ex] 

                       & \multicolumn{4}{c}{SMM, $\kappa = 0$, $N = 250$, $M = 20$}                              \\
$\alpha = 0.4$         & 0.392              & (0.01)               & 0.492               & (0.057)               \\
$\beta_1 = 1$          & 1.001              & (0.004)              & 1.002               & (0.009)               \\
$\beta_2 = 1.5$        & 1.500              & (0.007)              & 1.496               & (0.016)               \\
$\gamma_1 = 5$         & 5.013              & (0.034)              & 3.884               & (0.295)               \\
$\gamma_2 = -3$        & -2.993             & (0.052)              & -4.048              & (0.354)               \\\hline\\[-1.8ex] 
                       & \multicolumn{4}{c}{SMM, $\kappa = 15$, $N = 250$, $M = 20$}                             \\
$\alpha = 0.4$         & 0.400              & (0.009)              & 0.428               & (0.009)               \\
$\beta_1 = 1$          & 1.000              & (0.004)              & 0.999               & (0.004)               \\
$\beta_2 = 1.5$        & 1.500              & (0.008)              & 1.499               & (0.008)               \\
$\gamma_1 = 5$         & 4.996              & (0.034)              & 4.677               & (0.034)               \\
$\gamma_2 = -3$        & -3.005             & (0.055)              & -3.387              & (0.055)     \\ \bottomrule     
\end{tabular}
\begin{tablenotes}[para,flushleft]
\footnotesize Note: In each subnetwork, the spherical coordinates of individuals are generated from a von Mises--Fisher distribution with a location parameter $(1, 0, 0)$ and intensity parameter $\kappa$. Predicted probabilities are computed using the mean of the posterior distribution. We chose the weight associated with the nuclear norm penalty to minimize the RMSE through cross-validation. This value of $\lambda=600$ is smaller than the recommended value in \citeOA{alidaee2020recovering}. Instruments are built using only second-degree peers, i.e., $\mathbf{G^2X}$. We set $R = S = T = 1$.
\end{tablenotes}
\end{threeparttable}
\end{table}

\begin{table}[!htbp]
\centering 
\small
\caption{Simulation results with ARD and unobserved $\bG\bX$} 
\label{tab:gx:unobs:full}
\begin{threeparttable}
\begin{tabular}{@{\extracolsep{5pt}}ld{4}ld{4}cd{4}c}
\\[-1.8ex]\hline 
\hline \\[-1.8ex] 
Parameter       & \multicolumn{2}{l}{Breza et al.} & \multicolumn{2}{l}{Alidaee et al.} & \multicolumn{2}{l}{True distribution} \\
                       & \multicolumn{1}{c}{Mean}           & \multicolumn{1}{c}{Std}        & \multicolumn{1}{c}{Mean}            & \multicolumn{1}{c}{Std}        & \multicolumn{1}{c}{Mean}            & \multicolumn{1}{c}{Std}       \\\hline\\[-1.8ex] 
                & \multicolumn{6}{c}{SMM, $\kappa = 0$, $N = 250$, $M = 20$}                                                    \\
$\alpha = 0.4$  & 0.717          & (0.463)         & 0.700           & (0.268)          & 0.400            & (0.056)            \\
$\beta_1 = 1$   & 0.988          & (0.022)         & 0.995           & (0.017)          & 1.000            & (0.015)            \\
$\beta_2 = 1.5$ & 1.505          & (0.03)          & 1.503           & (0.029)          & 1.501            & (0.021)            \\
$\gamma_1 = 5$  & 1.778          & (4.473)         & 1.512           & (2.37)           & 4.991            & (0.455)            \\
$\gamma_2 = -3$ & -2.205         & (1.24)          & -0.405          & (0.955)          & -3.005           & (0.287)            \\\hline\\[-1.8ex] 
                & \multicolumn{6}{c}{SMM, $\kappa = 15$, $N = 250$, $M = 20$}                                                   \\
$\alpha = 0.4$  & 0.603          & (0.069)         & 0.870           & (0.202)          & 0.434            & (0.394)            \\
$\beta_1 = 1$   & 0.989          & (0.014)         & 0.984           & (0.015)          & 0.998            & (0.021)            \\
$\beta_2 = 1.5$ & 1.504          & (0.029)         & 1.509           & (0.029)          & 1.501            & (0.023)            \\
$\gamma_1 = 5$  & 2.866          & (0.566)         & 0.246           & (1.973)          & 4.638            & (3.887)            \\
$\gamma_2 = -3$ & -2.458         & (0.379)         & -1.539          & (0.602)          & -2.913           & (1.037)    \\ \bottomrule     
\end{tabular}
\begin{tablenotes}[para,flushleft]
\footnotesize Note: see Table \ref{tab:gx:obs:full}, except that we fix $R = 500$ and $S = T = 1$.
\end{tablenotes}
\end{threeparttable}
\end{table}
\renewcommand{\arraystretch}{1}
\clearpage